\documentclass[aps,prx,twocolumn, notitlepage,longbibliography,superscriptaddress,nofootinbib]{revtex4-2}
\usepackage{amsmath}
\usepackage{amsthm}
\usepackage{thm-restate}
\usepackage{amsfonts}
\usepackage{amssymb}
\usepackage{dsfont}
\usepackage{graphicx}
\usepackage{physics}
\usepackage{tikz}
\usepackage{hyperref}
\usepackage{mathtools}
\usepackage{comment}
\usepackage{multirow}
\usepackage{ulem}
\usepackage{pst-node}
\usepackage{tikz-cd}
\usepackage[shortlabels]{enumitem}

\usepackage[english]{babel}

\newtheorem{theorem}{Theorem}
\newtheorem{corollary}{Corollary}[theorem]

\newtheorem{lemma}{Lemma}

\usepackage{pst-all}
\usepackage{leftidx}

\def\H{\mathcal{H}}
\def\M{\mathcal{M}}
\def\Z{\mathbb{Z}}

\def\L{\mathcal{L}}
\def\SS{\mathcal{S}}
\DeclareMathOperator*{\Motimes}{\text{\raisebox{0.25ex}{\scalebox{0.8}{$\bigotimes$}}}}

\newcommand{\RR}{\mathbb{R}}

\newcommand{\ZZ}{\mathbb{Z}}

\newcommand{\be}{\begin{equation}}
\newcommand{\ee}{\end{equation}}
\newcommand{\bc}{\begin{center}}
\newcommand{\ec}{\end{center}}
\newcommand{\nin}{\noindent}

\newcommand{\non}{\nonumber}
\newcommand{\lo}{\overline}

\definecolor{dualblue}{RGB}{3,101,192}

\usepackage[bottom]{footmisc}

\begin{document}

\title{Non-Clifford and parallelizable fault-tolerant logical gates on constant and almost-constant rate homological quantum LDPC codes via higher symmetries}


\author{Guanyu Zhu}
\email{guanyu.zhu@ibm.com}
\affiliation{IBM Quantum, T.J. Watson Research Center, Yorktown Heights, NY 10598 USA}

\author{Shehryar Sikander}
\affiliation{New High Energy Theory Center, Rutgers University, New Brunswick, NJ 08854-8019, USA}

\author{Elia Portnoy}
\affiliation{Mathematics Department, Massachusetts Institute of Technology, Cambridge, MA 02139, USA}

\author{Andrew W. Cross}
\affiliation{IBM Quantum, T.J. Watson Research Center, Yorktown Heights, NY 10598 USA}

\author{Benjamin J. Brown}
\affiliation{IBM Quantum, T.J. Watson Research Center, Yorktown Heights, NY 10598 USA}

\begin{abstract}
We study parallel fault-tolerant quantum computing for families of homological quantum low-density parity-check (LDPC) codes defined on 3-manifolds with constant or almost-constant encoding rate.
We derive generic formula for a transversal $T$ gate on color codes defined on general 3-manifolds, which acts as collective non-Clifford logical CCZ gates on any triplet of logical qubits with their logical-$X$ membranes having a $\mathbb{Z}_2$ triple intersection at a single point.  The triple intersection number is a topological invariant, which also arises in the path integral of the emergent higher symmetry operator in a topological quantum field theory (TQFT):  the $\mathbb{Z}_2^3$ gauge theory.   Moreover, the transversal $S$ gate of the color code corresponds to a higher-form symmetry in TQFT supported on a codimension-1 submanifold, giving rise to exponentially many addressable and parallelizable logical CZ gates.  A construction of constant-depth circuits of the above logical gates via cup product cohomology operation is also presented for three copies of identical  toric codes on arbitrary 3-manifolds.  We have developed a generic formalism to compute the triple intersection invariants for 3-manifolds, with the structure encoded into an  interaction hypergraph which determines the logical gate property and also corresponds to the hypergraph magic state that can be injected into the code without distillation (`\textit{magic state fountain}'). We also study the scaling of the Betti number and systoles with volume for various 3-manifolds, which translates to the encoding rate and distance. We further develop three types of LDPC codes supporting such logical gates: (1) A quasi-hyperbolic code from the product of 2D hyperbolic surface and a circle, with almost-constant rate $k/n=O(1/\log(n))$ and $O(\log(n))$ distance; (2) A homological fibre bundle code from twisting the product by an isometry of the surface based on the construction by Freedman-Meyer-Luo, with  $O(1/\log^{\frac{1}{2}}(n))$ rate and $O(\log^{\frac{1}{2}}(n))$ distance; (3) A specific family of 3D hyperbolic codes: the Torelli mapping torus code, constructed from mapping tori of a pseudo-Anosov element in the Torelli subgroup, which has constant rate while the distance scaling is currently unknown.   We then show a generic constant-overhead scheme for applying a parallelizable universal gate set with the aid of logical-$X$ measurements.

\end{abstract}

\maketitle

\tableofcontents


\section{Introduction}

Rapid progress in quantum low-density parity-check (qLDPC) codes has been made in recent years to achieve asymptotically optimal quantum information storage \cite{fiberbundlecode21,9567703,9490244,hastingswr21,pkldpc22,9996782,lh22,guefficient22,dhlv23,lzdecoding23,gusingleshot23}, including reaching constant encoding rate and linear code distance. On the other hand, resource-efficient fault-tolerant quantum computation should use  not only an optimal quantum memory, but also fast schemes to implement logical gates with low overhead. In addition to being fast, logic gates should also be applicable in parallel.  Although there exists a scheme with constant space overhead \cite{gottesman14}, it only allows the sequential application of logic gates and is therefore relatively slow in practice.  Therefore, parallel schemes for implementing universal fault-tolerant logic gates are especially desirable.

Indeed, the quest of parallel schemes is quite challenging  since qLDPC codes typically are composed of a single or constant number of code blocks in order for the encoding rate to be constant. This means that it can be difficult to address logical qubits individually. { For instance, 
conventional schemes using global or fold  transversal logical gates \cite{Kubica:2015br, Moussa:2016, Zhu:2017tr, breuckmann2022fold, Quintavalle:2022fold, Bravyi:2024wc} tend to induce collective logic gates to many logical qubits in a code block.  With the exception of topological codes that encode only a small number of logical qubits \cite{Kubica:2015br, Bombin:2015jk, Moussa:2016, Zhu:2017tr, bombin2018transversal}, few, or no examples are known for transversal gates can address individual logical qubits for codes that encode a large number of logical qubits in a single code block. } Measurement-based schemes with individual addressability based on logical measurements have been proposed, but these schemes use additional ancilla qubits with $O(d^2)$ overhead over a long time, where $d$ is the code distance \cite{cohen22}. A parallelizable scheme has recently been proposed using constant-rate concatenated codes~\cite{yamasaki2022timeefficient}. However, this scheme uses codes with high-weight stabilizers and therefore goes beyond the scope of qLDPC codes. It remains an open question how to do individually addressable and parallelizable logical gates with qLDPC codes, without introducing additional overhead. 

Another fundamental challenge is that most of the recently proposed qLDPC codes are based on a 3-term chain complex which can be considered as a generalization of 2D surface codes, and are hence only expected to perform logical Clifford gates.  In particular, a no-go theorem for logical non-Clifford gates has been given for the family of hypergraph-product codes achieving constant rate \cite{burtonbrowne21}.   Therefore, one has to perform a magic state distillation process \cite{bravyi2005} in order to achieve universality which will lead to significant resource overhead. It is hence desirable to acquire a native logical non-Clifford gate on qLDPC codes with $O(1)$ space-time overhead, such that one is able to directly implement these logical gates or efficiently inject magic states without extensive rounds of distillation.  Note that there exists a previous proposal for applying a parallelizable universal logical gate set (including non-Clifford gates) on constant-rate quantum code with constant-depth circuits corresponding to Dehn twists \cite{Lavasani2019universal} on 2D hyperbolic Fibonacci-Turaev-Viro codes \cite{levin2005, Koenig:2010do}.  However, these codes correspond to non-Abelian topological orders beyond the stabilizer models and hence need more sophisticated schemes for error correction and decoding \cite{schotte2020quantum, Schotte:2023Fault} which anticipates further optimization.  Instead, our current work focuses on fault-tolerant computation schemes with the CSS qLDPC codes belonging to the stabilizer models for which the error correction schemes are more well-developed.

In this paper, we develop a generic scheme for homological qLDPC codes defined on closed manifolds in three dimensions (abbreviated as 3-manifolds from now on).    Such codes are naturally associated with a 4-term chain complexes that evade the no-go theorem in Ref.~\cite{burtonbrowne21}.    We use a \textit{fattening} procedure by Bombin and Delgado in Ref.~\cite{Bombin2007} to convert any triangulation of a 3-manifold into a 4-valent and 4-colorable lattice (see also Ref.~\cite{Vuillot2022}).  Therefore, we are able to define a 3D color code on any 3-manifold with this procedure.  We then derive that a transversal T gate on an arbitrary 3-manifold color code implements a collective logical CCZ gates acting on any triplet of logical qubits if and only if their logical-$X$ membranes supported on 2-cycles of the manifold have $\ZZ_2$ triple intersection number 1 (mod 2).   In addition, a transversal $S$ gate in the color code supported on a codimension-1 (2D) submanifold (2-cycle) corresponds to Clifford logical CZ gate acting on two logical qubits of which the logical-$X$ membranes have a $\ZZ_2$ triple intersection with the submanifold where the transversal gate is supported.   Therefore, the logical gate structure is completely determined by the $\ZZ_2$ triple intersection numbers  (either 0 or 1), which are  topological invariants and hence independent of geometric details.   Note that since the transversal $S$ gate is only applied on a submanifold, instead of the entire 3-manifold, the corresponding logical CZ gates have individual addressability and can also be applied in parallel due to their mutual commutativity. 

We note that the number of addressable logical Clifford gates scales exponentially as $|H_2(\M^3;\ZZ_2)| = 2^{b_2(\M^3;\ZZ_2)}=2^k$, where $|H_2(\M^3;\ZZ_2)|$ is the size of the second $\ZZ_2$-homology group of the 3-manifold $\M^3$ and counts the total number of equivalent classes of non-contractible membranes (2-cycles) in $\M^3$.  The independent generator of this gate set is proportional to the second Betti number $b_2(\M^3;\ZZ_2)$ equaling the number of logical qubit $k$. This is in sharp contrast and a significant improvement compared to the situation of the fold-transversal gate acting on the entire system \cite{Kubica:2015br, Moussa:2016, Zhu:2017tr, breuckmann2022fold, Quintavalle:2022fold}, which has been recently applied to the context of qLDPC code \cite{breuckmann2022fold, Quintavalle:2022fold}. Since the fold-transversal gate always involves an isometry of the underlying lattice, the total number of addressable logical gates can be applied is propotional to the size of the isometry group.  In the context of the hyperbolic codes,  the size of the isometry group is at most $O(n)$, where $n$ is the number of qubits equaling the number of edges on the hyperbolic lattice.   Due to the constant rate,  one expects at most $O(n)=O(k)$ addressable logical gates which is much smaller than the exponential scaling of the size of the Clifford group and is hence not scalable. Similar scaling is expected for fold-transversal gates on other qLDPC codes.

\begin{figure}[t] \includegraphics[width=1\columnwidth]{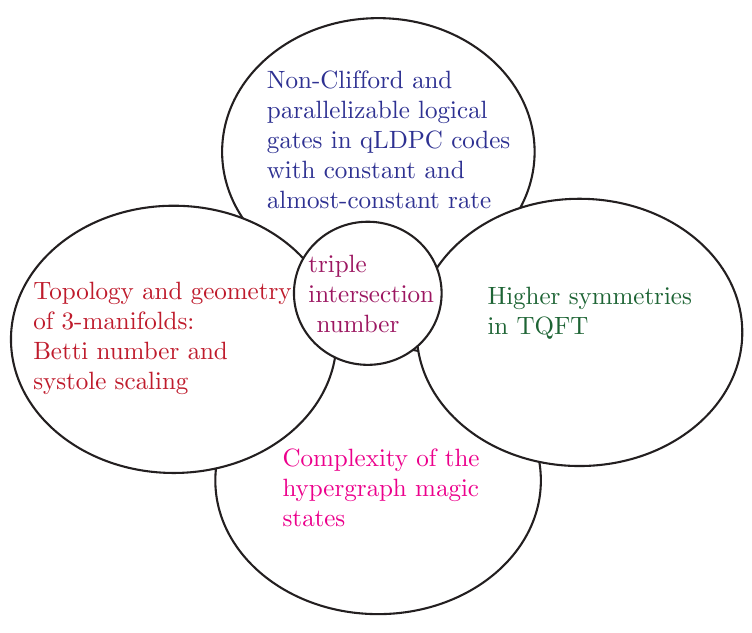}
  \caption{Summary of the connection between various subjects in this paper. }
  \label{fig:van-diagram}
\end{figure}

Since the logical gate structure is completely topological, we further explore its deep origin from (3+1)D topological quantum field theory (TQFT), more specifically a $\ZZ_2 \times \ZZ_2 \times \ZZ_2$ gauge theory equivalent to three copies of 3D toric codes following the formalism developed in Ref.~\cite{10.21468/SciPostPhys.14.4.065}.  We then see that the logical CCZ and CZ operators correspond to emergent \textit{higher symmetries}  \cite{gaiotto2014,benini2019,cordova2022,mcgreevy2022, bhardwaj2022universal, 10.21468/SciPostPhys.14.4.065} in TQFT, which has been a very active direction in the fields of high-energy and condensed matter physics in recent years. In a $(D+1)$-dimensional TQFT, one can think of applying a non-trivial symmetry operation on a $(D-k+1)$-dimensional submanifold  by sweeping a codimension-$k$ defect along some closed trajectory in time. A special class of defects that are invertible define a higher group symmetry. Invertible codimension-$k$ defects define a $(k-1)$-form symmetry \cite{gaiotto2014}.  In a $(D+1)$-dimensional TQFT, this leads to a series of higher groups $K_k$, for $k = 1,\dots, D$.  This higher group structure has a correspondence with the Clifford hierarchy in quantum information in various cases \cite{Bravyi:2013dx, Bombin_2015, Yoshida_global_symmetry_2016, Yoshida2017387, bombin2018transversal, 10.21468/SciPostPhys.14.4.065, barkeshli2022higher, Vasmer:2022morphing}.

In the (3+1)D case,  the invertible defects are expected to define a categorical 3-group symmetry \cite{Hsin2020liquid,kapustinThorngren2017FermionSPT}.  The invertible codimension-$k$ defects corresponding to logical CCZ and CZ operators \cite{Bombin:2015jk,   bombin2018transversal, bombin20182d, Vasmer2019} are so-called gauged symmetry-protected topological (SPT) defects \cite{Yoshida_gate_SPT_2015, Yoshida_global_symmetry_2016, Yoshida2017387,   WebsterBartlett18, 10.21468/SciPostPhys.14.4.065, barkeshli2022higher}, which can be obtained from gauging a trivial phase with a SPT phase decorated on a codimension-$(k-1)$ submanifold. Specifically, the logical CCZ operator is a $0$-form symmetry corresponding to sweeping a codimension-1 gauged $\ZZ_2 \times \ZZ_2 \times \ZZ_2$ SPT defect, while the logical CZ operator is a $1$-form symmetry corresponding to sweeping a codimension-2 gauged $\ZZ_2 \times \ZZ_2$ SPT defect.  When these symmetries are applied on a generic 3-manifold, we see that the triple intersection topological invariants   naturally arise in the TQFT path integral of these symmetry operators which completely determines the logical gate structure independent of the microscopic details of the lattice models, i.e., whether they are a single 3D color code or 3 copies of 3D toric codes or some other variants. While the 0-form symmetry acts on the entire system just like the usual transversal logical gate and hence induces a collective non-Clifford gate, the 1-form (higher-form) symmetry only acts on a codimension-1 submanifold, giving rise to individually addressable and parallelizable logical Clifford gates.  We note that although the idea of higher-group and higher-form symmetries in TQFT has been connected to quantum information in previous works \cite{Yoshida_gate_SPT_2015, Yoshida_global_symmetry_2016, Yoshida2017387} and has also been applied to fault-tolerant quantum computing with fractal topological codes \cite{PRXQuantum.3.030338}, this is,  as far as the authors know, the first work which finds significant advantage of applying higher symmetries in the context of  { high-rate} quantum LDPC codes.

{
Moreover, we use an operator-valued cochain formalism orignated from the TQFT to explicitly construct a constant-depth circuit composed of overlapping CCZ gates acting on three identical copies of toric codes defined on arbitrary triangulation on arbitrary 3-manifolds from a specific type of cohomology operation in algebraic topology: the triple cup product.   We then use the operator commutation relations, the Leibniz rule of cup product, and the cocycle condition to show that such a circuit is indeed a logical gate which preserves the code space.   We further show that this circuit induces a logical CCZ gate coupling a triple of logical qubits in three copies of identical toric codes  when the corresponding logical-X membranes have a non-trivial triple intersection.   Similarly, we derive the addressable logical CZ gate between two copies of toric codes.  More elaborated study and application of this formalism on a much wider class of logical gates are discussed in the follow-up papers \cite{zhu2025topological, Hsin2024:classifying}.
}

Although topology determines the logical gate structure, it is the geometry of the 3-manifold which determines code parameters in the corresponding homological qLDPC code.   More specifically, the 1st Betti number (equaling the 2nd Betti number in 3D due to Poincaré duality) determines the number of logical qubits, and its scaling with the volume determines the encoding rate. On the other hand, the 1-systole and 2-systole (the length of the shortest non-contractible 1-cycles and 2-cycles) determine the code distance.  While for 2D hyperbolic manifolds, the Gauss-Bonnet theorem gives rise to a linear scaling of the Betti number with the area (volume) and hence constant encoding rate, such constraint does not hold generically for hyperbolic manifolds in higher dimension.   Instead, there is only an upper bound for the Betti number / volume ratio due to Gromov's inequality \cite{gromov1982volume}.   We say 3-manifolds that saturate Gromov's upper bound while having systoles growing with volume reach the optimal quantum information storage per volume.  Finding such 3-manifolds or those close to them is of importance for fault-tolerant quantum computation with LDPC codes.   

We then construct various 3-manifolds and the corresponding homological qLDPC codes including their color code variants.   We also compute or bound their Betti number and systoles, and  develop a formalism to compute the triple intersection numbers.   In particular, the $\ZZ_2$ triple intersection invariants of a 3-manifold  determine its underlying interaction hypergraph and graph encoding the structure of the collective logical CCZ gates and parallelizable logical CZ gates respectively.   More specifically, each hyperedge of the interaction hypergraph connects to three vertices representing logical qubits of which the logical-$X$ membranes have a $\ZZ_2$ triple intersection.  Meanwhile, edges with the same color type in the colored interaction graph represents logical CZ gates which are simultaneously applied.     

The interaction hypergraph also has a one-to-one correspondence with the hypergraph magic state \cite{Rossi_2013, chen2023magic, Takeuchi:2019_hypergraph} injected to the code, which can be considered as generalization of graph states.  Such logical states are valuable computational resources for magic state distillation and injection protocol either in the context of gate-based or measurement-based computation scheme at the logical level  \cite{bravyi1998, chen2023magic, Takeuchi:2019_hypergraph}.  In particular, the transversal T gate applied on these homological LDPC codes provides a high-fidelity magic state injection with the fidelity directly corresponding to the code distance.  We also characterize the complexity of such hypergraph magic states by the number of logical CCZ gates acted on the code and show that it is proportional to the total number of $\ZZ_2$ triple intersection points of the underlying 3-manifold.  We hence observe a deep connection between the complexity of quantum states and the topology/geometry of manifolds.  

In this paper, we have constructed three different families of 3-manifolds and homological qLDPC codes.   The first type is the quasi-hyperbolic code defined on a product of a 2D arithmetic hyperbolic genus-$g$ surface $\Sigma_g$ \cite{Guth:2014cj} with a circle $S^1$, i.e., $\Sigma_g \times S^1$.  The circle is chosen to have $O(\log V)$ length where $V$ is the  volume.   Due to the Gauss-Bonnet theorem for the 2D hyperbolic surface, one can obtain the 1st Betti number - volume scaling as $b_1 / V = O(\log V)$, equivalent to an encoding rate $k/n= O(1/\log n)$, while the systoles have a lower bound $\Omega(\log V)$ corresponding to a distance lower bound $d =\Omega(\log n)$.   The $O(\log n)$ qubit overhead comes from the length of the circle.  There are $g=O(k)$  triple intersection points. The interaction hypergraph and graph in this case can be completely determined.  

As an attempt to resolve the $\log(n)$ overhead, we consider the second type of codes based on a fibre bundle construction (also a mapping torus in this special case) where we choose the genus-$g$ hyperbolic surface as the fibre and a circle as the base, following the construction by Freedman-Meyer-Luo \cite{Freedman_systole_2002} \footnote{It is based on the 3-manifold they constructed before doing the Dehn surgery to kill most of the homology cycles.}. We call these codes homological fibre bundle codes. In order to avoid space overhead, we choose the base circle to be unit length.  In order to avoid reducing the 1-systole to a constant, we apply a twist on the fibre surface, which is an isometry with order $O(\log^\frac{1}{2}(g)) = O(\log^\frac{1}{2}(V))$, such that a 1-cycle has to travel $O(\log^\frac{1}{2}(V))$ times through the base to come back to itself.   We can then prove a lower bound on both the 1-systole and 2-systole, being $\Omega(\log^\frac{1}{2}(V))$ and $\Omega(\log(V))$ respectively, which leads to the code distance being  $d=\Omega(\log^\frac{1}{2}(n))$.  Nevertheless, the twist also kills some homologies and reduces the 1st Betti number from $O(g)$ to $O(g/\log^\frac{1}{2}(g))=O(V/\log^\frac{1}{2}(V))$, namely with a square root logarithmic reduction.  This leads to the encoding rate reducing to $k/n = O(1/\log^\frac{1}{2}(n))$.   There exist  $O(g/\log^\frac{1}{2}(g))=O(k/\log^\frac{1}{2}(k))$ triple intersection points and the corresponding interaction hypergraph is isomorphic to that of the quasi-hyperbolic code up to a reduction in size by $O(1/\log^\frac{1}{2}(k))$.  Although also being almost-constant rate, we think this code could be more useful since its lattice can be embedded in constant number layers connected via twisted vertical links without any intra-layer crossing.  

In order to avoid the problem of Betti number reduction, we consider a more general mapping torus construction such that the twist is a general mapping class instead of a finite-order isometry.  In particular, we obtain a family of 3D hyperbolic codes based on mapping tori of a pseudo-Anosov element in the Torelli subgroup, following the 3-manifold construction in Ref.~\cite{AgolLeiningerMargalit}, which we name the Torelli mapping-torus code.    The underlying 3-manifold has the 1st Betti number scaling linear with the volume \cite{KojimaMcShane} and hence saturates the Gromov bound, which leads to a  constant encoding rate.   However, the scaling of the systoles and code distance remain unknown.  We conjecture that the systoles/distance should also grow with the volume due to the twist of the fibre surface, in analogy with the case of the homological fibre-bundle code mentioned above.  There exist $g=O(k)$ triple intersection points and the corresponding interaction graph is isomorphic to that of the quasi-hyperbolic codes.  

We then provide generic schemes for performing parallel fault-tolerant universal computation with  homological qLDPC codes based on 3-manifolds and demonstrate the scheme using the interaction hypergraph corresponding to all the three codes mentioned above. In addition to the logical CZ and CCZ gates, we also implement parallel measurement of the logical-$X$ operator based on lattice surgery \cite{lavasani2018} to perform parallel logical Hadamard, CNOT and SWAP gates. With the logical SWAP, one can shuffle an arbitrary triplet of logical qubits to an ancilla qLDPC code copy in order to perform addressable logical CCZ gates.     The combination of logical CCZ gates and logical Hadamard gates can form a universal logical gate set.   In the specific case of the quasi-hyperbolic code, we further enhance the parallelizability by thickened Dehn twists corresponding to logical CNOT gates. 

Finally, we notice that the parallelizability of the logical gates is limited by the star-like interaction hypergraph of the three types of codes constructed in this work, while constructing a code and equivalently a 3-manifold corresponding to a sparse interaction hypergraph can significantly improve the parallelizability.   We hence provide a general recipe to back-engineer a 3-manifold with a desired interaction hypergraph following Sullivan's construction \cite{Sullivan},  while the geometry and code parameters anticipate further optimization in future.  
The connections of various subjects are summarized in Fig.~\ref{fig:van-diagram}.

The structure of the paper is the following:  Sec.~\ref{sec:non-Clifford} shows the construction of color codes and surface codes on arbitrary 3-manifolds and derive the transversal T and S gates as the collective logical CCZ and parallelizable logical CZ gates respectively, as well as the connection to the triple-intersection number.  One also introduces the interaction hypergraph which encodes the logical gate structure and the connection to hypergraph magic state injection. Sec.~\ref{sec:TQFT} connects the logical gates to emergent $k$-form symmetries in the topological quantum field theory and reveals the origin of the triple intersection invariant from the path integral expression using an operator-valued cochain formalism. Mathematically, this just corresponds to cohomology operation in algebraic topology. We also use this formalism to explicitly construct a constant-depth circuit composed of CCZ or CZ gates in three identical copies of toric codes and show that they lead to collective logical CCZ gates and addressable logical CZ gates.  Sec.~\ref{sec:scaling} connects the scaling of the code parameters of the homological qLDPC codes such as the encoding rate and codes distance to the geometry of the underlying manifolds, i.e., the Betti number and systole scaling with the volume. Sec.~\ref{sec:code_construction} gives the explicit construction of the three types of homological qLDPC codes with constant or almost-constant encoding rate.  Sec~\ref{sec:application} shows how to perform parallel universal fault-tolerant quantum computation with the non-Clifford and parallelizable logical gates acting on the homological qLDPC codes constructed in the previous sections.  Additional logical operations including logical measurements and thickened Dehn twists are needed to achieve the universality and enhance the parallelizability.    The paper is concluded by Sec.~\ref{sec:outlook} with the discussion of the main results and outlook of future directions.     { Finally, we also discuss how to back-engineer the 3-manifolds from desired logical gate structure in  Appendix \ref{sec:back-engineering}, where the discussion focuses on the topology rather than the geometry of the 3-manifolds.}

\section{Non-Clifford and parallelizable  logical gates of color codes  on 3D manifolds}\label{sec:non-Clifford}

{
In this section,  we develop the theory of non-Clifford and parallelizable logical gates in the context of 3D color codes defined on 3-manifolds, while we reserve the equivalent description of three copies of 3D toric codes to Sec.~\ref{sec:TQFT}.  

We first introduce the construction of color codes and toric codes on general 3-manifolds in Sec.~\ref{sec:color_and_toric}, and connect the code space to the homology group and Betti numbers.   We then derive a generic formula for the collective logical CCZ gates implemented by the transversal $T$ gate and its connection to the triple-intersection invariant in the underlying 3-manifold in Sec.~\ref{sec:non-Clifford_triple}.   Next, we derive the generic formula for addressable and parallelizable logical CZ gates implemented by transversal S gates supported on a submanifold in Sec.~\ref{sec:parallelizable}, where we have also provided the formula connecting the total number addressable logical CZ gates to the 1st Betti number in the underlying 3-manifold.  Finally, in Sec.~\ref{sec:interaction_hypergraph},  we introduce the interaction hypergraph to help understanding the structure of the logical non-Clifford gates.  We further introduce the single-shot injection protocol of hypergraph magic states, which can also be represented by the interaction hypergraph.   
}

\subsection{Color codes and toric codes on 3D manifold}
\label{sec:color_and_toric}
We can define a $D$-dimensional color code on some arbitrary $D$-manifold $\mathcal{M}^D$~\cite{Bombin_2015} that is equivalent to $D$ copies of the $D$-dimensional toric-code model on the same manifold up to a constant-depth local circuit~\cite{Kubica:2015br}.  Let us review { how to obtain} the color-code model from an arbitrary cellulation of $\mathcal{M}^D$.

First, we define a color code \cite{Bombin:2007eh}. We follow the definition due to Ref.~\cite{Kubica2015} where we are only concerned with closed manifolds with no boundaries. We will state the model for general $D$-dimensional manifolds but we note that we are mainly concerned with $D=3$.

Here it will be convenient to define the color code on its dual lattice with a triangulation  composed of tetrahedra $\Delta$ where a tetrahedron is specified by a list of $j+1$ vertices lying on the extremities of its convex hull, i.e., $\Delta = \{ v_1, v_2,\dots, v_j \}$. The lattice is $D+1$-colourable such that no two vertices of the same colour share an edge. We show how to construct such a lattice shortly, but for the present definition we will assume that such a lattice can exist. Qubits are assigned to the $D$-cells of the $D+1$-colourable lattice, and Pauli-X (Pauli-Z)-type stabilizers are associated to the 0- ($(D-2)$-) cells of the color-code lattice. Let us denote the 0-cells ($(D-2)$-cells) as $v$ ($\delta$), such that we have stabilizers
\begin{equation}
S^X_v = \prod_ {\Delta \supset v} X_\Delta, 
\quad
S^Z_\delta = \prod_{\Delta \supset \delta} Z_\Delta, 
\end{equation}
where $X_\Delta$ ($Z_\Delta$) denotes a Pauli-X (Pauli-Z) operator acting on the qubit at $D$-cell $\Delta$, and we have used the subscript notation $\Delta \supset v$ ($\Delta \supset \delta$) to denote the product over all qubits $\Delta = \{u_1, u_2,\dots, u_{D+1}\}$ that include all the vertices of simplex $v = \{v\}$ ($\delta = \{ v_1, \dots, v_{D-1}\} $). We note that in general we can define Pauli-X and Pauli-Z stabilizers with respect to $x$- and $z$-cells, $\delta_x$ and $\delta_z$, respectively, with $x+ z = D-2$~\cite{Bombin_2015, Kubica2015}, but for our purposes we can concentrate on $x = 0$ and $z = 1 $ for the $D=3$ case.

We can produce a color-code lattice $\L_c$ given some arbitrary triangulation  of a 3-manifold using a \textit{fattening} procedure. This was shown first by Bombin and Delgado, see Appendix~A of Ref.~\cite{Bombin2007}, but we follow the description on the dual lattice $\L_c^*$ due to Breuckmann and Vuillot, see Ref.~\cite{Vuillot2022}. We take some triangulation  of a $D$ manifold, consisting of $j$-cells $\sigma_j$ and show we can obtain a color-code lattice as defined above. Once again, we denote the cells $\sigma_j$ of the $D$-dimensional manifold by a list of $D+1$ 0-cells $\sigma_0$. To distinguish the initial manifold that we map onto a $D+1$ colourable cellulation of the same manifold that we use to define the color code, we denote $j$-cells of the initial triangulation with the notation $\sigma_j$. In contrast we have denoted the $D$-cells, $(D-2)$-cells and $0$-cells with symbols $\Delta$, $\delta$ and $v$, respectively, for the case of the color code lattice.

We assign a vertex of the color-code lattice $v$ to all of the cells of $\mathcal{M}^D$, that is, a single vertex for each cell $\sigma_j$ for all $0\le j \le D$. We assign a different color to each vertex depending on the dimensionality of the cell $\sigma_j$ to which it is assigned. Two vertices of the color code lattice share an edge if and only if $\sigma_k \subset \sigma_j$ for all $k < j$. We then specify a tetrahedra $\Delta$ for every $(D+1)$-tuple of vertices $v$ associated to cells $\sigma_j$ if and only if $\sigma_0 \subset \sigma_1 \subset \dots \subset \sigma_D$. By definition we have $(D+1)$-colourable lattice since every tetrahedra $\Delta$ has a vertex of each colour, since $\Delta$ includes vertices $v$ associated to exactly one $j$-cell $\sigma_j$ for each value of $j$. We show this procedure for a $3$-cell $\sigma_3$ in Fig.~\ref{fig:ColorCodeLattice}.

In the following, we focus on discussing the properties of 3D color codes and switch back to the standard color-code lattice  where qubits are defined on the vertices. The 3D color code is then defined on a 3-dimensional lattice $\L_c$ which is 4-valent and 4-colorable.    The 3D color code is constant-depth equivalent to three copies of 3D toric (surface) codes, i.e., $CC(\L_c) \cong \Motimes_{i=1}^3 TC(\L_i) $. 
More concretely, there exists a constant-depth disentangling local unitary $V$, which can disentangle a single 3D color code into three toric codes \cite{Kubica:2015br}:
\be\label{eq:disentangler}
V[CC(\L_c)\otimes \SS]V^\dag = \Motimes_{i=1}^3 TC(\L_i),
\ee
where $\SS$  represents  the stabilizer groups of decoupled ancilla qubits and  $TC(\L_i)$ represents the stabilizer group of the toric code defined on the shrunk lattice $\mathcal{L}_i$ obtained from $\L_c$ by shrinking $3$-cells of color $\text{C}_i$.

The code space of a single toric code defined on a closed 3D manifold $\mathcal{M}^3$ (such that the lattice $\mathcal{L}_i$ forms a tessellation of $\mathcal{M}^3$) is
\be
\H_{TC(\M^3)} = \mathbb{C}^{|H_1(\mathcal{M}^3; \ZZ_2)|};
\ee
where $H_1(\mathcal{M}^3; \ZZ_2)$ represents the first $\ZZ_2$-homology group of $\M^3$, which correspond to the non-contractible 1-cycles where the logical-$Z$ strings are supported.   
On the other hand, the logical-$X$ membranes are supported on the non-contractible 2-cycles corresponding to the second $\ZZ_2$-homology group $H_2(\mathcal{M}^3; \ZZ_2)$. Due to Poincaré duality we have the following isomorphism:
\be
H_1(\mathcal{M}^3; \ZZ_2)\cong H^2(\mathcal{M}^3; \ZZ_2) \cong H_2(\mathcal{M}^3; \ZZ_2).
\ee
Now we define the $i^\text{th}$ $\ZZ_2$-Betti number of $\M^3$ as $b_i(\M^3;\ZZ_2)$$=$$Rank(H_i(\mathcal{M}^3; \ZZ_2))$ \footnote{For conciseness, we drop $\ZZ_2$ in the followed texts and when we say Betti number we will always refer to the $\ZZ_2$ Betti number}.  The number of encoded logical qubits $k$ equal to both the 1st and 2nd Betti number in this case, i.e., $k=b_1(\M^3; \ZZ_2)=b_2(\M^3; \ZZ_2)$, where the second equality is a manifestation of the Poincaré duality.   The code space of the 3D color code defined on manifold $\M^3$ is hence
\be
\H_{CC(\M^3)} =  \H_{TC(\M^3)}^{\otimes 3} = \mathbb{C}^{|H_1(\mathcal{M}^3; \ZZ_2)\oplus H_1(\mathcal{M}^3; \ZZ_2)\oplus H_1(\mathcal{M}^3; \ZZ_2)|}.
\ee

\begin{figure}
\includegraphics{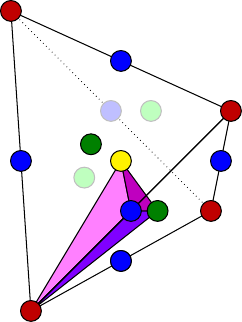}
\caption{A color-code lattice from a triangulation of manifold $\M^3$ via a fattening procedure. We show a single tetrahedron of $\mathcal{M}^3$.  We construct a color-code lattice with four-colourable vertices from an arbitrary cellulation by placing a vertex of colour $j$ onto each of the $j$-cells $\sigma_j$ of manifold $\mathcal{M}^3$. The $\sigma_j$ cells are each assigned a color-code vertex $v$ that is coloured red, blue, green, and yellow for $j = 0, 1, 2, 3$, respectively. We show a single qubit, associated to a tetrahedron $\Delta$ in purple. The figure shows it is adjacent to exactly one vertex of each colour. \label{fig:ColorCodeLattice}}
\end{figure}

\subsection{Logical non-Clifford gates and the triple-intersection invariant}\label{sec:non-Clifford_triple}

For each non-contractible 1-cycle $\alpha_1$ (string) belonging to the homological class $[\alpha_1]$$\in$$H_1(\mathcal{M}^3; \ZZ_2)$, there are three logical-$Z$ operators supported on it in the 3D color code,  which are denoted by $\overline{Z_{\alpha_1;1}}$, $\overline{Z_{\alpha_1;2}}$ and $\overline{Z_{\alpha_1;3}}$.  Similarly, for each non-contractible 2-cycle $\beta_2$ belonging to the homological class $[\alpha_2]$$\in$$H_2(\mathcal{M}^3; \ZZ_2)$,  there are three logical-$X$ operators supported on them, which are denoted by $\overline{X_{\beta_2;1}}$, $\overline{X_{\beta_2;2}}$ and $\overline{X_{\beta_2;3}}$. According to Ref.~\cite{Kubica:2015br}, the disentangling unitary $V$ can map the triplet of logical-$Z$ ($X$) operators in the 3D color code to the triplet of logical-$Z$ ($X$) operators in the corresponding 3D toric codes:       
\be
V: \overline{Z_{\alpha_1;i}} \rightarrow V \overline{Z_{\alpha_1;i}}  V^\dag = \overline{Z}_{\alpha_1}^{(i)} , \quad \overline{X_{\beta_2;i}} \rightarrow V \overline{X_{\beta_2;i}}  V^\dag = \overline{X}_{\beta_2}^{(i)}.
\ee
Here, $\overline{Z}_{\alpha_1}^{(i)}$ and $\overline{X}_{\beta_2}^{(i)}$ represent the logical $Z$ and $X$ operators supported on the cycles $\alpha_1$ and $\beta_2$ respectively in the $i$-th toric codes.

In order to label the logical qubits, we choose a particular 2nd homology basis $B_2 =\{\alpha_2\}$. Therefore, we can label the logical qubits with the label of the logical-$X$ membranes, i.e., $(\alpha_2; i)$ with $\alpha_2 \in B_2$.  We note that the dual logical-$Z$ string of the logical-$X$ membrane $\overline{X_{\alpha_2;i}}$ ($\overline{X}_{\alpha_2}^{(i)}$)  is $\overline{Z_{\alpha_1;i}}$  ($\overline{Z}_{\alpha_1}^{(i)}$) satisfying $|\alpha_1 \cap \alpha_2|=1$ and $|\alpha_1 \cap \alpha^{(2)}_2|=0$ for any $\alpha^{(2)}_2 \in B_2$ such that $\alpha^{(2)}_2 \neq \alpha_2$, where `$|\cdot\cap\cdot|\in \ZZ_2$' represents the $\ZZ_2$ intersection number, i.e., the number of intersection points mod~2 \footnote{A more general quantity to be considered is the algebraic intersection number $|\alpha_1 \cap \beta_2|_\ZZ \in \ZZ$, which has a `$\pm$' sign for each intersection point depending the orientation of the cycles and is a topological invariant. In contrast, geometric intersection number only associate a `$+$' sign for each intersection point, and is hence not topological but depending the choice of the representative of each cycle. In the $\ZZ_2$ case, $`+'$ and $`-'$ signs are identified and it turns out the algebraic and geometric  intersection numbers mod 2 are the same.  For more general case like the $\ZZ_N$ toric code, one should consider algebraic intersection number mod $N$.}. We see that $\alpha_1$ is just the Poincaré dual cycle of $\alpha_2$ \footnote{Strictly speaking Poincaré dual is between the cycle $\alpha_2$ and the cocycle $\alpha^1$.  In this paper, we sometimes relax the use of this terminology by calling $\alpha^1$ the Poincaré dual cycle of $\alpha_2$.}.  From now on, we use the same Greek letter to denote the pair of Poincaré dual cycles. The dual pairs obey the anti-commutation relations $\overline{X_{\alpha_2;i}} \ \overline{Z_{\alpha_1;i}}$$=$$-\overline{Z_{\alpha_1;i}} \  \overline{X_{\alpha_2;i}}$ and $\overline{X}_{\alpha_2}^{(i)} \overline{Z}_{\alpha_1}^{(i)} = - \overline{Z}_{\alpha_1}^{(i)} \overline{X}_{\alpha_2}^{(i)}$. We hence also obtain the corresponding 1st homology basis $B_1 =\{\alpha_1\}$.  Therefore, both  labels $(\alpha_2; i)$ and $(\alpha_1; i)$ correspond to the same logical qubit.  

Now we discuss the transversal non-Clifford gate~\cite{Bombin:2007eh, Kubica:2015br}. It is known that the corresponding graph of the color code lattice $G=(\mathcal{V}, \mathcal{E})$ containing the vertices and edges of $\L$ is a bipartite graph.  Therefore, the vertices of $\L$ can be divided into two groups, $\mathcal{V}^a$ and $\mathcal{V}^b$, i.e.,  $\mathcal{V}=\mathcal{V}^a \cup  \mathcal{V}^b$.  In particular, vertices in $\mathcal{V}^a$ are only adjacent to vertices in $\mathcal{V}^b$ and vice versa. We then define the transversal $T$ gate as:
\be\label{eq:defineT}
\widetilde{T} = \Motimes_{j\in \mathcal{V}^a} T(j) \Motimes_{j\in \mathcal{V}^b} T^\dag(j). 
\ee
It has been previously shown that the transversal $T$ gate maps the code space of the color code back to itself, i.e., $\widetilde{T}:\H_{CC(\M^3)} \rightarrow \H_{CC(\M^3)}$ \cite{Kubica:2015br}, and is hence a logical gate. In the following, we try to identify what type of logical gate $\widetilde{T}$ is. 

We first consider the action of $\widetilde{T}$ on a logical-$X$ membrane in the 3D color code denoted by $\lo{X_{\alpha_2; 1}}$, which is supported on the non-contractible 2-cycle $\alpha_2$. The action induces the following map:
\be
\widetilde{T} \lo{X_{\alpha_2; 1}} \widetilde{T}^\dag  = \lo{X_{\alpha_2; 1}} \widetilde{S}_{\alpha_2; 2,3}.
\ee
Here, $\widetilde{S}_{\alpha_2; 2,3}$ is the corresponding transversal $S$ gate with exactly the same support of the logical-X membrane $\lo{X_{\alpha_2; 1}}$ in the color code due to the conjugation of $\widetilde{T}$.  Therefore, $\widetilde{S}_{\alpha_2; 2,3}$ is also supported on the 2-cycle (membrane) $\alpha_2$.  According to Ref.~\cite{Kubica:2015br}, when conjugated by the disentangling unitary $V$, the operator is mapped to
\be
V \widetilde{S}_{\alpha_2; 2,3} V^\dag = \widetilde{\text{CZ}}^{(2,3)}_{\alpha_2},
\ee
where $\widetilde{\text{CZ}}^{(2,3)}_{\alpha_2}$ represents the transversal CZ operator acting on the 2-cycle (membrane) $\alpha_2$ between toric code copies 2 and 3. 
Therefore, when further applying the disentangling unitary $V$, the action induces the following map: 
\be\label{eq:T_gate_map}
V  \widetilde{T} \lo{X_{\alpha_2; 1}} \widetilde{T}^\dag V^\dag = \overline{X}_{\alpha_2}^{(1)} \widetilde{\text{CZ}}^{(2,3)}_{\alpha_2}.
\ee

  Now we consider a triplet of non-contractible 2-cycles belonging to the homology basis:  $\alpha_2$, $\beta_2$, $\gamma_2 \in B_2$.    One can see that only when the three 2-cycles have $\ZZ_2$ intersection at a single point, i.e.,  $|\alpha_2 \cap \beta_2 \cap \gamma_2|=1$, the transversal CZ operator $\widetilde{\text{CZ}}^{(2,3)}_{\alpha_2}$ can induce a logical CZ operator on the logical qubit labeled by $(\beta_2; 2)$ and $(\gamma_2; 3)$ with the corresponding logical-X operators being $\overline{X}_{\beta_2}^{(2)} $ and $\overline{X}_{\gamma_2}^{(3)}$.     This fact is illustrated in Fig.~\ref{fig:CCZ_arrangement}, and can be explained as follows.

\begin{figure}[t]
  \includegraphics[width=1\columnwidth]{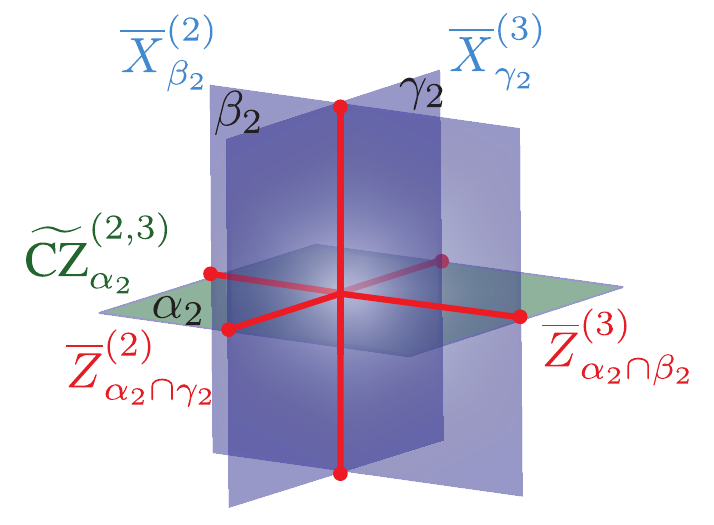}
  \caption{The intersection property which guarantees a non-trivial logical CCZ action.}
  \label{fig:CCZ_arrangement}
\end{figure}

The action of transversal CZ operator $\widetilde{\text{CZ}}^{(2,3)}_{\alpha_2}$ induces the following map on the logical operator $\overline{X}_{\beta_2}^{(2)}$ in the second copy of the toric code:  
\be\label{eq:CZ_mapping}
\widetilde{\text{CZ}}^{(2,3)}_{\alpha_2}:  \overline{X}_{\beta_2}^{(2)}  \rightarrow  \overline{X}_{\beta_2}^{(2)} \overline{Z}_{\alpha_2 \cap \beta_2}^{(3)},
\ee
where the logical-$Z$ string is supported on $\alpha_2 \cap \beta_2$. Now if the logical CZ gate acting on logical qubits associated with $\overline{X}_{\beta_2}^{(2)}$ and $\overline{X}_{\gamma_2}^{(3)}$, then one must satisfy the following anti-commutation relation:
\be
\overline{X}_{\gamma_2}^{(3)} \overline{Z}_{\alpha_2 \cap \beta_2}^{(3)} = - \overline{Z}_{\alpha_2 \cap \beta_2}^{(3)} \overline{X}_{\gamma_2}^{(3)},
\ee
which requires the following intersection condition:
\be
|(\alpha_2 \cap \beta_2) \cap \gamma_2| = 1.
\ee


Now when we do the similar investigation on the action of transversal gate $\widetilde{T}$ on all the logical-$X$ membranes  $\overline{X}_{\beta_2}^{(2)} $ and $\overline{X}_{\gamma_2}^{(3)}$, we can conclude that a logical CCZ gates is applied on the triplet of logical qubits labeled by $(\alpha_1; 1)$, $(\beta_2; 2)$ and $(\gamma_2; 3)$, which are associated with the logical-$X$ membranes  $\overline{X}_{\alpha_2}^{(1)} $ , $\overline{X}_{\beta_2}^{(2)} $ and $\overline{X}_{\gamma_2}^{(3)}$, for both the 3D color codes and toric codes, if and only if the triple intersection condition $|\alpha_2 \cap \beta_2 \cap \gamma_2| = 1$ is satisfied.

Now we can derive the following generic formula for the transversal gate on an arbitrary closed 3-manifold $\M_3$, including the special cases of hyperbolic 3-manifolds:
\be\label{eq:CCZ_formula}
\widetilde{T} = \prod_{\alpha_2, \beta_2, \gamma_2 \in B_2 }  [\lo{\text{CCZ}}((\alpha_2; 1),(\beta_2; 2),(\gamma_2; 3))]^{|\alpha_2 \cap \beta_2 \cap \gamma_2|},
\ee
where $B_2$ represents a 2nd homology basis.   We note that $V \widetilde{T} V^\dag$ implements the same collective logical gate on the three copies of toric code supported on the same 3-manifold $\M_3$.  

We emphasize that the logical gate structure is completely determined by the $\ZZ_2$ triple intersection number $|\alpha_2 \cap \beta_2 \cap \gamma_2|\in \ZZ_2$, i.e., the number of transversal intersections of the three surfaces mod 2. This number is a topological invariant, i.e., its value is independent of the re-triangulation of the 3-manifold $\M^3$ and is also independent of the choice of the representatives $\alpha_2, \beta_2, \gamma_2$ in their own homology class.  We can also re-express the $\ZZ_2$ triple intersection number as a skew-symmetric 3-form on $H_2(\mathcal{M}^3; \ZZ_2)$ with values in $\ZZ_2$. This 3-form is the Poincaré dual of the cup product (`$\cup$') on $H^1(\M^3; \ZZ_2)$, i.e. if $\alpha^1$, $\beta^1$, and $\gamma^1$ are the Poincaré dual cocycles in $H^1(\M^3; \ZZ_2)$ then 
\begin{align}
    | \alpha_2 \cap \beta_2 \cap \gamma_2 |:=\int_{\M^3} \alpha^1 \cup \beta^1 \cup \gamma^1   \in \ZZ_2. 
\end{align}
Note that this topological invariant naturally arises in the path integral of the corresponding topological quantum field theory (the $\ZZ_2^3$ gauge theory) as will be derived in Sec.~\ref{sec:TQFT}.

\subsection{Parallelizable logical Clifford gates}\label{sec:parallelizable}

From the discussion in { Sec.~\ref{sec:non-Clifford_triple} }, we can see that transversal CZ gate can be applied only on a particular codimension-1 (2-dimensional) submanifold, i.e., a particular membrane (2-cycle) $\alpha_2$, { as shown in Fig.~\ref{fig:CCZ_arrangement}.}

The transversal CZ gates acting on toric code copy $i$ and $j$ ($i\neq j$) and supported on the membrane $\alpha_2$, denoted by $\widetilde{\text{CZ}}^{(i,j)}_{\alpha_2}$ (or equivalently the transversal $S$  gate in the color code $\widetilde{S}_{\alpha_2; i,j}$), is equivalent to the following logical gate:
\be\label{eq:CZ_formula}
\widetilde{\text{CZ}}^{(i,j)}_{\alpha_2}\sim \widetilde{S}_{\alpha_2; i,j} = \prod_{\beta_2, \gamma_2 \in B_2 } [\lo{\text{CZ}}((\beta_2; i),(\gamma_2; j))]^{|\alpha_2 \cap \beta_2 \cap \gamma_2|}.
\ee

Note that these transversal gates are only applied on a submanifold $\alpha_2$ instead of the entire 3-manifold, and are hence a 1-form symmetry operators as will be discussed in Sec.~\ref{sec:TQFT} in details.   It can hence be individually addressed, in contrast to the collective logical CCZ gate in Eq.~\eqref{eq:CCZ_formula}.   More precisely, there are in total $|H_2(\M^3;\ZZ_2)| = 2^{b_2(\M^3;\ZZ_2)}=2^k$ such addressable logical gates.   This exponential scaling is in contrast to the expected $O(k)$ scaling in the fold-transversal gates in qLDPC codes \cite{breuckmann2022fold, Quintavalle:2022fold} as explained in the introduction, which is a 0-form symmetry acting on the entire system. Note that the number of independent 1-form symmetry  generators, or equivalently the generators for the corresponding logical gate set, is equal to the 2nd and 1st Betti number $b_2(\M^3;\ZZ_2)=b_1(\M^3;\ZZ_2)= k$. 

Since all the $\widetilde{\text{CZ}}^{(i,j)}_{\alpha_2}$ (or equivalently $\widetilde{S}_{\alpha_2; i,j}$) are composed with diagonal gates and hence commute with each other, they can be parallelized, as will be explained in more details in Sec.~\ref{sec:quasi-hyperbolic}.

\subsection{Interaction hypergraph and hypergraph magic state injection}\label{sec:interaction_hypergraph}

\begin{figure}[t]
  \includegraphics[width=1\columnwidth]{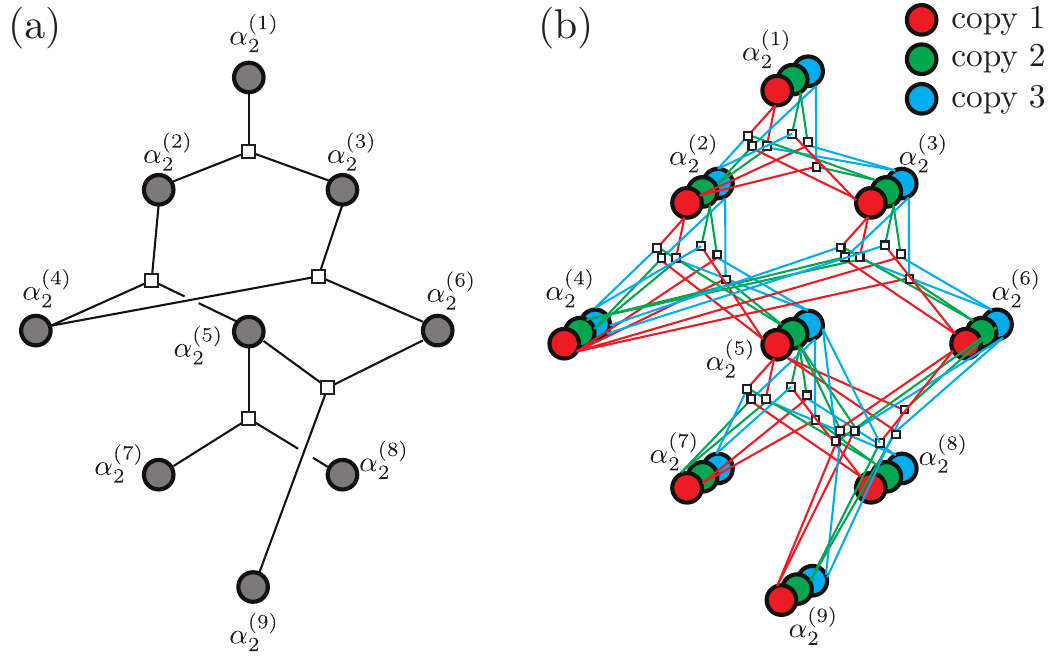}
  \caption{(a) The base interaction hypergraph (intersection hypergraph) of a particular 3-manifold. (b) The corresponding (full) interaction hypergraph of the color codes defined on the same 3-manifold. Red, green and blue vertices represent the logical-$X$ operators belonging to the three copies of toric codes respectively.}
  \label{fig:base_hypergraph}
\end{figure}

The collective non-Clifford logical CCZ gate structure in Eq.~\eqref{eq:CCZ_formula} can be completely encoded in an \textit{interaction hypergraph} of logical qubits as described below. 

We start with a 3-manifold $\M^3$ and with particular choice of the 2nd homology basis $B_2$. Its $\ZZ_2$ triple intersection number $|\alpha_2^{(i)} \cap \alpha_2^{(j)} \cap \alpha_2^{(k)}| \in \ZZ_2$ with $\alpha_2^{(i)}, \alpha_2^{(j)},\alpha_2^{(k)} \in B_2$ can be encoded into a \textit{base interaction hypergraph}, or synonymously an \textit{intersection hypergraph}, as shown in Fig.~\ref{fig:base_hypergraph}(a).   The hypergraph $G_h$$=$$(V, E_h)$ is composed of a set of vertices $V$ (represented by grey circles) and a set of hyperedge $E_h$ (represented by a tri-junction with three edges meeting at a square node) connecting triplets of vertices. Each vertex $v\in V$ corresponds to a 2-cycle $\alpha_2^{(i)} \in B_2$.  Each hyperedge $e_h \in E_h$ corresponds to a $\ZZ_2$ triple intersection point, which is connected to three vertices with their represented 2-cycles satisfying the condition $|\alpha_2^{(i)} \cap \alpha_2^{(j)} \cap \alpha_2^{(k)}|=1$. Note that as illustrated in Fig.~\ref{fig:base_hypergraph}(a), the base interaction hypergraph does not have to be a planar hypergraph since the hyperedges can have crossings.  

We then lift the base interaction hypergraph to the \textit{(full) interaction hypergraph} which encodes the full information of Eq.~\eqref{eq:CCZ_formula}, as illustrated in Fig.~\ref{fig:base_hypergraph}(b).   Since a single 3D color code is equivalent to three copies of 3D toric codes, each single vertex in the base interaction hypergraph is lift to three vertices corresponding to three logical qubits, represented by red, green, blue circles in Fig.~\ref{fig:base_hypergraph}(b) respectively.  A single hyperedge connecting a triplet of vertices in the base interaction hypergraph is now lifted to $3!=6$ hyperedges according to all possible permutations of the three copy labels for a given triplet of 2-cycles as suggested by Eq.~\eqref{eq:CCZ_formula},  where the edge connecting a particular vertex in the hyperedge (tri-junction) is chosen to have the same color as the vertex in  Fig.~\ref{fig:base_hypergraph}(b).  Each hyperedge corresponds to a logical CCZ gate acting on the three connected logical qubits (vertices). More concrete examples of 3-manifolds and their corresponding interaction hypergraphs will be given in Sec.~\ref{sec:code_construction}.   

Quite interestingly, the full-interaction hypergraph has a one-to-one correspondence with the \textit{quantum hypergraph states} \cite{Rossi_2013}, which can be considered as the generalization of graph states. One can define a \textit{$q$-uniform hypergraph state} on a hypergraph $G_h=(V, E_h)$ as follows.   One assigns to each vertex a qubit initialized in the $\ket{+}$ state. For each $q$-hyperedge, one perform a $\text{C}^{q} \text{Z}$ gate between the $q$ connected qubits (vertices) labeled by $v_1, v_2, ..., v_q$. We hence reach the following $q$-uniform hypergraph state:
\be
\ket{g_q} = \prod_{\{v_1, v_2, ..., v_q\} \in E_h} \text{C}^{q} \text{Z}(v_1, v_2, ..., v_q)\ket{+}^{\otimes m},
\ee
where $m=|V|$ represents the total number of qubits (vertices).

\begin{figure}[t]
\includegraphics[width=1\columnwidth]{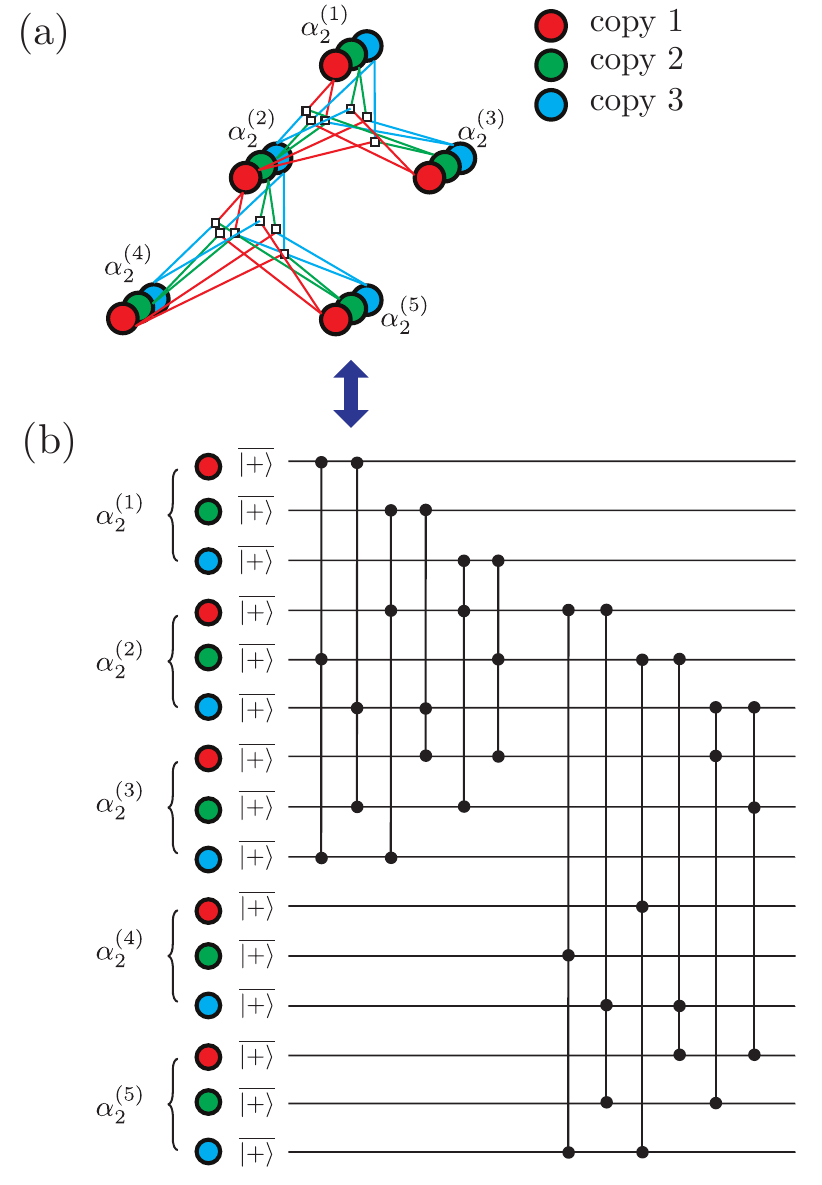}
  \caption{Correspondence between an interaction hypergraph of a color code on a 3-manifold in (a) with a logical 3-uniform hypergraph magic  state in (b) injected via a transversal $T$ gate acted on the entire code with each logical qubit initialized in state $\lo{\ket{+}}$.  The number of logical CCZ gates in the quantum circuit in (b) equaling the number hyperedges in the corresponding hypergraph shown in (a) defines the complexity of the hypergraph state.}
  \label{fig:hypergrah_state}
\end{figure}

In the case of the color code defined on 3-manifold, we have $q=3$.  If we initialize each logical qubit in the logical state $\lo{\ket{+}}$ and then apply the transversal $T$ gate, we arrive the following 3-uniform logical hypergraph state according to Eq.~\eqref{eq:CCZ_formula}: 
\begin{align}\label{eq:logical_hypergraph_state}
\non \lo{\ket{g_3}} =& \prod_{\alpha_2, \beta_2, \gamma_2 \in B_2 }  [\lo{\text{CCZ}}((\alpha_2; 1),(\beta_2; 2),(\gamma_2; 3))]^{|\alpha_2 \cap \beta_2 \cap \gamma_2|} \\
\non & \qquad \qquad \cdot \lo{\ket{+}}^{\otimes k} \\
=& \prod_{\{v_1, v_2, v_3 \}\in E_h} \lo{\text{CCZ}}(v_1, v_2, v_3)  \lo{\ket{+}}^{\otimes k}.
\end{align}
Here, $k=3\cdot b_1(\M^3; \ZZ_2)$ denote the number of logical qubits equaling three times the Betti number of the 3-manifold $\M^3$.  The quantum circuit of the hypergraph state corresponding to its interaction hypergraph is illustrated in Fig.~\ref{fig:hypergrah_state}.

Note that the hypergraph state contains so-called `\textit{magic}' \cite{chen2023magic}, since it is a non-stabilizer state beyond the description of the  stabilizer formalism. In our context, the logical hypergraph state in Eq.~\eqref{eq:logical_hypergraph_state} is a non-stabilizer state at the logical level, even though the microscopic wavefunction is in the code space a stabilizer code.  Therefore, this hypergraph state is a specific type of magic state which is a valuable computational resource for the magic state injection and distillation protocols \cite{bravyi1998}. We note that instead of slowly doing extensive rounds of magic state distillation starting from a low-fidelity magic state at $O(1)$ distance, here we can directly inject a high-fidelity hypergraph magic state at $O(d)$ distance without distillation. We hence call this scheme a `\textit{magic state fountain}'.  The injected hypergraph magic state can either be used as resources for gate-based or measurement-based quantum computation schemes \cite{Takeuchi:2019_hypergraph} at the logical level.  

In order to characterize this computation resource, we also define the \textit{complexity} of the 3-uniform hypergraph state as the number of CCZ gates acted upon  the initial product state, or equivalently the number of vertices $|V(G_h)|$ in the underlying hypergraph.  Therefore, the complexity of the logical 3-uniform hypergraph state $\kappa(g_3)$ is equal to the total number of $\ZZ_2$ triple intersection points $N_p(\M^3)$ in the underlying 3-manifold $\M^3$ times a factor of six:
\be
\kappa(g_3)=6 N_p(\M^3)=|V(G_h)|.
\ee
We hence observe a profound connection between the complexity of quantum states and the topology/geometry of manifolds.  We note that this complexity can be used to characterize the computation power of this resource state either in the gate-based or measurement-based computation scheme (at the  logical level), and can also be related to the complexity of classical simulation of quantum circuits involving such hypergraph states, which will be explored more in future works. 

Finally, we note that the parallelizable  logical CZ gate structure is completely encoded in an interaction graph, which will be introduced in Sec.~\ref{sec:triple_quasi-hyperbolic}.

{

\section{Logical gates for three identical copies of toric codes on 3D manifold via constant-depth circuits and the connection to emergent $k$-form symmetries in TQFT}\label{sec:TQFT}

In this section, we derive the logical gate expressions Eq.~\eqref{eq:CCZ_formula} and  \eqref{eq:CZ_formula}  in a more fundamental way using an operator-valued cochain formalism originated  from (3+1)D topological quantum field theory (TQFT) (see also Ref.~\cite{Hsin2024_non-Abelian} for the application of this formalism). This derivation is independent of the detailed microscopic lattice models one chooses such as a color code or multiple copies of toric codes, or any lattice variation of these models.   In particular, the logical gate operators originate from $k$-form emergent symmetries expressed as TQFT path integrals and mathematically correspond to cohomology operation via cup products. In this way, the triple intersection numbers naturally arise as  topological invariants in the TQFT path integral.  

Even more interestingly, this operator-valued cochain formalism naturally adapts to a discrete lattice or trianulation.   This allows us to explicitly construct constant-depth circuits acting on three identical copies of toric codes defined on an arbitrary triangulation of arbitrary closed 3-manifolds.   This is different form the transversal CCZ \cite{Vasmer2019} acted on three different toric codes, with two of them having larger stabilizer weights than the standard cubic lattice 3D toric-code.  Our construction hence allows implementing logical CCZ gate with lower-weight stabilizer.  More elaborated study of this formalism including the logical gates on arbitrary lattices such as cubical lattices and its application to a broader class of cohomology operation will be presented in the follow-up papers \cite{zhu2025topological, Hsin2024:classifying}.


\subsection{Operator-valued cochain formalism }

We now introduce an \textit{operator-valued cochain formalism} to describe the operators in the code, which originates from a (3+1)D topological quantum field theory (TQFT): the (3+1)D $\ZZ_2\times \ZZ_2 \times \ZZ_2$ gauge theory 
\footnote{The action of this TQFT in the continuum is
\begin{align}
   S_{\ZZ_2^3}= \pi \int_{\M^4} a^{(1)} db^{(1)} + a^{(2)} db^{(2)} + a^{(3)} db^{(3)}.
\end{align}
Here,  $a^{(i)}$ and $b^{(i)}$ represent the $\ZZ_2$ 1-form electric gauge fields and 2-form magnetic gauge fields defined  on the triangulated (3+1)D space-time manifold $\M^4 = \M^3 \times S^1_t$  ($\M^3$ being the space manifold, and the circle $S^1_t$ being the time part).} equivalent to a 3D color code $CC(\L_c)$ or three identical copies of 3D $\ZZ_2$ toric codes $\Motimes_{i=1}^3 TC(\L)$ supported on an arbitrary  triangulation $\L$ of an arbitrary 3-manifold \cite{10.21468/SciPostPhys.14.4.065}.
   We define the operator valued 1-cochain $a^{(i)} \in C^1(\mathcal{L}; \ZZ_2)$ with its eigenvalues belonging to $\{0,1\}$ and $i$ denotes the $i^\text{th}$ copy of the toric code. The coefficient of each edge (1-cell) $e$ in the 1-cochain corresponds to a Pauli-$Z$ operator via 
\be\label{eq:a_Pauli}
(-1)^{a^{(i)} (e)}= Z^{(i)}(e) \in \{-1,1\}.
\ee
Note that the Pauli operator has eigenvalues $\pm 1$ instead.
Physically, $a^{(i)}$ corresponds to the $\ZZ_2$ electric gauge field. Similarly, there also exists a 1-cochain $b^{(i)}
$ corresponding to the magnetic $\ZZ_2$  gauge field, which corresponds to a Pauli-X operator on each edge $e$ as 
\be\label{eq:b_Pauli}
(-1)^{b^{(i)} (e)}= X^{(i)}(e) \in \{-1,1\}.
\ee

The coboundary operator $d$, a lattice analog of the exterior derivative in the continuum case, acts on the $\ZZ_2$-valued 1-cochain as 
\be
(da^{(i)})(f) = \sum_{e \subset f}a^{(i)}(e).
\ee
The $Z$-stabilizer of the 2D toric code localed on face $f$ can be expressed in terms of the operator-valued cochain as
\be\label{eq:Z_stabilizer}
{S^{(i)}_{Z; f} = \prod_{e \subset f} Z^{(i)}(e) = (-1)^{da^{(i)}(f)}},
\ee
which is also considered as a flux term in the gauge theory. The $Z$-stabilizer condition $S^{(i)}_{Z; f}=1$ becomes the 0-flux condition 
\be
da^{(i)}(f)=0.  
\ee

The $X$-stabilizers on the vertex $v$ can be expressed as
\be
S^{(i)}_{X; v} = \prod_{e \supset v} X^{(i)}(e) = (-1)^{\sum_{e \supset v} b^{(i)}(e)},
\ee
The X-stabilizer condition $S^{(i)}_{X; v}=+1$ corresponds to Gauss's law (with zero charge) in the $\ZZ_2$ gauge theory and can be expressed as 
\be
\sum_{e \supset v} b^{(i)}(e)  = 0.
\ee
Furthermore, the anticommutation relation between the Pauli-$X$ and -$Z$  operator $X^{(i)}_eZ^{(i)}_{e'}=(-1)^{\delta_{e,e'}}Z^{(i)}_eX^{(i)}_{e'}$ leads to the following anticommutation relation:
\be\label{eq:anti-commutation}
  X^{(i)}_e (-1)^{\mathsf{f}(a^{(i)})}X^{(i)}_e=(-1)^{\mathsf{f}(a^{(i)}+\tilde e)},
\ee
where $\mathsf{f}$ is an aribtrary function of $a^{(i)}$ and $\tilde e$ is the indicator $1$-cochain that takes value $1$ on edge $e$ and 0 otherwise.

\begin{figure}[]
\includegraphics[width=0.7\columnwidth]{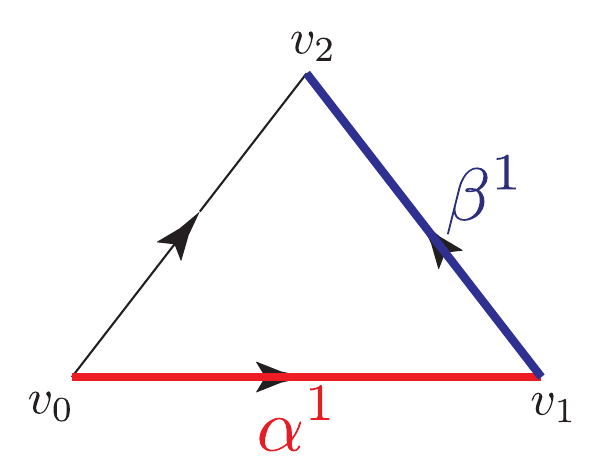}
  \caption{Evaluating the cup product on a 2-simplex via the formula $(\alpha^1 \cup \beta^1)([v_0v_1v_2])
=\alpha^1([v_0v_1])\beta^1([v_1 v_2])$. The three vertices have the order  $v_0 \prec v_1 \prec v_2$, and the branching structure is evaluated by the arrow pointing from lower-order to higher-order vertices.}
  \label{fig:cup_example}
\end{figure}

\subsection{Logical CCZ as a 0-form symmetry}
We now introduce a logical gate as a 0-form symmetry operator in three identical copies of toric codes supported on the triangulation $\mathcal{L}$ of a 3-manifold via a cohomology operation.   One can consider the following unitary originated from the TQFT path integral:
\be\label{eq:logical_CZ}
U = (-1)^{\int_{\mathcal{L}} a^{(1)} \cup a^{(2)} \cup a^{(3)}},
\ee
where $a^{(i)}$ correspond to the electric gauge field on the three copies, and $\int_{\mathcal{L}}$ is a discrete sum over all the simplices in  $\mathcal{L}$. Note that $\mathcal{L}$ should be consider as a 3-chain, and this sum should be viewed as a chain-cochain paring.  

For a triangulated manifold, one can define the cup product `$\cup$' between a $p$-cochain $\alpha^p$ and $q$-cochain $\beta^q$ evaluated on a $(p+q)$-simplex $[v_0v_1\cdots v_{p+q}]$ as \cite{Hatcher:2001ut} 
\begin{align}
\non & (\alpha^p \cup \beta^q)([v_0v_1\cdots v_{p+q}]) \\
=&\alpha^p([v_0v_1\cdots v_{p+q}])\beta^q([v_p v_{p+1},\cdots,v_{p+q}]), 
\end{align}
where the arguments are composed of the labels of ordered vertices $v_i$ obtained from a branching structure of the triangulation (satisfying $v_0 \prec v_1 \prec v_2, \cdots \prec v_{p+q}$). An illustration of $p=q=1$ is illustrated in Fig.~\ref{fig:cup_example}.
  For references of TQFT path integrals with cup products, we refer the readers to Refs.~\cite{chen2021higher, lan2019fermion, chen2023loops, tata2023anomalies} for more details.

The cup-product evaluated on each 3-simplex $[v_0v_1v_2v_3]$ ($v_i$ represents the four vertices in the simplex) 
\begin{align}
\non & (a^{(1)} \cup a^{(2)} \cup a^{(3)}) ([v_0v_1v_2v_3]) \\
=& a^{(1)}([v_0v_1])a^{(2)}([v_1v_2])a^{(3)}([v_2v_3]).
\end{align}
We can hence re-express the unitary $U$ as
\begin{align}\label{eq:logical_gate_3D}
\non U=& (-1)^{\int_{[v_0v_1v_2v_3] \in \L}  a^{(1)}([v_0v_1])a^{(2)}([v_1v_2])a^{(3)}([v_2v_3])} \\
=&\prod_{[v_0v_1v_2v_3]\in \mathcal{L}} \text{CCZ}([v_0v_1],[v_1v_2],[v_2v_3]),
\end{align}
which explicitly shows the corresponding constant-depth circuit composed of many CCZ gates coupling the qubits in three copies of toric codes located on edge $[v_0v_1], [v_1v_2]$ and $[v_2v_3]$ in each simplex $[v_0v_1v_2v_3]$, as illustrated in Fig.~\ref{fig:cup_3d}.

\begin{figure}[]
\includegraphics[width=0.7\columnwidth]{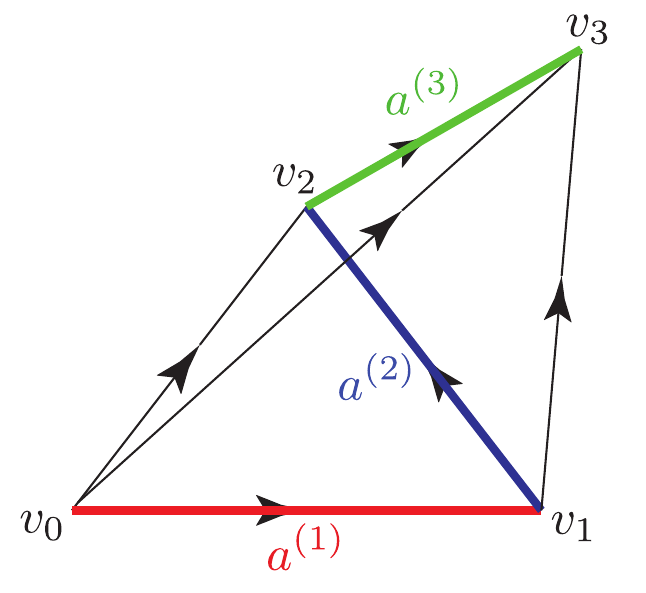}
  \caption{Illustration of the evaluation of the cup product in Eq.~\eqref{eq:logical_gate_3D}.  A CCZ gate is coupling the qubits located on red, blue and green edges, i.e., $[v_0v_1], [v_1v_2], [v_2v_3]$ in each simplex labeled by $[v_0v_1v_2v_3]$.  The CCZ gates in neighboring simplices may have overlap and they hence compose a constant-depth local circuit instead of a strictly transversal gate. }
  \label{fig:cup_3d}
\end{figure}

Now we need to check whether $U$ is a logical gate, meaning its action should preserve the code space, i.e., 
\be
U:\mathcal{C} \rightarrow \mathcal{C}.
\ee
An equivalent condition for the stabilizer code is that 
\be
P_\mathcal{C}[U, \mathcal{S}] P_\mathcal{C} = 0,
\ee
meaning that $U$ commute with the stabilizer group $\mathcal{S}$ when projected to the code space $\mathcal{C}$.  Since $U$ is a diagonal gate, it clearly commutes with all the $Z$-stabilizers.   Now we only need to check the commutation relation for the $X$-stabilizers. 

We consider an $X$-stabilizer $S^{(1)}_{X; v}$$=$$\prod_{e \supset v} X^{(1)}(e)$ located at vertex $v$ in the 1st copy of the toric code.   We then consider conjugating $U$ with $S^{(1)}_{X; v}$:
\begin{align}
\non & S^{(1)}_{X; v} U S^{(1)}_{X; v} \\
\non =& \prod_{e \supset v} X^{(1)}(e) \left[(-1)^{\int_{\Sigma_2} a^{(1)} \cup a^{(2)} \cup a^{(3)}}\right]   \prod_{e' \supset v} X^{(1)}(e')  \\
\non =& (-1)^{\int_{\Sigma_2} (a^{(1)} +\sum_{e' \supset v}\tilde{e}')  \cup a^{(2)} \cup a^{(3)}} \\
=& (-1)^{\int_{\Sigma_2} (a^{(1)} + d \tilde{v}) \cup a^{(2)} \cup a^{(3)}},
\end{align}
where we have used Eq.~\eqref{eq:anti-commutation} in the second equality.  In the last line, $d\tilde{v}=\sum_{e' \supset v}\tilde{e}'$ is an indicator 1-cochain taking value 1 on all the edges $e'$ connected to the vertex $v$ and 0 elsewhere, while $\tilde{v}$ is an indicator cochain taking value 1 on the vertex $v$ and 0 elsewhere. 

We now require the following condition for $U$ to be a logical gate:
\be\label{eq:commuting_ground}
P_\mathcal{C} S^{(1)}_{X; v} U S^{(1)}_{X; v} P_\mathcal{C}= U,
\ee
where $P_\mathcal{C}$ is the projector to the code space $\mathcal{C}$. This condition means $U$ commutes with the stabilizer $S^{(1)}_{X; v}$ in the code space, which is then equivalent to the following condition:
\be\label{eq:gauge-invariance}
P_\mathcal{C} (-1)^{\int_{\mathcal{L}} (a^{(1)} + d \tilde{v}) \cup a^{(2)} \cup a^{(3)}} P_\mathcal{C} = (-1)^{\int_{\mathcal{L}} a^{(1)} \cup a^{(2)} \cup a^{(3)}}. 
\ee
The condition in Eq.~\eqref{eq:gauge-invariance} should be considered as the gauge-invariance condition in the code space, since $a^{(1)} \rightarrow a^{(1)}+d \tilde{v}$ is just a gauge transformation in the $\ZZ_2$ gauge theory \footnote{This is in analogy to the famous gauge transformation $\vec{a} \rightarrow \vec{a} +  \nabla \chi $ in Maxwell's theory or the corresponding gauge theory in the continuum, where $\vec{a}$ is the quantum vector potential. The gauge invariance of the magnetic field $B=\nabla \times \vec{a}=0$ is satisfied due to the Stokes theorem: $\nabla \times \nabla \chi =0$}.

In order to show the gauge invariance, we need to use the Leibniz rule for the cup product:
\be\label{eq:Leibniz_rule}
d\tilde{v}\cup a^{(2)}= d(\tilde{v} \cup a^{(2)}) + \tilde{v} \cup d a^{(2)}.
\ee
note that we have ignored the minus sign and replace it with the plus sign since these cochains are all $\Z_2$-valued (binary variables).   We hence have 
\begin{align}
\non P_\mathcal{C}  \int_{_{\mathcal{L}}} d\tilde{v}\cup a^{(2)} P_\mathcal{C} =& P_\mathcal{C} \left[\int_{_{\mathcal{L}}} d(\tilde{v} \cup a^{(2)}) + \int_{_{\mathcal{L}}} \tilde{v} \cup d a^{(2)} \right] P_\mathcal{C}  \\
=& P_\mathcal{C} \int_{\partial{_{\mathcal{L}}}}(\tilde{v}\cup a^{(2)}) P_\mathcal{C}+ 0 =0,
\end{align}
where we have used the Stokes theorem in the second equality and the fact that we are considering the triangulation of a closed manifold or equivalently a 3-cycle ${\L}$ without boundaries, namely $\partial \L=0$, to show that the first term is zero.  The second is zero due to the cocycle condition or physically the zeror-flux condition in the code space, i.e., $da^{(2)}(f)=0$ on any face $f$, which corresponds to the stabilizer condition $S^{(2)}_{Z; f} = \prod_{e \subset f} Z^{(2)}(e) = (-1)^{da^{(2)}(f)}=1$ according to Eq.~\eqref{eq:Z_stabilizer}. We have hence proved Eq.~\eqref{eq:commuting_ground} and \eqref{eq:gauge-invariance}. By symmetry, we can also verify the commutation relation with $S^{(2)}_{X; v}$ and $S^{(3)}_{X; v}$. Therefore, we can conclude $U$ is indeed gauge-invariant in the code space, i.e.,
\begin{align}
\non & P_\mathcal{C} (-1)^{\int_{\mathcal{L}} (a^{(1)} + d \chi) \cup (a^{(2)} + d\lambda ) \cup (a^{(3)}+d\eta)} P_\mathcal{C} \\
=& (-1)^{\int_{\mathcal{L}} a^{(1)} \cup a^{(2)} \cup a^{(3)}},  
\end{align}
where $d\chi$, $d\lambda$ and $d\eta$ are arbitrary coboundaries. The unitary operator $U$ is hence a logical gate. 

Mathematically, the gauge invariance condition is associated with the topological invariance of $U$ with respect to arbitrary deformation of the cocycles by coboundaries.

Now we show what type of logical gate $U$ corresponds to. Note the cochain $a^{(i)}$ in Eq.~\eqref{eq:logical_CZ} becomes cocycles in the code space due to the zero-flux condition $da^{(i)}=0$. Therefore, we can re-express them using the 1-cocycle basis $\{\alpha^1\}$, $\{\beta^1\}$ and $\{\gamma^1\}$ for both copies of codes:
\be
a^{(1)}=\sum_{\alpha^1}\hat{n}_\alpha \alpha^1,  \quad a^{(2)}=\sum_{\beta^1} \hat{m}_\beta\beta^1, \quad a^{(3)}=\sum_{\gamma^1}\hat{l}_\gamma \gamma^1,
\ee
where the quantum variable $\hat{n}_\alpha$, $\hat{m}_\beta$ and $\hat{l}_\gamma$ with eigenvalues $\{0,1\}$ are the winding numbers, i.e., coefficients of the cocycle basis.  
We can hence re-express $U$ as
\begin{align}
 U =& \prod_{\alpha^1, \beta^1, \gamma^1}  (-1)^{\int_{\mathcal{L}} (\hat{n}_\alpha \alpha^1 )\cup (\hat{m}_\beta  \beta^1) \cup (\hat{l}_\gamma  \gamma^1)} \cr
   =& \prod_{\alpha^1, \beta^1, \gamma^1} \left[(-1)^{ \hat{n}_\alpha \hat{m}_\beta \hat{l}_\gamma}\right]^{\int_{\mathcal{L}} \alpha^1 \cup \beta^1 \cup \gamma^1}\cr
   =&\prod_{\alpha^1, \beta^1, \gamma^1} \overline{\text{CCZ}}[(\alpha^1; 1), (\beta^1; 2),(\gamma^1; 3)]^{\int_{\mathcal{L}} {\alpha^1 \cup \beta^1 \cup \gamma^1}} \cr
    =& \prod_{\alpha_2, \beta_2, \gamma_2} \overline{\text{CCZ}}[(\alpha_2; 1), (\beta_2; 2),(\gamma_2; 3)]^{ {|\alpha_2 \cup \beta_2 \cup \gamma_2|}},
\end{align}
where the third equality has use the relation $(-1)^{\hat{n}_\alpha\hat{m}_\beta\hat{l}_\gamma}$$=$$\overline{\text{CCZ}}$.  Here, $(\alpha^1; 1)$, $(\beta^1; 2)$ and $(\gamma^1; 3)$ are the labels of the logical qubit using the cocycle basis and the copy number.  The logical CCZ gate between logical qubits $(\alpha^1; 1)$, $(\beta^1; 2)$ and $(\gamma^1; 3)$ is only non-trivial (not logical identity) if and only if $\int_{\mathcal{L}} {\alpha^1 \cup \beta^1 \cup \gamma^1}$ evaluates non-trivially, i.e., have a non-trivial triple intersection of the Poincare dual:  $|\alpha_2 \cup \beta_2 \cup \gamma_2|=1$.

\subsection{Logical CZ gates as 1-form symmetries}

We now consider the following 1-form symmetry operator supported on a 2-cycle $\alpha_2$ (a codimension-1 submanifold) and 
\begin{align}\label{eq:twist_string_operator}
   &U'_{i,j}(\alpha_2) = (-1)^{\int_{\alpha_2} a^{(i)}\cup a^{(j)} },   \\
\non    &\quad (i\neq j=1,2,3),
\end{align}
which corresponds to a constant-depth circuit composed of CZ gates between the $i^\text{th}$ and $j^\text{th}$ copies of identical 3D toric codes.   
We can explicitly show the constant-depth circuit by evaluating the cup product as 
\begin{align}
\non U'_{i,j}(\alpha_2)=& (-1)^{\int_{[v_0v_1v_2] \in \alpha_2}  a^{(1)}([v_0v_1])a^{(2)}([v_1v_2])} \\
=&\prod_{[v_0v_1v_2]\in \alpha_2} \text{CZ}([v_0v_1],[v_1v_2]).
\end{align}

We can also re-express the above equation by doing the path integral over the entire 3-manifold triangulation $\L$ as
\be\label{eq:twist_string_operator_v2}
U'_{i,j}(\alpha_2)= (-1)^{\int_{\L} a^{(i)} \cup a^{(j)} \cup \alpha^1 },
\ee
where the 1-cocyle $\alpha^1$ is the Poincaré dual of the 2-cycle  $\alpha_2$.   It is then straightforward to derive the gauge-invariance property of $U'_{i,j}(\alpha_2)$ and show that it is indeed a logical gate, following the discussion in the previous subsection.

Finally, we can again derive the logical gate type by using the cocycle basis to rewrite Eq.~\eqref{eq:twist_string_operator_v2} as before:
\begin{align}
\non U'_{i,j}(\alpha_2) =&\prod_{\beta^1, \gamma^1}   (-1)^{\int_{\mathcal{L}}  (\hat{m}_\beta  \beta^1) \cup (\hat{l}_\gamma  \gamma^1) \cup \alpha^1} \cr
=& \prod_{\beta^1, \gamma^1}  \left[(-1)^{\hat{m}_\beta \hat{l}_\gamma}\right]^{\int_{\mathcal{L}} \alpha^1 \cup \beta^1 \cup \gamma^1}\\
\non   =& \prod_{\beta^1, \gamma^1} 
\overline{\text{CZ}}[(\beta^1; i),(\gamma^1; j)]^{\int_{\mathcal{L}} {\alpha^1 \cup \beta^1 \cup \gamma^1}} \\
    =& \prod_{\beta_2, \gamma_2} \overline{\text{CZ}}[(\beta_2; i),(\gamma_2; j)]^{ {|\alpha_2 \cup \beta_2 \cup \gamma_2|}},
\end{align}
which is an addressable logical CZ gate due to the 1-form symmetryn property as explained in Eq.~\ref{sec:parallelizable}.   The triple intersection condition is also explicitly derived from the cup product of the symmetry operator.  
 
}

\section{Scaling of the code parameters for homological quantum codes defined on manifolds}\label{sec:scaling}

{
In this section, we analyze the connection between the   systolic geometry of a manifold and the code parameters of the corresponding homological qLDPC code.  In particular, we show that the encoding rate and code distance corresponds to the Betti number and systoles of the underlying manifolds. 
}

\subsection{Review of hyperbolic geometry and codes}\label{sec:review_hyperbolic}
We first review the hyperbolic geometry in two dimensions, as well as the code parameters of the corresponding 2D hyperbolic codes.

For 2-manifolds, one has the Gauss-Bonnet formula:
\be\label{eq:Gauss1}
\frac{1}{2\pi}\int \kappa dA = \chi(\M^2) = 2-2g.
\ee
In the case of hyperbolic manifolds, we have the curvature $\kappa = -1$, which gives rise to 
\be\label{eq:Gauss2}
A=4\pi(g-1).
\ee
Note that the 1st Betti number is related to genus by $b_1(\M^2) =2g$.
Since in 2D, the volume  equals the area $A$, we have the following equality for the hyperbolic 2-manifold
\be
b_1(\M^2) = O(vol(\M^2)),
\ee
according to Eq.~\eqref{eq:Gauss2}.
Since the total number of qubits $n$ is proportional to the area $A$, and the number of logical qubits also scales as $k=2g$, we obtain the constant encoding rate for 2D hyperbolic codes:
\be
k/n = \text{const}.
\ee

The distance of the 2D hyperbolic code and more generally homological quantum code is determined by the \textit{systoles} of the $D$-manifold $\M^D$.   We define the $i$-systole, denoted by $sys_i(\M^D; \ZZ_2)$, as the length of the  shortest non-contractible $i$-cycle (over $\ZZ_2$ coefficients) $c_i \in H_i(\M^D; \ZZ_2)$.    For $D$$=$$2$, both the $Z$ and $X$ distance is proportional to the $\ZZ_2$ 1-systole i.e., $d_Z$$\sim$$d_X$$\sim$$sys_1(\M^2; \ZZ_2)$.
For $D$$=$$3$,   the $Z$ and $X$ distance equals to the 1-systole and 2-systole respectively, i.e., $d_Z=sys_1(\M^3; \ZZ_2)$ and $d_X=sys_2(\M^3; \ZZ_2)$.  The overall distance is the smaller one of them: $d=\min (d_Z, d_X)$.

The lower bound of the distance and 1-systole can be obtained from the \textit{injectivity radius}. The injectivity radius $R$ of a $D$-dimensional manifold $\M^D$ is the greatest length such that any geodisic $D$-ball $B^D$ of raidus $R$ can be embedded anywhere in $\M^D$.  It has been shown that for $D$-dimensional arithmetic hyperbolic manifold $\M^D_h$ \footnote{For the definition and construction of the $D$-dimensional arithematic manifold, see Ref.~\cite{Guth:2014cj} for details.}, one has $R \ge c \log vol(\M^D_h)$ with some constant $c>0$ \cite{Guth:2014cj}.    Intuitively, the logarithmic injectivity  radius can also be understood with the fact that the tessellation of the hyperbolic manifolds has a tree-like expansion structure and hence it only takes $O(\log vol(\M^D_h))$ steps to wrap around the manifolds.

For the 1-systole of $\M^D_h$, we can obtain the following lemma: 
\begin{lemma}$sys_1(\M^D_h; \ZZ_2) \ge 2R \ge c' \log vol(\M^D_h)$ for some constant $c'>0$. 
\label{lem:1-systole}
\end{lemma}

\begin{proof}
One can prove this lemma by contradiction.   If $sys_1(\M^D_h; \ZZ_2)<2R$, then there exists a hyperbolic $D$-ball with radius $R$ such that a non-contractible 1-cycle with length less then the diameter ($2R$) of the $D$-ball  can wrap around $\M^D_h$ while it can still be completely embedded  inside this $D$-ball.   This is impossible since a $D$-ball cannot have a non-contractible 1-cycle and hence such $D$-ball cannot be embedded in $\M^D_h$.   Since we have $R \ge c \log vol(\M^D_h)$, we can obtain the second inequality in the lemma. 
\end{proof}
For $D$$=$$2$, i.e., 2D arithemetic hyperbolic manifold $\M^2_h$, we can obtain the lower bound of the 1-systole 
\be
 sys_1(\M^2_h; \ZZ_2) \ge c' \log vol(\M^2_h). 
\ee 
 This is equivalent to the bound on the code distance of the 2D hyperbolic code based on arithmetic manifold
 \be
 d \ge c'\log n,
 \ee
 where $n$ is the total number of qubits.

\subsection{Betti number and encoding rate scaling  in three and higher dimensions}

We now consider Betti number scaling and encoding rate in higher dimensions. 
In $2m$ (even) dimensions, one has the Chern-Gauss-Bonett theorem: 
\be
\chi(\M^{2m})= \frac{1}{(2\pi)^{D/2}} \int Pf(\Omega),
\ee
where $\chi(\M^{2m})$ is the Euler characteristics and $Pf(\Omega)$ is the Pfaffian of the curvature form of the Levi-Civita connection.  Since the Euler characteristics can be expressed as an alternating sum of the $i^\text{th}$ Betti number as  $\chi(\M^{2m})=\sum_{i=0}^{2m}(-1)^i b_i$, it is in general difficult to obtain a simple scaling between the $i^\text{th}$ Betti number $b_i$ and the hyperbolic volume in dimension higher than two from this formula.     The only exception is in 4D, where the 2nd Betti number is lower bounded by the volume in Ref.\cite{Guth:2014cj}, leading to the corresponding $(2,2)$-4D hyperbolic code \footnote{(2,2) means both the logical-X and logical-Z operators are supported on 2-cycles.} having constant encoding rate similar to the case of 2D hyperbolic code. Furthermore, in odd dimensions, the Chern-Gauss-Bonnet theorem simply does not apply.

On the other hand, independent of even or odd dimensions, there exists an upper bound for the sum of $i^\text{th}$ Betti numbers according to Gromov's theorem for a $D$-dimensional manifold \cite{gromov1982volume}:
\be\label{eq:Gromov}
\sum_{i=0}^{D} b_i(\M^D) \le C_D \cdot vol(\M^D),
\ee
where $b_i(\M^D)$ is the $i^\text{th}$ Betti number,  $C_D$ is some constant only depending on the dimension $D$, and $vol(\M^D)$ represents the volume of the $D$-dimensional manifold.  

Although the proof of the inequality Eq.~\eqref{eq:Gromov} in Gromov's original work \cite{gromov1982volume} is quite involved, we can show the asymptotic scaling (in the limit $vol(\M^D) \rightarrow  \infty$)
\be\label{eq:Gromov_asymptotic}
 \sum_{i=0}^{D} b_i(\M^D) \le O(vol(\M^D)), 
\ee
in a quite simple way from an information-theoretic perspective.  For a given $D$-manifold, one can define $D-1$ copies of homological quantum codes (with $\lfloor D/2 \rfloor$ distinct types due to Poincaré duality), where the physical qubits are assigned to the $i$-cells  and the code spaces are associated with $\mathcal{H}_{(i)}= \mathbb{C}^{|H_i(\mathcal{M}^D; \ZZ_2)|}$ repsectively ($i=1,2, \cdots, D-1$).   Consider the case that $\M^D$ only has one connected component, we have $b_0(\M^D)=b_D(\M^D)=1$.   Therefore $\sum_{i=0}^{D} b_i(\M^D)-2$ is just the total number of logical qubits in these $D-1$ copies of homological quantum codes.   On the other hand,  $vol(\M^D)$ is just proportional to the total number of physical qubits.  Therefore, we see that the asymptotic scaling of the Betti number has to be at most linear with the volume as suggested in Eq.~\eqref{eq:Gromov_asymptotic}, simply because one cannot encode more logical qubits than the number of physical qubits. Namely, a Betti number scaling like $O(vol(\M^D)^{1+\epsilon})$ for any positive $\epsilon$ is impossible.  

We hence see a natural and profound connection between the Betti number scaling and the quantum information  storage capacity per volume, which gives the information-theoretic interpretation of Gromov's bound Eq.~\eqref{eq:Gromov}.

In the $D=3$ case which is the focus of this paper, the number of logical qubit equals to the 1st/2nd Betti number  $k=b_1(\M^3)=b_2(\M^3)$ (the second equality is due to Poincaré duality).  We also have $b_0(\M^3)=1$ if $\M^3$ is simply connected.  We then have the following statement:
\\

\nin \textbf{Definition 1.} \textit{If the upper bound $b_1(\M^3) \le O(vol(\M^3))$ is saturated and the corresponding  manifold $\M^3$ has 1-systole and 2-systole growing with its volume, we say the corresponding 3-manifold $\M^3$ has the maximal asymptotic quantum information storage capacity per volume.}

Note that we require the systole to grow with the volume since we need the code distance $d$ to diverge in the asymptotic limit such that the logical error rate $p_L \propto (p/p_\text{th})^{\lfloor (d+1)/2 \rfloor}$ is exponentially suppressed to zero when the physical error rate $p$ is below the threshold $p_\text{th}$ and quantum memory remains fault-tolerant.   

In fact it is extremely difficult to find 3-manifolds with such optimal Betti number scaling and hence quantum information storage capacity.   Nevertheless, in Sec.~\ref{sec:generic_3-manifold} we will present a specific family of 3D hyperbolic manifolds based on Trorelli mapping torus construction which reaches such linear Betti number scaling.  Although we are not currently able to give a rigorous lower bound on the systoles of such 3-manifold, we believe the systoles will grow with the volume poly-logarithmically .

\begin{figure*}[hbt]
  \includegraphics[width=1.4\columnwidth]{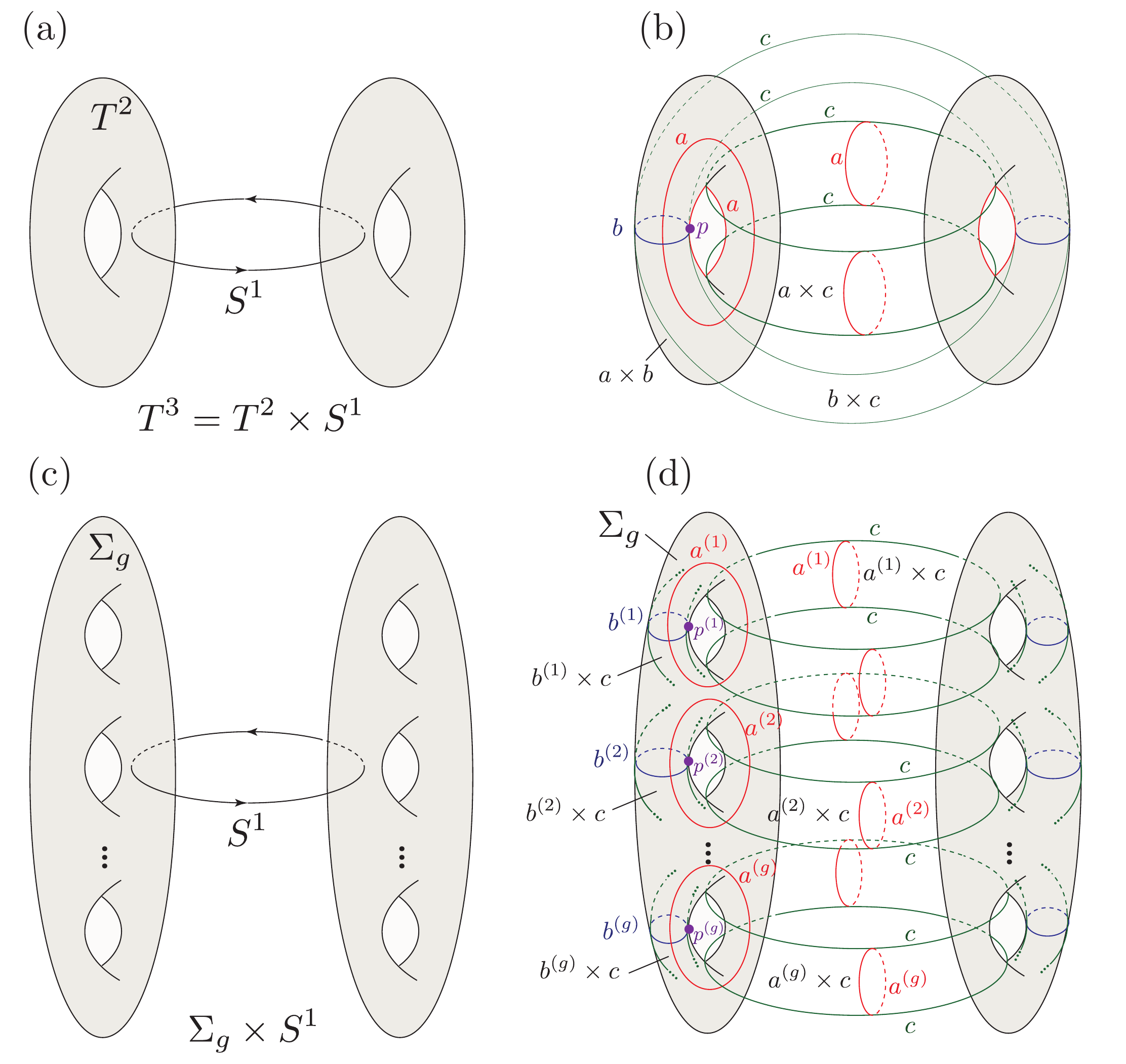}
  \caption{(a) Illustration of the 3-torus $T^3$ as a product of a 2-torus $T^2$ and a circle $S^1$.  (2) Three torus submanifold (2-cycles) $a \times c$, $b \times c$ and $a \times b$ have a triple intersection point at $p$.  (c)  Illustration of the product manifold $\Sigma_g \times S^1$.  (d) Submanifolds (2-cycles) $\Sigma_g$, $a^{(i)}\times c$ and $b^{(i)}\times c$ have a triple-intersection at point $p^{(i)}$.  The total triple intersection number is $g$.  }
  \label{fig:three-manifold_examples}
\end{figure*}

\section{Code construction and logical gate structure for 3-manifolds}\label{sec:code_construction}

{
In this section, we investigate three classes of homological codes defined on 3-manifolds, which can all be viewed as a fibre-bundle construction or more specifically a mapping torus construction. We further investigate the triple intersection invariants and the corresponding logical CCZ and CZ gates in these codes. 

In Sec.~\ref{sec:quasi-hyperbolic},  we first introduce the simplest construction: the quasi-hyperbolic code, which is formed by the product of a 2D hyperbolic code and a circle with logarithmic length. The code has almost-constant rate (up to logarithmic reduction) and logarithmic distance. This can be viewed as a fibre-bundle with a trivial (identity) twist.  
In Sec.~\ref{sec:fibre_bundle}, we consider a fibre-bundle construction first proposed by Freedman, Meyer and Luo \cite{Freedman_systole_2002}, where we take a twisted product of the 2D hyperbolic code and a circle with unit length, aiming to reduce the logarithmic overhead. We also study the triple intersection property in the presence of the twist.  Nevertheless, we notice that the presence of the twist can reduce the homology which still leads to an almost-constant rate with a logarithmic reduction, while the distance also has a logarithmic lower bound.  In Sec.~\ref{sec:generic_3-manifold}, we introduce a specific family of 3D hyperbolic codes called the Torelli mapping torus code, which are constructed from mapping tori (fibre bundle) of a pseudo-Anosov element in the Torelli subgroup. The pseudo-Anosov twist in this case does not lead to the reduction of homology and one hence maintain a constant rate, while the lower bound of the distance remains open.     
}

\subsection{Quasi-hyperbolic codes}\label{sec:quasi-hyperbolic}
\subsubsection{Code construction and code parameters}

We first construct a family of simple 3D quasi-hyperbolic codes whose code parameters are easy to bound and almost approach the desired optimal parameters for 3D homological LDPC codes up to logarithmic correction and already quite useful for practical purposes.

We start with a family of 2D arithmetic hyperbolic genus-$g$  surfaces $\Sigma_g$ ($g\ge 2$).  Now we take a product of $\Sigma_g$ with a circle $S^1$ and obtain a 3-manifold $\Sigma_g \times S^1$, as illustrated in Fig.~\ref{fig:three-manifold_examples}.  One can view this 3-manifold as a trivial fibre bundle, where the base is $S^1$ and the fibre is the genus-$g$ surface $\Sigma_g$.  We note that $\Sigma_g \times S^1$ is not hyperbolic, but the fibre bundle can be generalized to the mapping torus construction which can produce hyperbolic manifolds as will be discussed in the next section.  We call $\Sigma_g \times S^1$ a quasi-hyperbolic manifold, and the color code or surface code defined on that a quasi-hyperbolic code.  

Now we investigate the homologies of the $\Sigma_g \times S^1$.  For illustrative purpose, we start with the $g=1$ case, i.e., $\Sigma_{g=1} \times S^1=T^2\times S^1=T^3$ is a 3-torus, which is Euclidean (zero curvature).  The homology group of the 3-torus is $H_1(T^3)=H_2(T^3)=\ZZ_2 \oplus \ZZ_2 \oplus \ZZ_2 \equiv \ZZ_2^3$.  One choice of the 1st homology basis is $B_1=\{a,b,c\}$ as illustrated by Fig.~\ref{fig:three-manifold_examples}(b), where $a,b$ are homology cycles on the 2-torus $T^2$ and $c$ is the homology cycle going through the circle $S^1$.   The dual  2nd homology basis is $B_2=\{b\times c, a\times c, a\times b \}$ as illustrated in Fig.~\ref{fig:three-manifold_examples}(b), where we have the $i^\text{th}$ entry in the 1st homology basis being the Poincaré dual of the $i^\text{th}$ entry in the 2nd homology basis, due to the intersection condition $|a \cap (b\times c)|=|b \cap (a\times c)|=|c \cap (a\times b)|=1$ meaning that the dual pairs intersect at a single point.  One can see that each 2-cycle in $B_2$ is a 2-torus.

\begin{figure*}[hbt]
  \includegraphics[width=1.6\columnwidth]{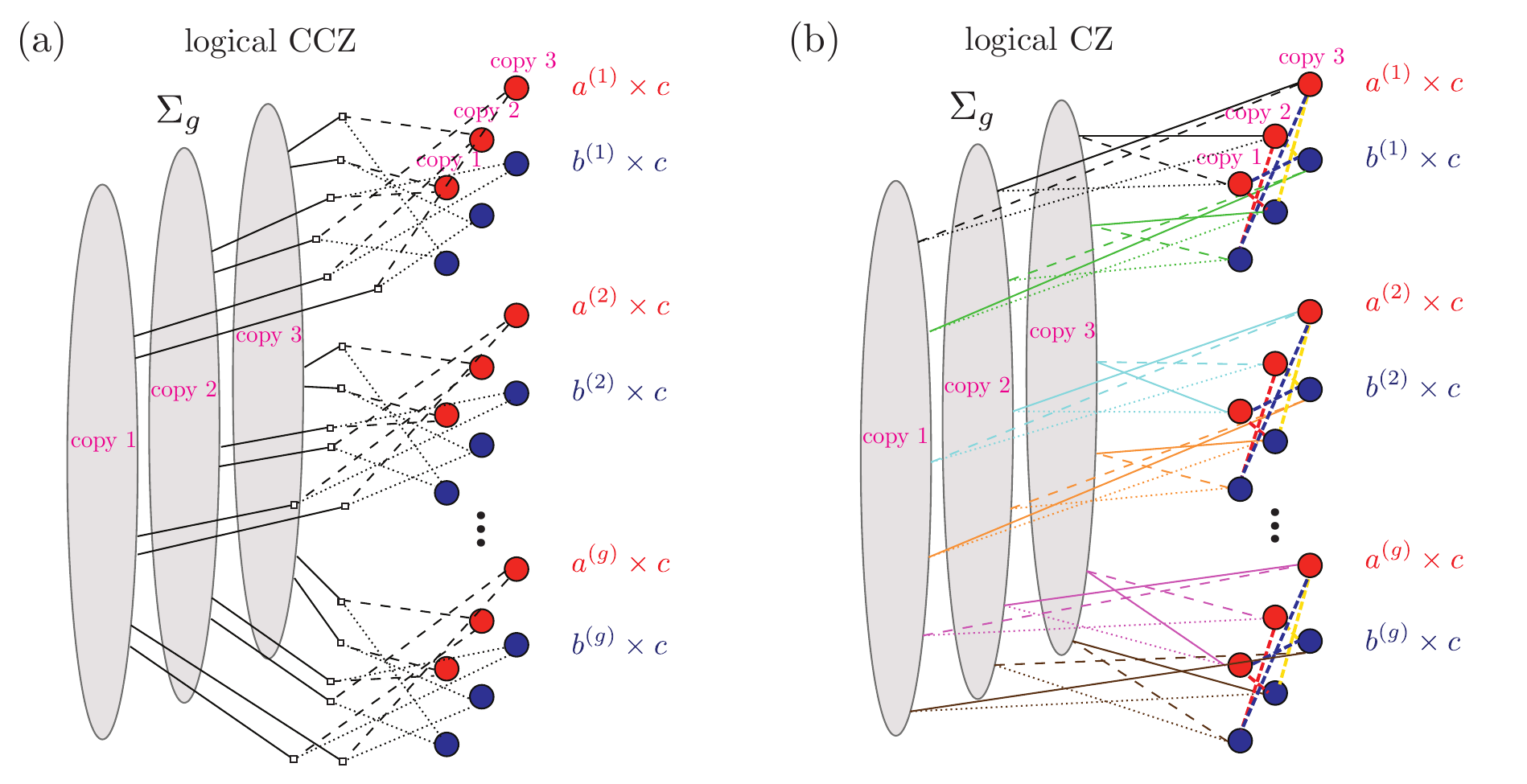}
  \caption{(a) The interaction hypergraph of the quasi-hyperbolic code. Each hyperedge (shown as a tri-junction) is connected to three vertices (represented by circles or ellipses), which corresponds to a logical CCZ gate.   (b) The colored interaction graph of the logical CZ gate.  Colored edge with the same type (represented by color and line type) corresponds to logical CZ gates being simultaneously applied.  For all edges connected to the logical qubits on $\Sigma_g$, each type only has two edges.  }
\label{fig:interaction_hyper-graph}
\end{figure*}

For $g \ge 2$, the  genus-$g$ surface is topologically equivalent to a connected sum of $g$ tori, therefore the  homology group of the surface can be decomposed as
\begin{align}
\non H_1(\Sigma_g; \ZZ_2)=& H_1(T^2 \# T^2...\#T^2; \ZZ_2)\\
\non =&H_1(T^2;\ZZ_2) \oplus H_1(T^2;\ZZ_2) ... \oplus H_1(T^2;\ZZ_2) \\
 =& \ZZ_2^2 \oplus \ZZ_2^2  \dots \oplus \ZZ_2^2 =\ZZ_2^{2g}.
 \end{align}
 We pick the following 1st homology basis on the genus-$g$ surface: $\{a^{(i)},b^{(i)} |  i=1,2, \dots, g\}$, where $a^{(i)}$ and $b^{(i)}$ are homology cycles from each torus participating in the connected sum.   The $g=2$ case is illustrated in Fig.~\ref{fig:three-manifold_examples}(d).  The non-contractible 1-cycle along the direction of $S^1$ is again called $c$.   The 1st-homology of the 3-manifold can be obtained from the Künneth formula: 
\begin{align}\label{eq:Künneth}
\non &H_1(\Sigma_g\times S^1; \ZZ_2) \\
\non =&[H_1(\Sigma_g; \ZZ_2)\otimes H_0(S^1; \ZZ_2)] \oplus [H_0(\Sigma_g; \ZZ_2)\otimes H_1(S^1; \ZZ_2)] \\
=& \ZZ_2^{2g} \oplus \ZZ_2 = \ZZ_2^{2g+1}
\end{align}
We hence obtain the Betti number $b_1$ and number lof logical qubit $k$ as 
\be
k=b_1(\Sigma_g\times S^1)=Rank(H_1(\Sigma_g\times S^1; \ZZ_2))=2g+1.
\ee

We choose the following 1st homology basis for the 3-manifold: 
$B_1=\{a^{(i)},b^{(i)},c | \  i=1,2, \dots, g\}$. The dual 2nd homology basis is $B_2=\{b^{(i)}\times c,a^{(i)} \times c, \Sigma_g | \  i=1,2, \dots, g\}$, which satisfies the following intersection condition with the 1-cycles in $B_1$:   
$|a^{(i)} \cap (b^{(i)}\times c)|=|b^{(i)} \cap (a^{(i)}\times c)|=|c \cap \Sigma_g|=1$.  Note that the 2-cycles $b^{(i)}\times c$ and $a^{(i)}\times c$ are all 2-tori, while the 2-cycle labeled by $\Sigma_g$ is the genus-$g$ fibre surface.  

Now we discuss the code parameter scaling of this family of quasi-hyperbolic codes.   The area $A$ of the genus-$g$ surface $\Sigma_g$ is determined by the Gauss-Bonnet formula Eq.~\eqref{eq:Gauss2}, i.e., $A=4\pi(g-1)$.   Based on the injectivity radius scaling and Lemma~\ref{lem:1-systole}, we can obtain the lower bound of the $\ZZ_2$ 1-systole of $\Sigma_g$ as 
\be
sys_1(\Sigma_g; \ZZ_2) \ge c' \cdot \log A.
\ee
Now when taking the product of $S^1$, the increase of the length of $S^1$ will introduce volume/qubit overhead. On the other hand,  its increase also leads to the growth of the non-contractible 1-cycle $c$ along $S^1$ and hence the code distance.  The increase of the code distance stops when the length of $S^1$ reaches the systole of the genus-$g$ surface $sys_1(\Sigma_g; \ZZ_2)$.  Therefore, we choose the length of $S^1$ to be exactly the systole of the genus-$g$ surface, i.e., $l= O( \log A)$ in order to optimally balance the loss of the encoding rate and gain in the distance.   The code distance hence scales as $d \propto l=O(\log A)$.   The volume $V$ of the quasi-hyperbolic 3-manifold is
\be\label{scaling_quasi_1}
V = A \cdot l = 4\pi(g-1)\cdot O( \log A)=(g-1) \cdot O(\log A)
\ee
Meanwhile, we have the following identity
\be\label{scaling_quasi_2}
V=A\cdot O(\log A).
\ee
By combining Eq.~\eqref{scaling_quasi_1} and Eq.~\eqref{scaling_quasi_2}, we obtain the Betti number scaling and encoding rate as 
\be
b_1/V=O(g/V) = O\left(1/\log (V)\right) \ \Longrightarrow \  k/n=O(1/\log (n)). 
\ee
Meanwhile, the code distance scaling as
\be
l=O(\log (V)) \ \Longrightarrow \ d = O(\log (n)).
\ee

\subsubsection{Triple intersection and logical gates}\label{sec:triple_quasi-hyperbolic}

We first look at the simplest example of 3-torus illustrated in Fig.~\ref{fig:three-manifold_examples}(a,b).  The three non-contractible 2-cycles intersect as
\be
(a\times b)\cap (b\times c)\cap (c\times a)
= (a\times b) \cap c = p,
\ee
where the first equality comes from the fact that intersection of $b\times c$ and $c \times a$ is the 1-cycle $c$. Here, $p$ represents a single point where the 1-cycle $c$ intersects with the 2-cycle (torus) $a\times b$. Therefore we obtain the $\ZZ_2$ intersection number $|(a\times b)\cap (b\times c)\cap (c\times a)| = 1$.  According to Eq.~\eqref{eq:CCZ_formula}, a transversal $T$ gate leads to the following logical gate:
\be
\widetilde{T} =  \prod_{(r,s,l)} \lo{\text{CCZ}}((a\times b; r),(b\times c; s),(c\times a; l)),
\ee
 implying a global logical CCZ gates  applied among the 6 triplets of logical qubits, where $(r,s,l)$ represent labels of three different toric-code copies with arbitrary possible permutations (6 permutations in total).   According to Eq.~\eqref{eq:CZ_formula}, parallelizable transversal CZ gate can be applied along the 2-cycle $a\times b$ between toric-code copy $r$ and $s$ ($r \neq s$),  which effectively gives rise to the following logical gate
\be
\widetilde{\text{CZ}}^{(r,s)}_{a\times b} =  \lo{\text{CZ}}((b \times c; r),(c \times a; s)) \lo{\text{CZ}}((b \times c; s),(c \times a; r)).
\ee
One can apply transversal CZ gate along the other 2-cycles, and the results can be obtained similar by simple permutations.  

We then look at the general genus-$g$ case.  We check the intersection between the triplet:
\be
\Sigma_g \cap (a^{(i)} \times c) \cap (b^{(i)} \times c)= \Sigma_g \cap c = p^{(i)},
\ee 
where they intersect at a single point $p^{(i)}$. There are in total $g$ such intersecting triplets, while there is no other triplet which has non-trivial intersection.  
 A transversal $T$ gate leads to the following logical gate:
\be
\widetilde{T} =  \prod_{(r,s,l)} \lo{\text{CCZ}}((\Sigma_g; r),(a^{(i)} \times c; s),(b^{(i)} \times c; l)),
\ee
where $(r,s,l)$ represent labels of three different toric-code copies with arbitrary possible permutations.

The structure of the logical CCZ gates is encoded in the interaction hypergraph introduced in Sec.~\ref{sec:interaction_hypergraph}, as illustrated in Fig.~\ref{fig:interaction_hyper-graph}(a).  In this figure, each circle or ellipse corresponds to a vertex in the interaction hypergraph and represents a logical qubit.  In particular, the three long ellipses represent logical qubits labeled by the genus-$g$ surface $\Sigma_g$ on the three toric-code copies, while the triplet of red (blue) circles represent logical qubits labeled by  $a^{(i)}\times c$ ($b^{(i)}\times c$) on the three copies respectively.  The vertices are connected by hyperedges composed of lines with three types (solid, dashed, dotted) meeting at a tri-junction, which represent the triple intersection between the three logical-$X$  membranes (2-cycles).    Each hyperedge corresponds to a logical CCZ gate acting on the three connected logical qubits.

Now we consider applying parallelizable logical CZ gates. The logical structure is again encoded in the interaction hypergraph shown in Fig.~\ref{fig:interaction_hyper-graph}(a).  Each hyperedge connects three logical qubits. If one applies the transversal gate 
$\widetilde{\text{CZ}}^{(r,s)}_{\alpha_2}$ between copy $r$ and $s$ on the membrane $\alpha_2$ labeling one logical qubit, the other pairs of logical qubits connected by the same hyperedge will be acted by the logical CZ gate at the same time.

For the quasi-hyperbolic code, it is wiser to apply transversal CZ gate only along the sub-manifold (2-torus) $a^{(i)}\times c$ (or $b^{(i)}\times c$) since it only has triple intersection with $b^{(i)}\times c$ (or $a^{(i)}\times c$) and $\Sigma_g$. According to Eq.~\eqref{eq:CZ_formula}, when applying transversal CZ along the sub-manifold $a_i \times c$ between toric code copy $r$ and $s$, it gives rise to the following logical gate:   
\be\label{eq:CZ_gate_quasi_hyperbolic_n1}
\widetilde{\text{CZ}}^{(r,s)}_{a^{(i)}\times c} =  \lo{\text{CZ}}((b^{(i)} \times c; r),(\Sigma_g; s)) \lo{\text{CZ}}((b^{(i)} \times c; s),(\Sigma_g; r)).
\ee
Similarly, one can obtain
\be\label{eq:CZ_gate_quasi_hyperbolic_n2}
\widetilde{\text{CZ}}^{(r,s)}_{b^{(i)}\times c} =  \lo{\text{CZ}}((a^{(i)} \times c; r),(\Sigma_g; s)) \lo{\text{CZ}}((b^{(i)} \times c; s),(\Sigma_g; r)).
\ee
The above expressions suggest that the logical qubit labeled by the genus-$g$ surface $\Sigma_g$ (or equivalently the dual 1-cycle $c$) can serve as a global entangling bus, such that the logical qubits labeled by the 2-tori $b^{(i)}\times c$  and $a^{(i)}\times c$ (equivalently labeled by the dual 1-cycle $a^{(i)}$ and $b^{(i)}$) can all be separately entangled with this global entangling bus.  Besides, there are also collective logical $CZ$ gates when applying the transversal $CZ$ operator on the $\Sigma_g$ surface and between toric-code copy $r$ and $s$:
\be
\widetilde{\text{CZ}}^{(r,s)}_{\Sigma_g} = \prod_i \lo{\text{CZ}}((a^{(i)} \times c; r),(b^{(i)} \times c; s)).
\ee

The corresponding logical CZ gate structure is also encoded in the \textit{interaction graph} shown in Fig.~\ref{fig:interaction_hyper-graph}(b).  This interaction graph is composed of vertices (circles and ellipses) corresponding to the logical qubits and colored edges corresponding to the logical CZ gate between the two connected two vertices (logical qubits).  The logical CZ gates on the edges with same color type will have to be applied simultaneously instead of independently.  In Fig.~\ref{fig:interaction_hyper-graph}(b), we use the combination of actual color and lines types (due to the limit of usable colors) to represent a particular type of colored edge.    Note that for all the shown logical CZ gates in the quasi-hyperbolic code, there is always a pair of colored edges of the same type which corresponds to logical CZ gates (coupling the same two toric-code copies) applied simultaneously, as can also be seen clearly in Eq.~\eqref{eq:CZ_gate_quasi_hyperbolic_n1} and Eq.~\eqref{eq:CZ_gate_quasi_hyperbolic_n2}.
Since there are only constant number of colored edges with the same type connected to logical qubits on $\Sigma_g$, the corresponding logical CZ gates are parallelizable and can be individually addressed, which is again due to the fact that they correspond to higher-form symmetries.

We conclude this section with the following theorem:
\begin{theorem}\label{theorem:quasi-hyperbolic}
There exists a family of quasi-hyperbolic color code with encoding rate $O(1/\log(n))$ and distance  $d$$=$$O(\log(n))$,  which supports fault-tolerant logical non-Clifford gates and parallelizable logical Clifford gates.
\end{theorem}

\subsection{Homological fibre bundle code}\label{sec:fibre_bundle}

\subsubsection{Code construction and code parameters}
The construction of the quasi-hyperbolic codes mentioned above can be viewed as a trivial case of a fibre bundle, where the genus-$g$ surface $\Sigma_g$ is the fibre and the circle $S^1$ is the base.  The length of the base circle scales as $\log(V)$, which increases the volume and qubit overhead by a factor of $\log(V)$.  One could make the base circle having unit length and hence not increasing the volume, but the systole and hence code distance will also become constant.  Now we consider implementing a twist along the fibre following the construction of the Freedman-Meyer-Luo 3-manifolds \cite{Freedman_systole_2002}, such that the systole (distance) will be greatly increased and scales with the volume even though the base circle has unit length.   We call the codes defined on such twisted manifold the homological fibre bundle codes, and those defined on the 4-colourable cellulation of the manifold the fibre bundle color codes.   As we will see, the twist also slightly reduces the Betti number of the manifold, and hence makes the encoding rate non-constant also.   This issue will be resolved in the more general mapping torus construction discussed in Sec.~\ref{sec:generic_3-manifold}.   Nevertheless, the homological fibre bundle code is still an interesting and inspiring example for stuyding the homology, systole, triple intersection and logical gates in more general 3-manifold.  It could also be easier for experimental implementation since its lattice can be embedded in constant number of planes (without crossing) connected with twisted vertical links. 

We first review the construction of such 3-manifolds. We start with the choice of the fibre being a particular family of arithmetic hyperbolic 2-manifolds $\Sigma_g$ which are the quotient of the infinite hyperbolic plane $\mathbb{H}^2$ by a Fuchian group. We consider the co-compact torsion free Fuchscian group $\Gamma_{(-1,p)}$ derived from the type $\frac{-1,p}{\mathbb{Q}}$ quaternion algebra with the prime number $p= 3 \ \text{mod} \ 4$.  One can explicitly write out the group $\Gamma_{(-1,p)}$ as 
\begin{align}
\non \Gamma_{(-1,p)} =& \bigg\{ 
A=\begin{pmatrix}
a+b\sqrt{p}, &  -c + d\sqrt{p} \\
c + d\sqrt{p}, & a-b\sqrt{p} \\
\end{pmatrix}:  a,b,c,d \in \mathbb{Z}, \\
 & \ det(A)=1
\bigg\} \bigg{/} \pm id.
\end{align}
For integers $N > 2$, we have the normal subgroup of $\Gamma_{(-1,p)}$:
\begin{align}
\non \Gamma_{(-1,p)}(N) =& \bigg\{ 
A=\begin{pmatrix}
1+N(a+b\sqrt{p}), &  N(-c + d\sqrt{p}) \\
N(c + d\sqrt{p}), & 1+N(a-b\sqrt{p}) \\
\end{pmatrix}  \\
&: a,b,c,d \in \mathbb{Z},
\ det(A)=1
\bigg\} \bigg{/} \pm id,
\end{align}
with bounded index 
\be\label{eq:subgroup_index}
O(N^2) \le [\Gamma_{(-1,p)}: \Gamma_{(-1,p)}(N)] \le O(N^3).
\ee
We hence obtain a family of closed arithmetic hyperbolic surfaces $\Sigma_g=\mathbb{H}^2/\Gamma_{(-1,p)}(N)$.  

As has been shown in Ref.~\cite{Freedman_systole_2002}, the class of closed hyperbolic surfaces $\Sigma_g=\mathbb{H}^2/\Gamma_{(-1,p)}(N)$ satisfies the following three properties:
\begin{enumerate}[(1).]
    \item The smallest eigenvalue of the Laplacian $\lambda_1(\Sigma_g) \ge c_1$,
    \item There exits an isometry $\tau : \Sigma_g \rightarrow \Sigma_g$, with order $|\tau| \ge c_2 (\log g)^{\frac{1}{2}}$,
    \item
    The map $\Sigma_g \rightarrow \Sigma_g / \langle\tau\rangle $ is a covering projection to the base surface $_gS =: \Sigma_g / \langle \tau \rangle$ (of genus  $< g$), which has injectivity radius $R(_gS) \ge c_3 (\log g)^{\frac{1}{2}}$.
\end{enumerate}
where $c_i$ ($i$=1,2,3) are positive constants independent of $g$.   We now demonstrate the above properties in the following paragraphs.

Using the property of the Fuchscian group, it has been proved in Ref.~\cite{Schaller1995ExtremalRS} that the injectivity radius of $\Sigma_g=\mathbb{H}^2/\Gamma_{(-1,p)}(N)$ is:
\be\label{eq:raidus_Fuchscian}
R(\mathbb{H}^2/\Gamma_{(-1,p)}(N)) = O(\log N).
\ee
The area of $\Sigma_g=\mathbb{H}^2/\Gamma_{(-1,p)}(N)$ is related to the area of $\mathbb{H}^2/\Gamma_{(-1,p)}$ via
\be
A(\mathbb{H}^2/\Gamma_{(-1,p)}(N)) = A(\mathbb{H}^2/\Gamma_{(-1,p)})\cdot [\Gamma_{(-1,p)}: \Gamma_{(-1,p)}(N)],
\ee
since both surfaces are obtained from compactifying $\mathbb{H}^2$ by quotienting out the corresponding isometry groups. Combined with the Gauss-Bonnet theorem  Eq.~\eqref{eq:Gauss2}, we know that the genus $g(N) \equiv genus(\mathbb{H}^2/\Gamma_{(-1,p)}(N))$ is proportional to the index $[\Gamma_{(-1,p)}: \Gamma_{(-1,p)}(N)]$. Using Eq.~\eqref{eq:subgroup_index}, we obtain 
\be\label{eq:genus_scaling}
O(N^2) \le g(N) \le O(N^3).   
\ee
Combined with Eq.~\eqref{eq:raidus_Fuchscian}, one can obtain the injectivity radius of $\Sigma_g$ as
\be\label{eq:injectivity_radius_vs_g}
R(\Sigma_g) = O(\log g).
\ee

We now define two sequence of natural numbers $J_k \equiv k$ and $N_k \equiv k^{\lfloor\log (k) \rfloor}$ with $k$$=$$1,2, \cdots, \infty$, which leads to $J_k|N_k$ ($J_k$ divides $N_k$). We further define $g_k$$\equiv$$g(N_k)$$\equiv $$genus(\mathbb{H}^2/\Gamma_{(-1,p)}(N_k))$ and $h_k$$\equiv$$h(J_k)$$\equiv $$genus(\mathbb{H}^2/\Gamma_{(-1,p)}(J_k))$.
According to Eq.~\eqref{eq:genus_scaling}, we can obtain 
\be
\log(g_k) = O(\log(N_k)) \ \text{and} \ \log(h_k) = O(\log(J_k)),
\ee
which leads to 
\begin{align}\label{eq:relation_g_h}
\non \log(g_k) &= O(\log (N_k)) = O\left(\log(k^{\lfloor \log (k) \rfloor})\right) \\
&=O(\log(k)^2) = O(\log (h_k)^2).
\end{align}
We can then use $h_k$ to construct an isometry with large order and translation length, which satisfies property (2) and (3) respectively.  

Since $J_k|N_k$, we have $\Gamma_{(-1,p)}(N_k) \lhd \Gamma_{(-1,p)}(J_k)$.  This leads to the fact that $\Sigma_{g_k} = \mathbb{H}^2/\Gamma_{(-1,p)}(N_k)$ is a \textit{covering space} of  $\Sigma_{h_k} = \mathbb{H}^2/\Gamma_{(-1,p)}(J_k)$, which defines the covering projection map:
\be\label{eq:projection_map}
\Sigma_{g_k} \rightarrow \Sigma_{h_k}.
\ee
We denote the shortest non-contractible closed simple curve in $\Sigma_{h_k}$ as $\alpha$, we hence have 
\be\label{eq:length_alpha}
length(\alpha) = \Omega(\log (h_k)), 
\ee
according to Eq.~\eqref{eq:injectivity_radius_vs_g}. Choose a base point $q \in \Sigma_{g_k}$ for $\alpha$, and then lift $\alpha$ to $\tilde{\alpha}$ in $\Sigma_{g_k}$.
We now introduce the isometry $\tau_{k}: \Sigma_{g_k} \rightarrow \Sigma_{g_k}$, satisfying $\tau_{k} \in \Gamma_{(-1,p)}(J_k)$.  The isometry $\tau_{k}$ connects both ends of $\tilde{\alpha}$.  The concatenated curve $\tilde{\alpha} (\tau_{k}\tilde{\alpha})(\tau_{k}^2\tilde{\alpha})\cdots (\tau_{k}^{n-1}\tilde{\alpha})$ is a closed \textit{geodisic loop}, which is non-contractible and hence to be longer than the injectivity radius of $\Sigma_{g_k}$.  This implies that 
\be
|\tau_{k}|\cdot length(\tilde{\alpha}) \ge R(\Sigma_{g_k})=O(\log g_k),
\ee
where the equality is due to Eq.~\eqref{eq:injectivity_radius_vs_g} and $|\tau_{k}|$ represents the order of $\tau_{k}$.
According to Eq.~\eqref{eq:relation_g_h} and  Eq.~\eqref{eq:length_alpha}, one hence obtains 
\be
length(\tilde{\alpha}) = length (\alpha) = O(\log (h_k)) = O(\log^{\frac{1}{2}}(g_k)), 
\ee
which givies rise to property (2): 
\be\label{eq:isometry_order_bound}
|\tau_{k}|= \Omega(\log^{\frac{1}{2}}(g_k)).
\ee
Intuitively one can understand this fact in the following way: since the length of the curve $\alpha$ is $O(\log^{\frac{1}{2}}(g_k))$, one needs to travel at least $\Omega(\log^{\frac{1}{2}}(g_k))$ times before reaching the non-contractible geodisc loop with length $\Omega(\log  g_k)$. Note that when comparing with properties (2) and (3), we can suppress the index $k$ since the sequence of $k$ generates the sequence of genus $g$. We should also keep in mind that only a subsequence of $\Sigma_g$ generated by the sequence of $k$ is under our consideration rather than choosing arbitrary integer $g$.

We now factor the previous covering Eq.~\eqref{eq:projection_map} as
\be
\Sigma_{g_k} \rightarrow  \Sigma_{g_k}/\langle \tau_k \rangle  \rightarrow \Sigma_{h_k}.
\ee
We define $\Sigma_{g_k}/\langle \tau_k \rangle=: {_gS}$.  Since $_gS$ covers $\Sigma_{h_k}$, we have obtained property (3), i.e.,
\be\label{eq:injectivity_radius_gS}
R(_gS) \ge R(\Sigma_{h_k}) \ge O(\log h_k) = O(\log^\frac{1}{2} (g_k)).
\ee
From now on, we drop all the index $k$ for conciseness.

In the above, we have demonstrated the properties of the hyperbolic surfaces $\Sigma_g$. We now construct a fibre bundle with $\Sigma_g$$=$$\mathbb{H}^2/\Gamma_{(-1,p)}(N)$ as the fibre and a unit-length circle $S^1 \equiv [0,1]$ as the basis, which is also called a \textit{mapping torus} construction in this special case (see Sec.~\ref{sec:generic_3-manifold} for the general construction).    We can introduce the new manifold via a twisted product: 
\begin{align}
\non \mathcal{M}^3_g =& (\Sigma_g \times \mathbb{R})/\big((x,t) \sim (\tau x, t+1)\big) \\
\equiv &  (\Sigma_g \times [0,1])/\big((x,0) \sim (\tau x, 1)\big),
\end{align}
which is a mapping torus of the isometry $\tau$.  One can interpret this construction more intuitively as follows: one starts with the product 3-manifold $\Sigma_g \times S^1$ where the circle $S^1$ is parametrized by $t \in [0,1]$.  One then cuts the 3-manifold at $t=1^- \sim 0^-$ and twists the hyperbolic surface $\Sigma_g$ by the isometry $\tau$ described by the action $\tau x$. One then glues the twisted surface parameterized by $\tau x$ at $t=1^- $ with the original surface parameterized by $x$ at $t= 0$ which completes the construction of $\M^3_g$. Note that the final glueing procedure is just the geometric interpretation of the identification  $(x,0) \equiv (\tau x, 1)$ in the mapping torus construction.  The twisted manifold is illustrated in Fig.~\ref{fig:twisted_manifold}.

\begin{figure*}[hbt]
  \includegraphics[width=1.6\columnwidth]{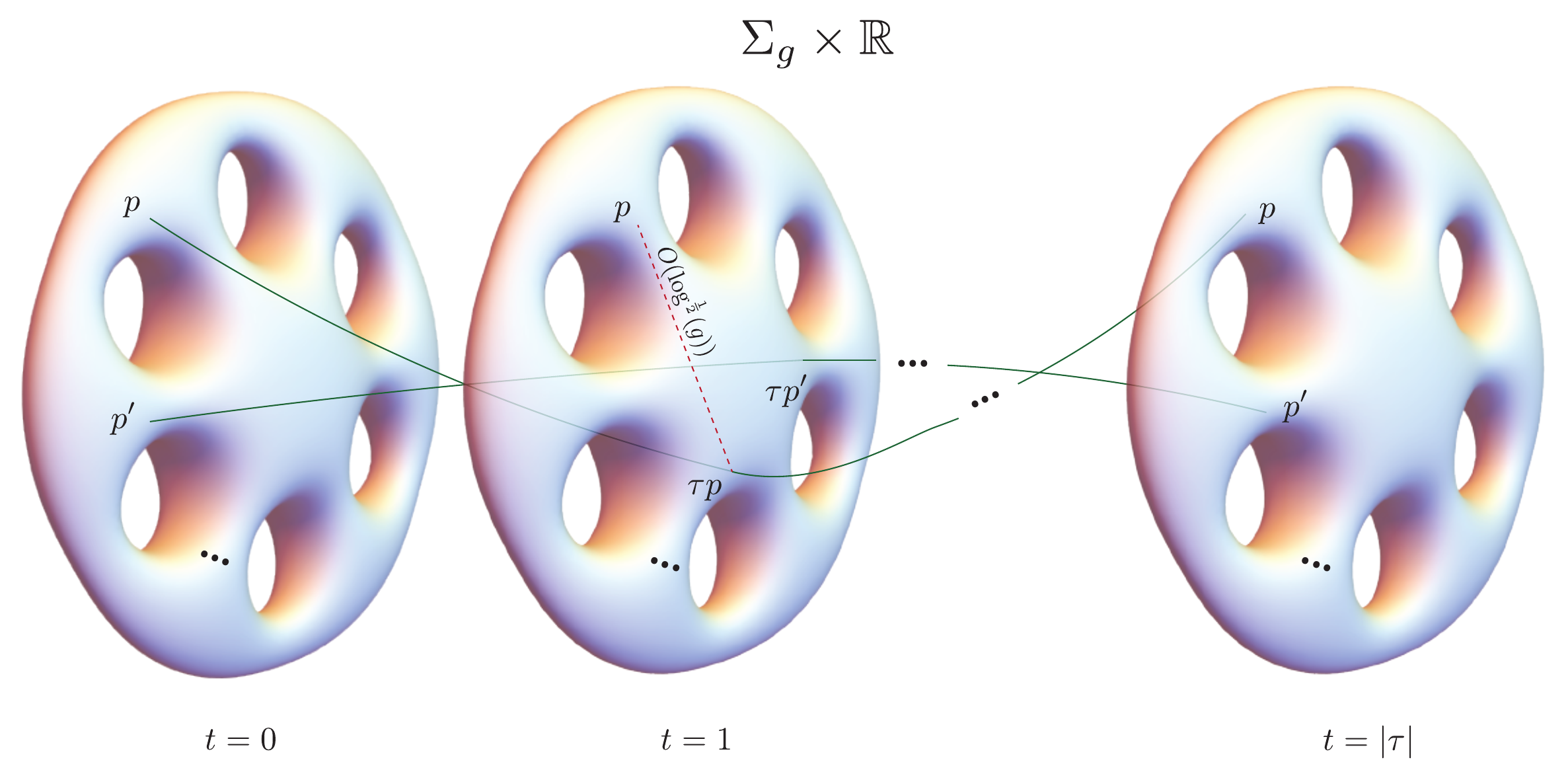}
  \caption{Illustration of the covering space $\Sigma_g \times \RR$ of the twisted manifold $\M^3_g$.  The twist $\tau$ maps a point $p$ on $\Sigma_g$ to $\tau p$ which is $O(\log^\frac{1}{2}g)$ away from its original position. It takes $|\tau|=O(\log^\frac{1}{2}g)$ steps for the point $p$ to come back to its original position.  These two conditions, as well as the 1-systole of $\Sigma_g$, make sure there is no non-contractible cycle with length smaller than $O(\log^\frac{1}{2}g)=O(\log^\frac{1}{2}V)$, which hence determines the 1-sytole of this manifold. }
  \label{fig:twisted_manifold}
\end{figure*}

Now we consider the Betti number scaling with the volume of the twisted manifold $\M^3_g$, as well as its systoles.  We state the following theorem: 
\begin{theorem}\label{theorem:3-manifold}
There exists a family of 3-manifolds $\M^3$ such that the 1st Betti number scales with the volume as $b_1(\M^3)= \Omega (V/\log^{\frac{1}{2}}(V))$  while the $\ZZ_2$ 1-systoles and 2-systoles grow with the volume as $sys_1(\M^3; \ZZ_2)=\Omega(\log^{\frac{1}{2}}(V))$ and $sys_2(\M^3; \ZZ_2)=\Omega(\log(V))$.
\end{theorem}

\begin{proof}

See Appendix~\ref{sec:proof}.

\end{proof}

Note that the combination of the lower bound on both the 1-systole and 2-systole gives rise to the lower bound of the code distance, i.e., 
\be
d \sim min(sys_1(\M_g; \ZZ_2), sys_2(\M_g; \ZZ_2)) = \Omega(\log^{\frac{1}{2}}(n)).
\ee

\subsubsection{Triple intersection and logical gates}

Let $\alpha_2^{(i)} \subset \M^3_g$, for $1\leq i \leq 3$, be three oriented surfaces representing three distinct homology classes in $H_2(\M^3_g; \ZZ_2)$. Assume that $\alpha_2^{(1)}$ represents the fiber class, which equals $\Sigma_g$, of the fibration $\M^3_g \to S^1$. Let $\alpha_1^{(2)}$ and $\alpha_1^{(3)}$  be the Poincaré dual cycles of $\alpha_2^{(2)}$ and $\alpha_2^{(3)}$ respectively in $H_1(\M_g^3; \ZZ_2)$. Let $v_1$ and $v_2$ be the image of  ${\alpha}_1^{(2)}$ and  ${\alpha}_1^{(3)}$ respectively, in $I(\tau_*)$ under the isomorphism Eq.~\eqref{isomorphism}. In Sec.~\ref{sec:generic_3-manifold} and Eq.~\eqref{intersectionpointsmappingtorus} the following equality will be  shown 
\begin{align}
\label{intersectionpointsfreedman}
    |\cap_{1\leq i \leq 3} \alpha_2^{(i)}| = | v_1 \cap v_2 | \in \ZZ_2
\end{align}
where the left hand side is the $\ZZ_2$ triple intersection number of the 2-cycles and $|\cdot \cap \cdot |$ on the right hand side is the $\ZZ_2$ intersection of two 1-cycles on $\Sigma_g$ [see Eq.~\eqref{eq: intersectionsurface} for precise definition].

More precisely, let $|\cdot \cap \cdot |_R$ be the restriction of the intersection pairing on $H_1(\Sigma_g, \ZZ_2)$ to the subspace $I(\tau_*)$ and let $| \cdot \cap \cdot |_{\tau}$ be the intersection pairing on $H_1(\Sigma_g/\langle\tau\rangle, \ZZ_2)$. Since $\tau$ is an isometry, the isomorphism Eq.~\eqref{isomorphism} extends to the following commutative diagram
\be\label{eq:commuting_diagram}
\begin{tikzcd}
 \vline \,\cdot\cap \cdot |_R : I(\tau_*) \times I(\tau_*)  \arrow{r}{} \arrow[swap]{d}{} & \ZZ_2 \arrow{d}{Id} \\%
\vline \,\cdot\cap \cdot |_\tau:H_1(\Sigma_g/\tau; \ZZ_2) \times H_1(\Sigma_g/\tau; \ZZ_2) \arrow{r}{}& \ZZ_2
\end{tikzcd}
\ee

Now let $v_1$ and $v_2$ be in $H_1(\Sigma_g/\tau, \ZZ_2)$ such that 
$| v_1 \cap v_2 |_{\tau} =1$ and let $\hat{v}_1$ and $\hat{v}_2$ be their extension in $H_2(\M^3_g; \ZZ_2)$. Let $\alpha_2^{(1)}$ in $H_2(\M^3_g; \ZZ_2)$ be fiber class of the fibration $\M^3_g \to S^1$. Then 
\begin{align}
    |\alpha_2^{(1)} \cap \hat{v}_1 \cap \hat{v}_2| = \  | v_1 \cap v_2 |_{\tau} =1
\end{align}
It follows that there exist $genus(\Sigma_g/\langle\tau\rangle)=\frac{g-1}{|\tau|} -1 $  many pairs $A_1$ and $A_2$ in $H_2(\M^3_g; \ZZ_2)$ such that 
\begin{align}
    |\alpha_2^{(1)} \cap A_1 \cap A_2 | =1 
\end{align}
where $\alpha_2^{(1)}$ is the fiber class.

When we use the fattening procedure to make a 4-colorable celluation of the 3-manifold $\M^3_g$, we obtain the logical non-Clifford gates and prallelizable logical Clifford gates in the corresponding \textit{fibre-bundle color code}, or more specifically a collective logical CCZ gate and parallelizable logical CZ gates via transversal $T$ and $S$ gates respectively.  Interestingly, the corresponding interaction hypergraph encoding the structure of the logical gates is isomorphic to that of the quasi-hyperbolic code shown in Fig.~\ref{fig:interaction_hyper-graph}(a) up to a reduction in size by $O(1/\log^\frac{1}{2}(g))=O(1/\log^\frac{1}{2}(k))$.

We hence reach the following theorem:
\begin{theorem}\label{theorem:non-Clifford}
There exists a family of fibre-bundle color code with encoding rate $O(1/\log^{\frac{1}{2}}(n))$ and distance lower bound  $d$$=$$\Omega(\log^{\frac{1}{2}}(n))$,  which supports fault-tolerant logical non-Clifford gates and parallelizable logical Clifford gates.
\end{theorem}

\subsection{Generic 3D hyperbolic codes via the mapping torus construction and constant-rate 3D hyperbolic   codes}\label{sec:generic_3-manifold}

In this section, we discuss homological quantum LDPC codes based on generic hyperbolic 3-manifolds obtained via the mapping torus construction, which can be considered as a more general twisted product construction with the twist being the general mapping class instead of just isometry. We also develop the theory of evaluating and bounding the Betti number and triple intersection of such generic hyperbolic 3-manifolds.  In the end of this section, we provide a specific 3D hyperbolic color code construction based on mapping tori of a pseudo-Anosov element in the Torelli subgroup.  Such family of hyperbolic 3-manifolds have a 1st Betti number proportional to the volume with the corresponding code having constant rate, resolving the issue of Betti number reduction due to the twist in the homological fibre-bundle code based on the Freedman-Meyer-Luo 3-manifolds discussed in Sec.~\ref{sec:fibre_bundle}.   Meanwhile, such 3-manifolds still have non-trivial triple intersection, which guarantees the existence of fault-tolerant non-Clifford gate in the 3D hyperbolic color code defined on such manifolds.   Although we cannot provide a proof for the lower bound of the systole and hence code distance for technical reasons, we expect that the systole (distance) growing with the volume (total number of qubits) due to the twist in a similar manner to the homological fibre-bundle code. 

Let $\Sigma_g$ be a closed, connected, and oriented surface of genus $g>0$. Let $\Gamma(\Sigma_g)$ be the group of isotopy classes of orientation preserving diffeomorphisms of $\Sigma_g$, explicitly, 
\begin{align}
    \Gamma(\Sigma_g) = Diff^+(\Sigma_g)/ Diff^+_0(\Sigma_g) 
\end{align}
where $Diff^+(\Sigma_g)$ is the group of orientation preserving diffeomorphims of $\Sigma_g$ and $Diff^+_0(\Sigma_g) \subset Diff^+(\Sigma_g)$ is the connected component of the identity with respect to the compact open topology on $Diff^+(\Sigma_g)$. $\Gamma(\Sigma_g)$ is called the \textit{mapping class group}, with its elements called mapping classes. Let $f \in \Gamma(\Sigma_g)$, and define the 3 dimensional manifold 
\begin{align}
    \mathcal{M}(f) := \Sigma_g \times [0, 1] / \sim_{f} , 
\end{align}
where the equivalence $\sim_{f}$ is $(x, 0) \sim (f(x), 1)$. This manifold fibers over the circle with each fiber diffeomorphic to $\Sigma_g$, i.e.  
\begin{align}
    \pi: M(f) \to S^1, \quad \quad \pi^{-1}(x)\cong \Sigma_g \,\,\text{for all} \, x \in S^1 
\end{align}
and with monodromy equal to $f$. See figure \eqref{fig:mappingtorus} for illustration. Stallings' fibration theorem, see Ref.~\cite{Stallings}, implies that a necessary and sufficient condition for a  compact irreducible three manifold $\mathcal{M}$ to fiber over the circle is that there exist a map $\pi_1(\mathcal{M}) \to \ZZ$ with kernel which is finitely generated and not $\ZZ_2$. Notice that if $Rank H_1(\mathcal{M}; \ZZ)\geq 1$ then the abelianizaton of the fundamental group composed with the projection 
\begin{align}
    \pi_1(\mathcal{M}) \to H_1(\mathcal{M};  \ZZ) \to \ZZ
\end{align}
gives the desired map in the Stallings' theorem. Hence, there exists a surface $\Sigma_g$ and a mapping class $f$ in $\Gamma(\Sigma_g)$ such that 
\begin{align}
    \mathcal{M} \cong \mathcal{M}(f)
\end{align}
On the other hand, it is entirely possible that $\mathcal{M}$ fibers over the circle in infinitely many different ways. Ian Agol's breakthrough proof of Thurston's virtual fibering conjecture, see \cite{Agol}, implies that any compact, closed, and oriented hyperbolic three manifold $\mathcal{M}$ has a finite degree cover 
\begin{align}
\mathcal{M}(f) \to \mathcal{M}
\end{align}
such that $\mathcal{M}(f)$ is a pseudo-Anosov mapping torus of a mapping class $f$ in $\Gamma_g$ for some $g$.

\begin{figure}[hbt]
  \includegraphics[width=0.6\columnwidth]{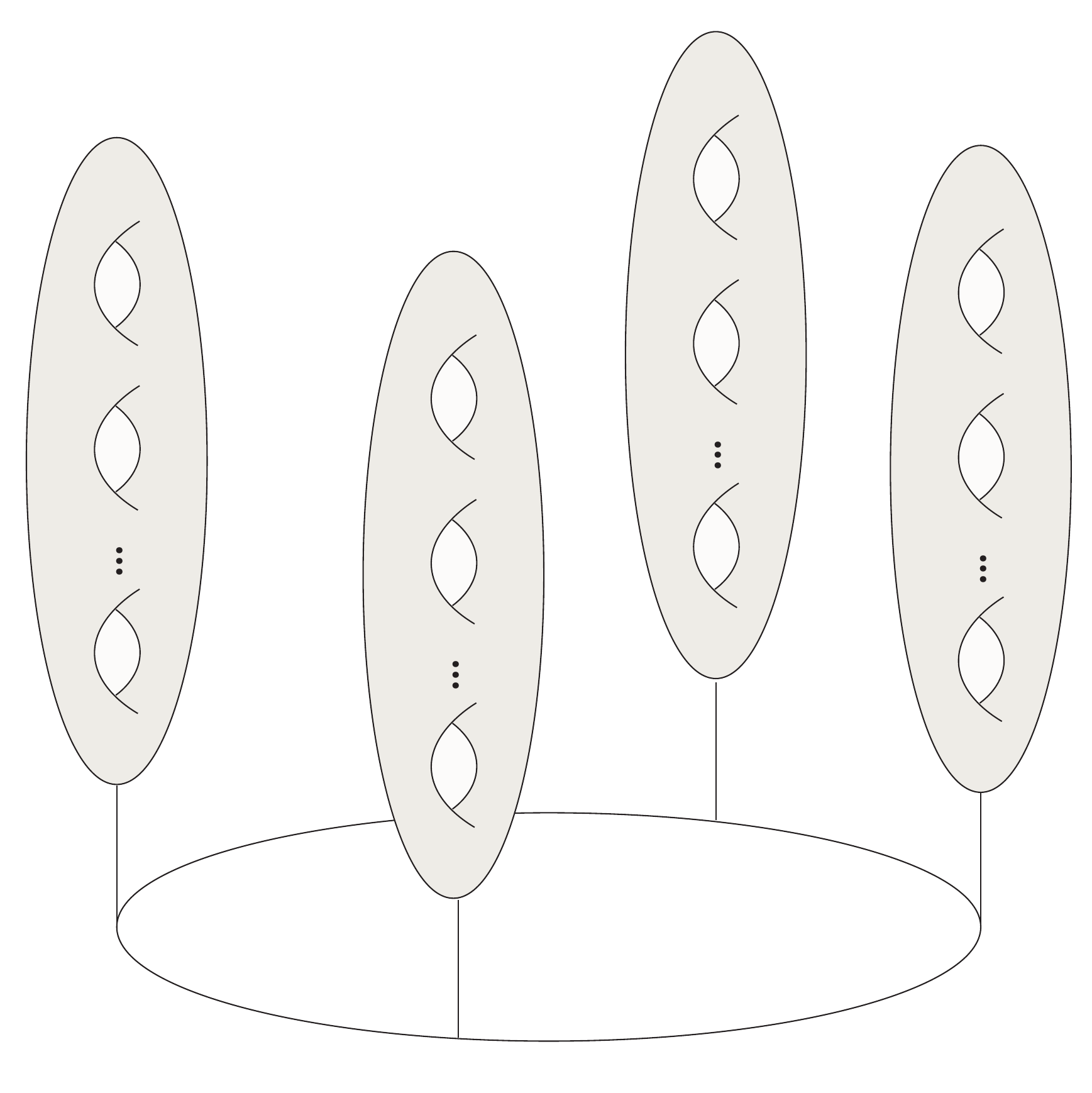}
  \caption{Mapping torus of a genus-$g$ surface.}
  \label{fig:mappingtorus}
\end{figure}

The action of $f$ on $\Sigma_g$ induces an action on the lattice $H_1(\Sigma_g; \ZZ)\cong \ZZ^{2g}$. Let $\hat{f}$ denote this induced action on $H_1(\Sigma_g; \ZZ)$; $\hat{f}$ preserves the intersection pairing on $H_1(\Sigma_g; \ZZ)$ and belongs to the group $Sp(2g, \ZZ)$. In fact, this procedure of inducing action on the homology gives a short exact sequence
\begin{align}
    1\to T_g \to \Gamma(\Sigma_g) \to Sp(2g, \ZZ) \to 1
\end{align}
where the kernel $T_g$ is called the Torelli group. 

For $g>1$, one has the Nielsen-Thurston  classification of elements of $\Gamma(\Sigma_g)$, which generalizes the classification of elements of $\Gamma(\Sigma_1)=SL(2, \ZZ)$. This classification shows that every mapping class is either  finite order, reducible, or a pseudo-Anosov element. A mapping class  $f$  is finite    order  if there exists an integer $N>1$ such that $f^N$ is isotopic to the identity. A     mapping    class $f$ is reducible     if  it preserves a finite union of disjoint simple closed curves in $\Sigma_g$. A mapping class is pseudo-Anosov if it is neither finite order nor reducible. See \cite{FarbMargalit}  for a complete discussion. Thurston's hyperbolization theorem states that a mapping torus $\mathcal{M}(f)$ is hyperbolic if and only if $g>1$ and $f$ is a pseudo-Anosov element. See theorem 13.4 and 13.5 of \cite{FarbMargalit} for geometric structures on mapping tori $\mathcal{M}(f)$ for $f$ reducible, finite order, and for the case when $g=1$.

\begin{figure}[hbt]
  \includegraphics[width=1\columnwidth]{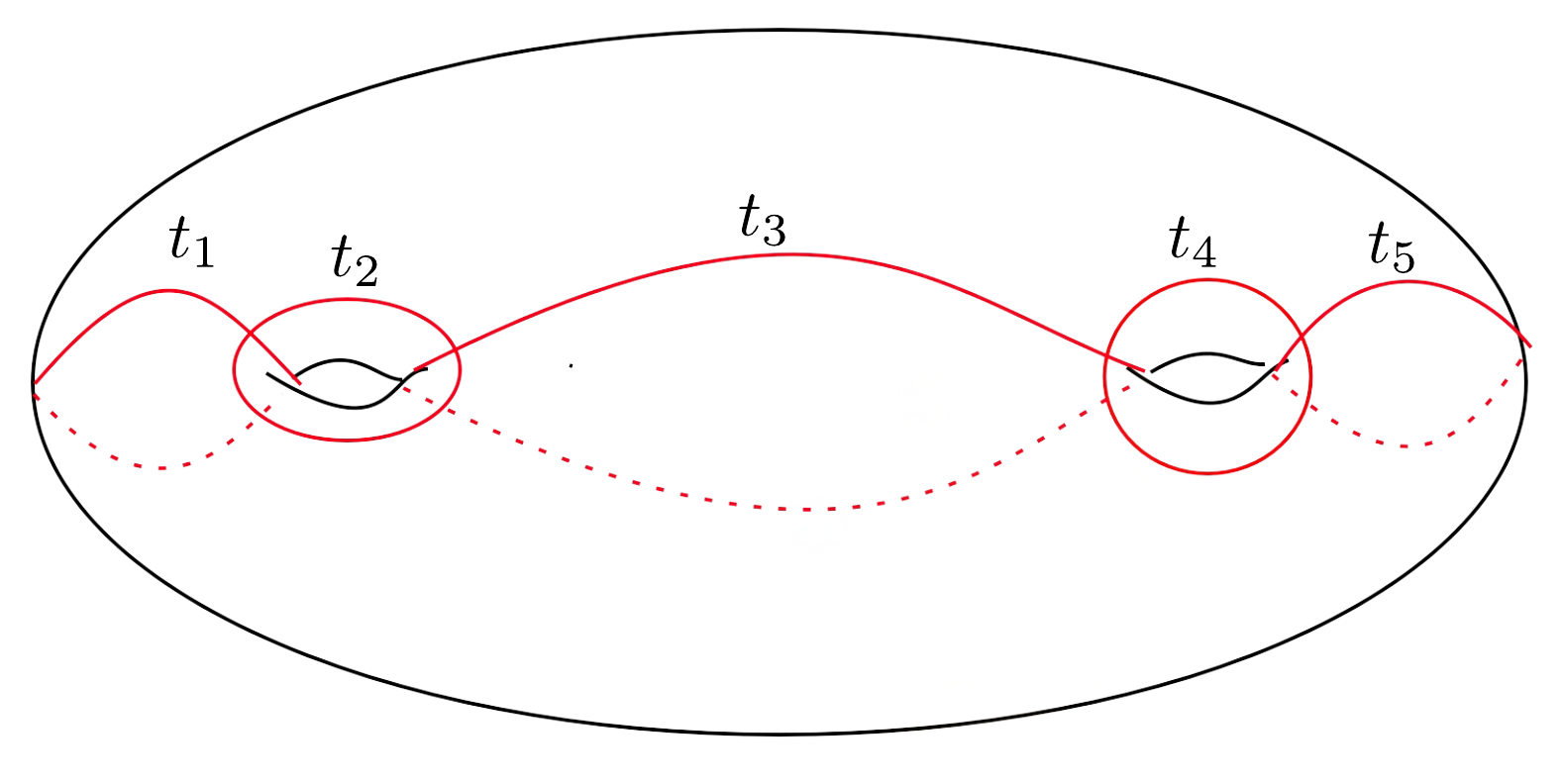}
  \caption{Humphries' generators of the genus three mapping class group. }
  \label{fig:humphriesgenerators}
\end{figure}

For a given $f$ in $\Gamma(\Sigma_g)$, the topology of $\mathcal{M}(f)$ can be computed as follows: decomposing the base of the fibration $\mathcal{M}(f) \to S^1$ into two open sets induces a decomposition of $\mathcal{M}(f)$. Applying the Mayer-Vietoris to this decompostion gives the following long exact sequence 
\begin{align}
    \dots \to H_1(\Sigma_g ; \ZZ) \xrightarrow{\hat{f}-Id}  H_1(\Sigma_g; \ZZ) \to H_1(\mathcal{M}(f); \ZZ) \to \ZZ \to 0 
\end{align}
which yeilds 
\begin{align}
\label{h1}
    H_1(\mathcal{M}(f); \ZZ) = \ZZ \oplus Coker(\hat{f} - Id)
\end{align}
where the infinite cyclic factor $\ZZ$ is generated by the class of the base of the fibration. 

There exist two matrices $P$ and $Q$ in $GL(2g, \ZZ)$ such that 
\begin{align}
    P.(\hat{f} - Id).Q = \begin{bmatrix}
    A(1) &  & &&&&\\
    & A(2) & &&&&\\
    & &\ddots & &&&\\
    & & & A(N)&&&\\
    &&&&0&&\\
    &&&&&\ddots &&\\
    &&&&&&&0
  \end{bmatrix}
\end{align}
where for all $1\leq i \leq N \leq 2g$ the diagonal entries $A(i)$ are  positive integers such that for all $i>1$, $A(i-1)$ divides $A(i)$. The diagonal matrix on the RHS is called the Smith normal form of $(\hat{f} -Id)$ and will be denoted as $SNF(\hat{f} -Id)$. The $SNF(\hat{f} -Id)$ does not depend on the pair $P,Q$ and is an invariant of the conjugacy class of $\hat{f}$ in $Sp(2g, \ZZ)$. In particular, 
\begin{align}
\label{coker}
    Coker(\hat{f} - Id) = \ZZ^{2g-N} \oplus \ZZ/A(1) \oplus \cdots \oplus \ZZ/A(N)
\end{align}
The diagonal entries $A(i)$ of $SNF(\hat{f} -Id)$ can be computed from the the entries of the matrix $(\hat{f} -Id)$ as follows: Let $D(i)$ be the greatest common divisor of all the determinants of $i \times i$ minors of $(\hat{f} -
Id)$, and let $D(0):=0$,  then 
\begin{align}
\label{algo}
    A(i)= \frac{D(i)}{D(i-1)}          
\end{align}
The fastest known algorithm for computing these integers $A(i)$ has runtime complexity 
\begin{align}
O(\norm{\hat{f}-Id}\log \norm{\hat{f}-Id} 16g^4\log 2g) 
\end{align}
The lattice $H_1(\mathcal{M}(f); \ZZ)$ modulo torsion is the lattice generated by the set of the following vectors 
\begin{align}
\label{invariant}
    \{v \in H_1(\Sigma_g; \ZZ) | \hat{f}(v) = v\} 
\end{align}
direct sum with an infinite cyclic factor $\ZZ$ coming from the homology class of the base of the fibration $M(f) \to S^1$. In particular 
\begin{align}
    1 \leq Rank (H_1(\mathcal{M}(f); \ZZ)) \leq 2g+1 
\end{align}
with the maximal 
\begin{align}
    Rank (H_1(\mathcal{M}(f); \ZZ)) = 2g+1 
\end{align}
if and only if $f$ belongs to the Torelli group $T_g$.

For any prime $p$  we can consider $H_1(\Sigma_g; \ZZ_p)$. Any mapping class $f$ in $\Gamma(\Sigma_g)$ also acts on $H_1(\Sigma_g; \ZZ_p$): The construction of the mapping torus and calculation of its homology $H_1(\mathcal{M}(f); \ZZ_p)$ goes through verbatim since the Smith normal form exists  over any principle ideal domain. Henceforth we shall restrict to $\ZZ_2$.

Let $\alpha_2^{(1)}$, $\alpha_2^{(2)}$, and $\alpha_2^{(3)}$ be three oriented surfaces contained in the 3-manifold $\mathcal{M}(f)$ which are representitives of three distinct classes in $H_2(\M(f); \ZZ_2)$. One can form transversal intersections of the three surfaces and calculate their $\ZZ_2$ intersection number. This count of $\ZZ_2$ intersection numbers defines a skew-symmetric 3-form on $H_2(\mathcal{M}(f); \ZZ_2)$ with values in $\ZZ_2$. This 3-form is the Poincaré dual of the cup-product on $H^1(M(f); \ZZ_2)$, i.e. if $\alpha^1_{(1)}$, $\alpha^1_{(2)}$, and $\alpha^1_{(3)}$ are Poicaré duals in $H^1(\M(f); \ZZ_2)$ then 
\begin{align}
\label{eq: algebraicintersection} 
    | \alpha_2^{(1)} \cap \alpha_2^{(2)} \cap \alpha_2^{(3)} |:=\int_{\M(f)} \alpha^1_{(1)} \cup \alpha^1_{(2)}\cup \alpha^1_{(3)}   \in \ZZ_2 
\end{align}

    Simplest example is $g=1$ and $f= ID$. In this case $M(f)=T^3=\RR^3/\ZZ^3$, $H_1(T^3; \ZZ_2)=\ZZ_2^3$. It is easy to see that the $\ZZ_2$ intersection number of the three standard cycles of $T^3$ is one [see Fig.~\ref{fig:three-manifold_examples}(a,b)]. In the case $g>1$ and $f = ID$ one has that $\mathcal{M}(f) = \Sigma_g \times S^1$ and the Künneth formula implies 
    \begin{align}
        H_1(\Sigma_g \times S^1; \ZZ_2) = H_1(\Sigma_g; \ZZ_2) \oplus H_1(S^1; \ZZ_2) 
    \end{align}
    Let $\alpha^{(1)}_2$ be the class in $H_2(\mathcal{M}(f); \ZZ_2)$ which is the Poincaré dual of the class of $S^1$ in $H_1(\mathcal{M}(f); \ZZ_2)$ and let $\alpha_2^{(2)}$ and $\alpha_2^{(3)}$ be Poincaré duals in $H_2(\mathcal{M}(f); \ZZ_2)$ of two arbitrary classes $\alpha_1^{(1)}$ and $\alpha_1^{(2)}$ in $H_1(\Sigma_g; \ZZ_2)$, then 
    \begin{align}
        |\alpha_2^{(1)} \cap \alpha_2^{(2)} \cap \alpha_2^{(3)}|= | \alpha_1^{(1)}\cap \alpha_1^{(2)} | \in \ZZ_2 
    \end{align}
where the right hand side is the intersection form  
\begin{align}
\label{eq: intersectionsurface}
|\cdot \cap \cdot | :H_1(\Sigma_g, \ZZ_2) \times H_1(\Sigma_g, \ZZ_2)\to \ZZ_2   
\end{align}
defined for $v_1$, $v_2$ in $H_1(\Sigma_g, \ZZ_2)$ by 
\begin{align}
| v_1 \cap v_2 | = \int_{\Sigma_g} v^*_1\cup v^*_2
\end{align}
where $v^*_1$ and $v^*_2$ are the Poincaré duals in $H^1(\Sigma_g, \ZZ_2)$ of $v_1$ and $v_2$.

Notice that for Poincaré duals $\alpha_2^{(i)}$, where $1\leq i \leq 3$, in $H_2(\mathcal{M}(f); \ZZ_2)$   of any three classes $\alpha_1^{(i)}$ in $H_1(\Sigma_g; \ZZ_2)$ one has that 
\begin{align}
\label{intnumber}
    |\alpha_2^{(1)} \cap \alpha_2^{(2)} \cap \alpha_2^{(3)}| = 0
\end{align}
See Fig.~\ref{fig:three-manifold_examples}(c,d).

In the case of arbitrary $g>0$ and $f$ in $\Gamma(\Sigma_g)$, it follows from \eqref{h1}, and Poincaré duality, that $H_2(\mathcal{M}(f); \ZZ_2)$ is 
\begin{align}
\label{intersectionpointsmappingtorus} 
Coker(\hat{f} - Id) = H_1(\Sigma_g; \ZZ_2)/(\hat{f}-Id)H_1(\Sigma_g; \ZZ_2) 
\end{align}
direct sum with $\ZZ_2$ generated by the class of the base of the fibration.  In particular, the intersection form on $H_1(\Sigma_g; \ZZ_2)$ restricts to an intersection form on  $Coker(\hat{f} - Id) $. Let $\alpha_2^{(1)}$ be the Poincaré dual in $H_2(\mathcal{M}(f); \ZZ_2)$ of the homology class generated by the base of the fibration, and let  $\alpha_2^{(2)}$ and $\alpha_2^{(3)}$ be Poincaré duals in $H_2(\mathcal{M}(f); \ZZ_2)$ of two arbitrary cycles $\alpha_1^{(1)}$ and $ \alpha_1^{(2)} $ in $Coker(\hat{f} - Id)$. Then we have
\begin{align}
     |\alpha_2^{(1)} \cap \alpha_2^{(2)} \cap \alpha_2^{(3)}| = |\alpha_1^{(1)}\cap \alpha_1^{(2)}| \in \ZZ_2 
\end{align}
In this case, if $\alpha_2^{(1)}, \alpha_2^{(2)}$, and $ \alpha_2^{(3)}$ are Poincaré duals in $H_2(\mathcal{M}(f);  \ZZ_2)$  of three arbitrary classes in $Coker(\hat{f} - Id)$ then it could well be that 
\begin{align}
    |\alpha_2^{(1)} \cap \alpha_2^{(2)} \cap \alpha_2^{(3)}| \neq 0
\end{align}
unlike in the case of Eq.~\eqref{intnumber}

We present a few examples: Figure \ref{fig:humphriesgenerators}
shows the set of Humphries generators for the mapping class group specialized to the case of genus two. The Dehn twist around each simple closed curve $t_i$ acts on $H_1(\Sigma_2; \ZZ)$ through the following symplectic matrices
\begin{align}
    T_1=\begin{pmatrix} 1	&0	&1	&0\\
0&	1&	0&	0\\
0&	0&	1&	0\\
0&	0&	0&	1
  \end{pmatrix}, 
    T_2=\begin{pmatrix} 1&	0&	0&	0\\
0&	1&	0&	0\\
-1&	0&	1&	0\\
0&	0&	0&	1
  \end{pmatrix},
    T_3=\begin{pmatrix} 1&	0&	1&	1\\
0&	1&	1&	1\\
0&	0&	1&	0\\
0&	0&	0&	1
\end{pmatrix} 
    \end{align}
    \begin{align}
    T_4=\begin{pmatrix} 1&	0&	0&	0\\
0&	1&	0&	0\\
0&	0&	1&	0\\
0&	-1&	0&	1
 \end{pmatrix},
    T_5=\begin{pmatrix} 1&	0&	0&	0\\ 
0&	1&	0&	1\\
0&	0&	1&	0\\
0&	0&	0&	1
  \end{pmatrix}
\end{align}
A generic word in $T_i$ does not preserve any homology class of the surface, hence the first/second homology group of the resulting mapping torus has rank one and the 3-form is trivial.

In the case we take some positive power $a$ of a Dehn twist around a single $t_i$, for example $t_5$, then the Smith normal form of $T^a_5 - Id$ is 
\begin{align}
    \begin{pmatrix}
        a	&0	&0	&0\\
0&	0&	0&	0\\
0&	0&	0&	0\\
0&	0&	0&	0
    \end{pmatrix}
\end{align}
It follows that 
\begin{align}
    H_2(\mathcal{M}(t_5); \ZZ) = \ZZ^4 \oplus \ZZ/a
\end{align}
The mapping class $t_5$ preserves three homology cycles, labelled $\alpha_1^{(1)}, \alpha_1^{(2)}$, and $\alpha_1^{(3)}$, in $\Sigma_2$, see Fig.~\ref{fig:threecycles}.
\begin{figure}[t]
  \includegraphics[width=0.8\columnwidth]{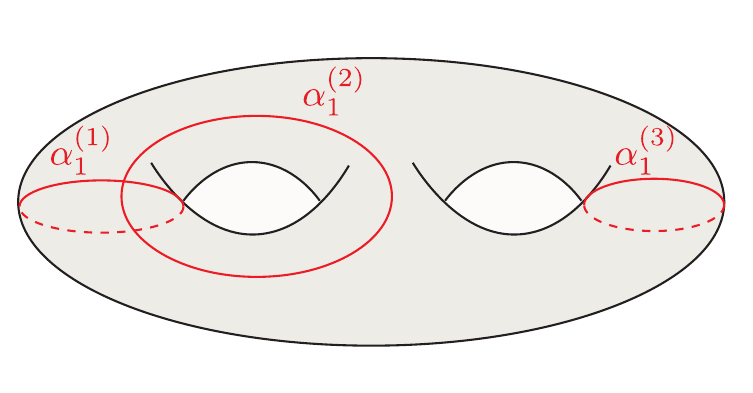}
  \caption{A choice of three cycles such that two of them have intersection number one and the other two pairs have intersection number zero. }
  \label{fig:threecycles}
\end{figure}

\begin{figure}[t]
  \includegraphics[width=1\columnwidth]{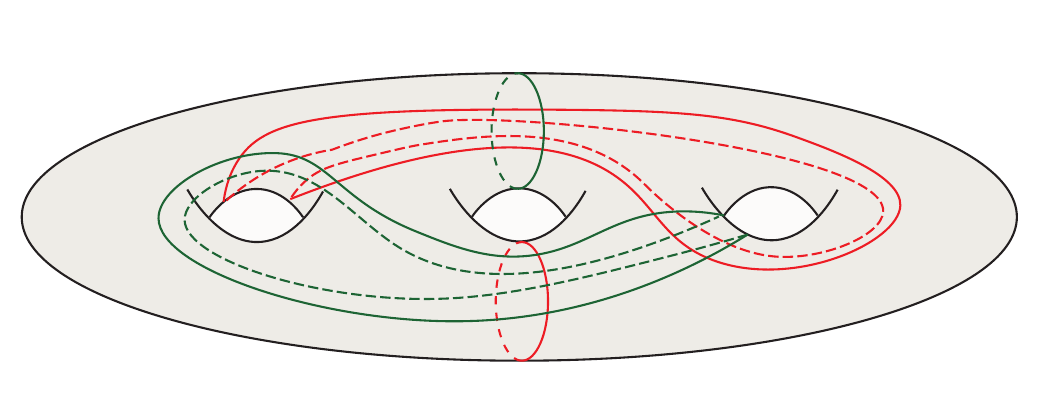}
  \caption{Two multicurves in red and green which generate pseudo-Anosov in the Torelli group of genus three using Thurston's construction. }
  \label{fig:pseudo-anosov}
\end{figure}

\begin{figure}[hbt]
  \includegraphics[width=0.8\columnwidth]{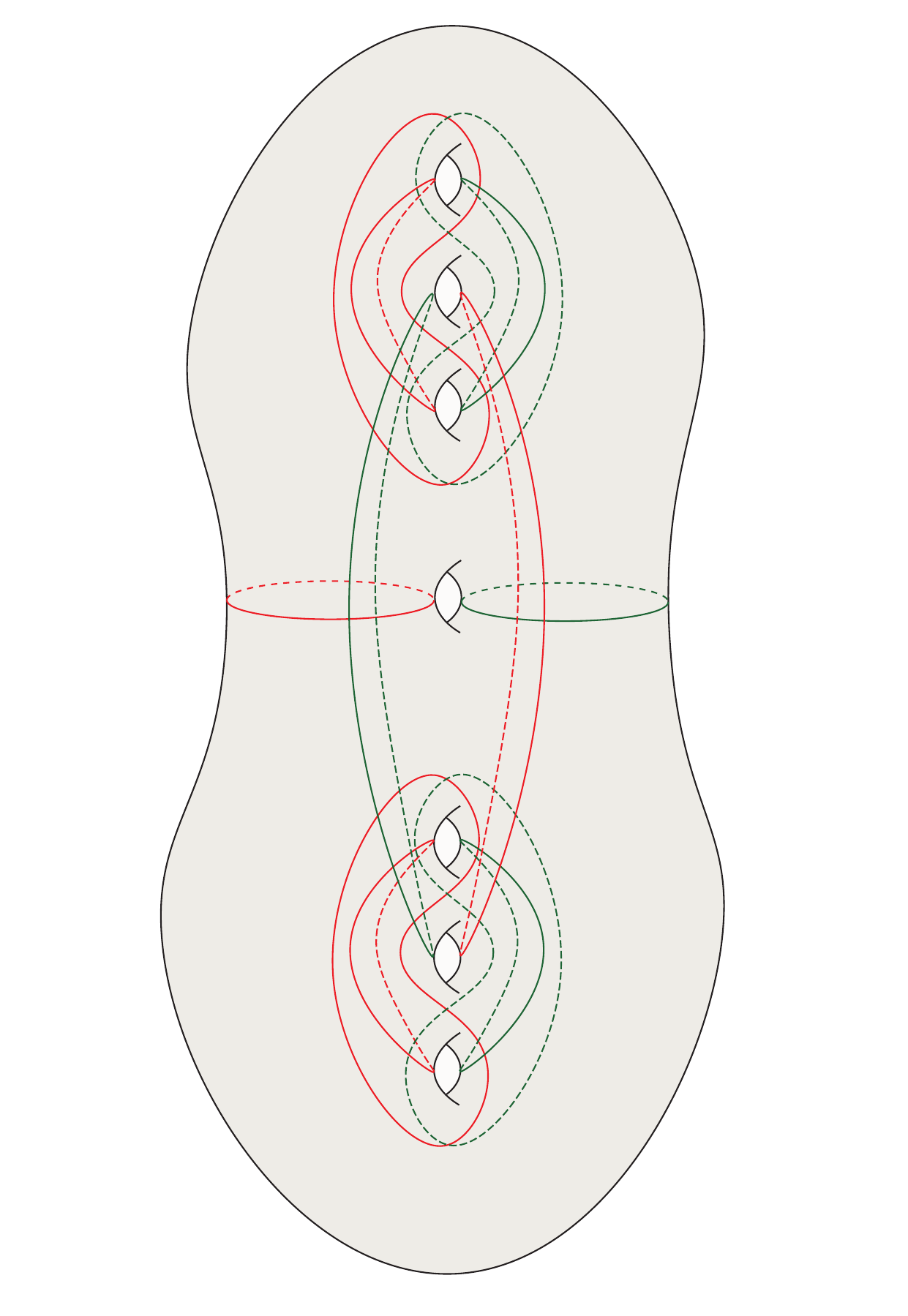}
  \caption{Two multicurves in red and green which generate pseudo-Anosov in the Torelli group of genus seven using Thurston's construction. }
  \label{example2}
\end{figure}

These three homology cycles extend to give three of the classes in $H_2(\mathcal{M}(t_5); \ZZ)$, which we denote by $\alpha_2^{(1)}, \alpha_2^{(2)}$, and $\alpha_2^{(3)}$, and the fourth class comes from the base of the fibration which we denote as $\alpha_2^{(0)}$. Let $\overline{\alpha}_2^{(1)}, \overline{\alpha}_2^{(2)}, \overline{\alpha}_2^{(3)}$, and $\overline{\alpha}_2^{(0)}$ be the images of these cycles in $H_2(\M(f); \ZZ_2)$ under the universal coeffecient theorem.  It is easy to see that 
\begin{align}
    |\overline{\alpha}_2^{(0)} \cap \overline{\alpha}_2^{(1)}\cap \overline{\alpha}_2^{(2)}| = 1 
\end{align}
and 
\begin{align}
  |\overline{\alpha}_2^{(0)}\cap \overline{\alpha}_2^{(1)}\cap \overline{\alpha}_2^{(3)} |  &=  |\overline{\alpha}_2^{(0)}\cap \overline{\alpha}_2^{(2)}\cap \overline{\alpha}_2^{(3)} | \\&=|\overline{\alpha}_2^{(1)}\cap \overline{\alpha}_2^{(2)}\cap \overline{\alpha}_2^{(3)} | \\
  &= 0 
\end{align}
It should be pointed out that since $T_5$ is a reducible mapping class $M(T_5)$ is not a hyperbolic manifold.

To construct a hyperbolic example with non-trivial triple intersection numbers, we shall use a mapping class from the Torelli subgroup. Torelli subgroup can be generated by performing Dehn twists around simple closed curves in $\Sigma_g$ which are null-homologous. Thurston gave a general construction of a family of pseudo-Anosov elements by performing powers of Dehn twists simultaneously along the components of a pair of filling multicurves, indeed, if all the components of the multicurves are null-homologous then one obtains a  pseudo-Anosov in the Torelli group\footnote{Jakob Nielsen had conjectured that there would be no pseudo-Anosov elements in the Torelli group. Thurston's construction was a counter example to this conjecture}. We recall Thurston's construction very briefly and then present an example: 

A multi curve $A$ on a surface $\Sigma_g$ is a collection of simple closed curves $\{a_1, \dots, a_n\}$ such that they do not intersect each other pairwise. A pair of multicurves $A$ and $B$ is called filling if the complement of the union of the two multicurves is a disjoint union of simply connected domains. Given two multicurves $A=\{a_1, \dots, a_n\}$ and $B=\{b_1, \dots, b_m\}$  such that they are filling, form the $n \times m $ matrix $N_{i,j} = I\{a_i, b_j\}$, where $I\{a_i, b_j\}$ denotes the geometric intersection number between the curve $a_i$ and $b_j$. The square matrix $N.N^T$ is a Perron-Frobenious matrix, denote by $\nu$ its Perron-Frobenious eigenvalue which is a positive real number. Let $T_A$, respectively $T_B$, denote the simultaneous Dehn twists along all the components of $A$, respectively all the components of $B$. Map the two mapping classes $T_A$ and $T_B$ to $\mathbb{P}SL(2, \RR)$ by the following map 
\begin{align}
    T_A \mapsto \begin{pmatrix}
        1 & \sqrt{\nu}\\ 0 &1 
    \end{pmatrix}, \quad \quad T_B \mapsto \begin{pmatrix}  1 & 0\\ -\sqrt{\nu} & 1\end{pmatrix}
\end{align}
Any word in $T_A$ and $T_B$ such that its image in $\mathbb{P}SL(2, \RR)$ has absolute value bigger than two corresponds to a pseudo-Anosov element. Moreover, the larger of its two eigenvalues corresponds to the stretch factor of the pseudo-Anosov. 

An example of the construction above  of pseudo-Anosovs in the Torelli group is given by the following collection of multi curves on a genus three surface, see Fig.~\ref{fig:pseudo-anosov}, drawn in green and red respectively.

One calculates that 
\begin{align}
    N.N^T = \begin{pmatrix} 4 & 8\\ 0 & 4 \end{pmatrix}
\end{align}
and 
\begin{align}
    \nu = 16(3 + 2\sqrt{2}) 
\end{align}
The absolute value of the trace of the image of $T_A.T_B$ is larger than two and one easily calculates that the corresponding stretch factor is 
\begin{align}
    23+16 \sqrt{2}+4 \sqrt{65+46 \sqrt{2}} = 91.2439
\end{align}
The authors in Ref.~\cite{KojimaMcShane}, see also Ref.~\cite{BrockBromberg}, give an upper bound on the volume of the hyperbolic mapping torus in terms of the logarithm of the stretch factor of the pseudo-Anosov, their result implies 
\begin{align}
    vol(\M(T_A.T_B))\leq 3\pi(2g-2)\log(91.2439)
\end{align}

Since $T_A.T_B$ belongs to the Torelli group and preserves every class $H_1(\Sigma_3; \ZZ_2)=\ZZ_2^6$, it follows that 
\begin{align}
    H_2(\M(T_A.T_B; \ZZ_2) = \ZZ_2^{7}
\end{align}
where one factor $\ZZ_2$ comes from the base of the fibration and $\ZZ_2^6$ comes from the fact that every class in $H_1(\Sigma_3; \ZZ_2)$ extends to a class in the first/second homology group of $\mathcal{M}(T_A.T_B)$. If we denote by $\alpha_2^{(0)}$ the class in $H_2(\mathcal{M}(T_A.T_B); \ZZ_2)$ generated by the fibers of the fibration $\M(T_A.T_B) \to S^1$, and denote by $\alpha_2^{(1)}$ and $\alpha_2^{(2)}$ two arbitrary classes in $H_2(\mathcal{M}(T_A.T_B; \ZZ_2)$ which are extensions of two classes $\alpha_1^{(1)}$ and $\alpha_1^{(2)}$ in $H_1(\Sigma_3; \ZZ_2)$, then we have 
\begin{align}
\label{eq: lowerbound}
    |\alpha_2^{(0)} \cap \alpha_2^{(1)} \cap \alpha_2^{(2)}| = | \alpha_1^{(1)} \cap \alpha_1^{(2)}| \in \ZZ_2   
\end{align}
where the right hand side is the intersection form on $H_1(\Sigma_3; \ZZ_2)$. 

We end this section with a result of Ref.~\cite{AgolLeiningerMargalit} which generalizes the above construction of genus three mapping tori of a pseudo-Anosov element in the Torelli subgroup to higher genus. Their proposition 3.1 along with the result of Ref.~\cite{KojimaMcShane} implies the following theorem:
\begin{theorem}
Let genus $g=3k+1$, then for all $k\geq2$ there exists a pseudo-Anosov mapping class $f_g$ in the Torrelli subgroup $T_g$ of the mapping class group $\Gamma(\Sigma_g)$ such that 
\begin{align}
\non     &b_1(\mathcal{M}(f_g))= Rank(H_1(\mathcal{M}(f_g); \ZZ_2)) =2g+1, \\
    &vol(\mathcal{M}(f_g)) \leq 36\pi (2g+1)\log(2),
\end{align}
\end{theorem}
\begin{corollary}
    There exists $g$ triple intersection points in $\M(f_g)$ for all $g$. 
\end{corollary}
\begin{proof}
    There exist a basis $\alpha_1^{(i)}$, $\beta_1^{(i)}$ for $1\leq i \leq g$ of $H_1(\Sigma_g; \ZZ_2)$ such that $|\alpha_1^{(i)} \cap \beta_1^{(j)}| = \delta_{i, j}$.  Since $f_g$ is in the Torelli group, all $\alpha_1^{(i)}$ and $\beta_1^{(i)}$  extend to distinct elements in $H_2(\M(f_g), \ZZ_2)$. Denote these extensions by $\hat{\alpha}_2^{(i)}$ and $\hat{\beta}_2^{(i)}$ and denote by $\Gamma$ a representative in $H_2(\M(f), \ZZ_2)$ of the fiber class of the fibration.  It follows from \eqref{eq: lowerbound} that 
    \begin{align}
        |\Gamma \cap \hat{\alpha}_2^{(i)} \cap \hat{\beta}_2^{(j)}| = |\alpha_1^{(i)} \cap \beta_1^{(j)}| = \delta_{i, j}
    \end{align}
\end{proof}

Quite surprisingly, the corresponding interaction hypergraph of this code is identical to that of the quasi-hyperbolic code shown in Fig.~\ref{fig:interaction_hyper-graph}(a) only up to a change of the labels of the 2-cycles.

Moreover, $\mathcal{M}(f_g)$ has the least hyperbolic volume amongst all mapping tori of genus $g$ with the property that $b_1=Rank(H_1(\mathcal{M}(f_g); \ZZ_2)) =2g+1$. These pseudo-Anosov elements are generated using, just as above, Thurston's construction relying on Dehn twists around a pair of filling multicurves. Figure \ref{example2} shows the choice of multicurves $A$ and $B$ in red and green respectively for genus five; this is taken from Ref.~\cite{AgolLeiningerMargalit}.

When using the fattening procedure to create 4-colorable cellulation of the manfiold $\mathcal{M}(f_g)$, we can obtain a 3D hyperbolic color code with the transversal $T$ and $S$ gate being the collective logical CCZ gate and parallelizable logical CZ gate respectively. We call these codes Torelli mapping class codes.  We hence reach the following theorem:

\begin{theorem}\label{theorem:non-Clifford}
There exists a family of 3D hyperbolic color code with constant encoding rate,  which supports fault-tolerant logical non-Clifford gates and parallelizable logical Clifford gates.
\end{theorem}

\begin{figure*}[hbt]
  \includegraphics[width=1.6\columnwidth]{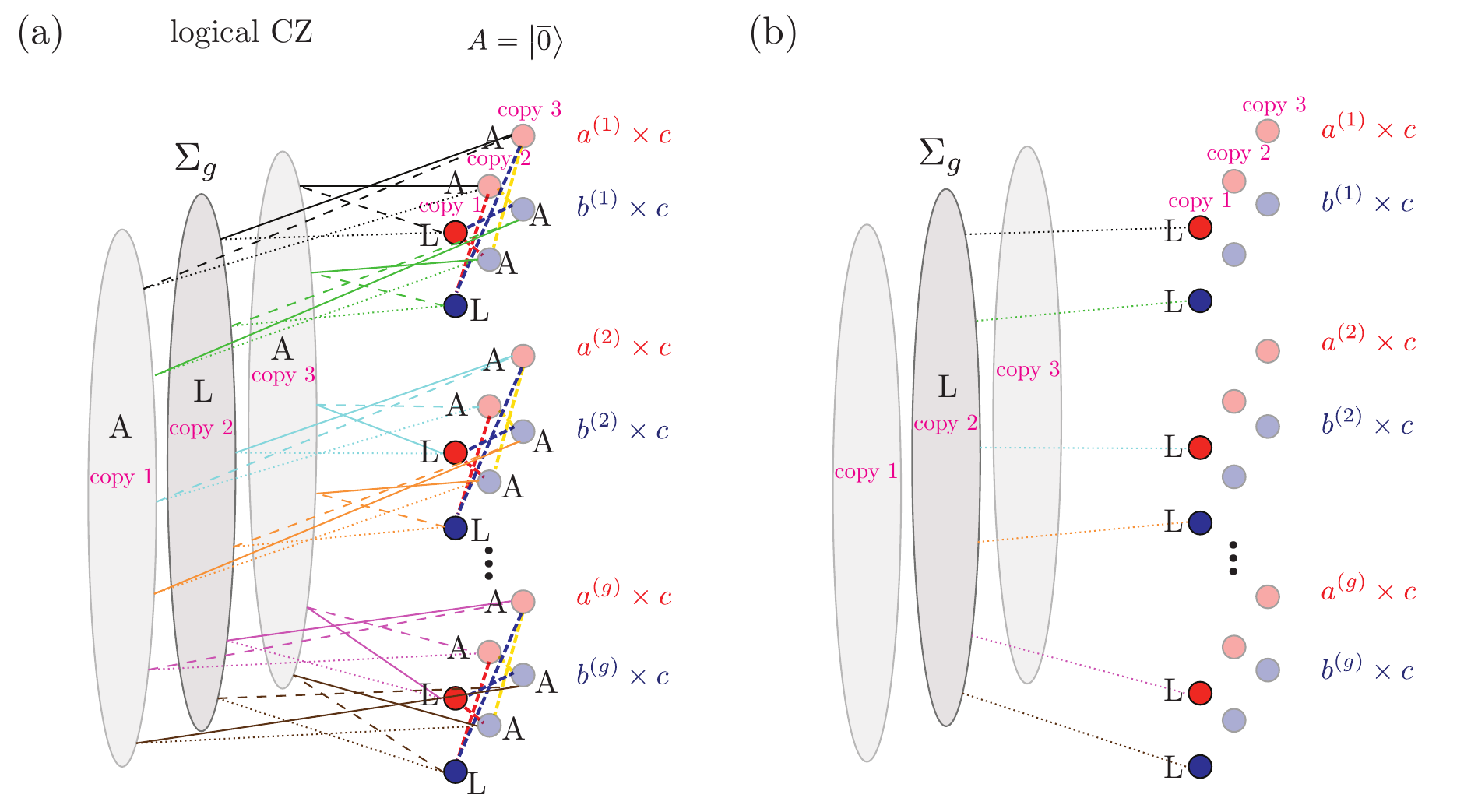}
  \caption{(a) Interaction graph for the logical CZ gates of the quasi-hyperbolic code. Colored edges with the same type (represented by both colors and line types) represent logical CZ gates that get simultaneously applied.  Choosing certain vertices (represented by circles or ellipses) as logical qubits (L) and others as ancillae (A) set in the logical state $\ket{0}$.  (b) The effective interaction graph for the logical qubits (L). Note that each colored edge has a unique type, representing individually addressable and parallelizable logical CZ gate.   }
\label{fig:application_interaction-graph}
\end{figure*}

\section{Application to parallel universal fault-tolerant quantum computation}\label{sec:application}

In this section, we consider a generic scheme for parallel universal fault-tolerant quantum computation for constant-rate or almost-constant-rate quantum LDPC codes based on 3-manifolds and potentially extendable to even more general cases such as quantum expander codes. The scheme is independent of the specific geometry of 3-manifold one chooses.  However, since the three codes discussed in Sec.~\ref{sec:code_construction}
has the same interaction hypergraph structure, we discuss our scheme using this particular structure.  For simplicity, we use the particular cycle labels of the quasi-hyperbolic code, while one should keep in mind that the scheme works the same for the other two codes.

\subsection{Generic scheme for parallelizable logical Clifford gates}\label{sec:parallel_Clifford}

\begin{figure}[hbt]
  \includegraphics[width=1\columnwidth]{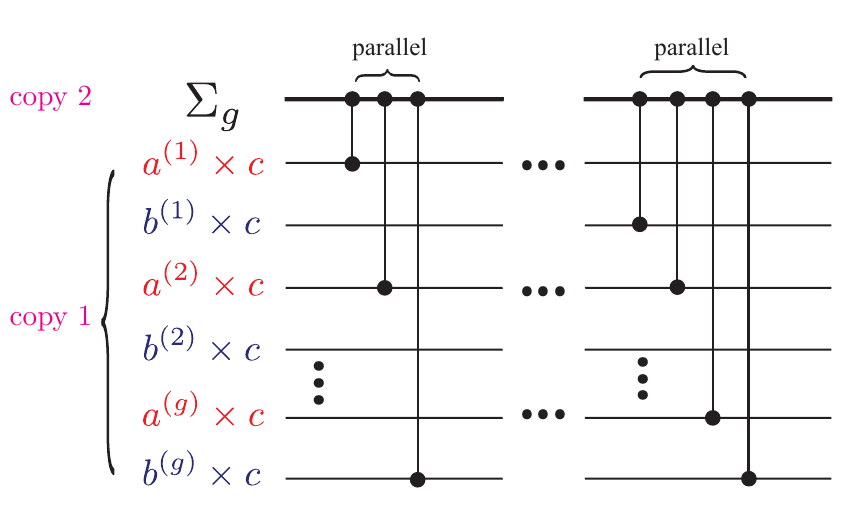}
  \caption{Logical CZ gates can be applied in parallel, but are all connected to the logical qubit supported on $\Sigma_g$ due to the star-like interaction graph. }
\label{fig:circuit_example}
\end{figure}

\begin{figure*}[hbt]
  \includegraphics[width=1.6\columnwidth]{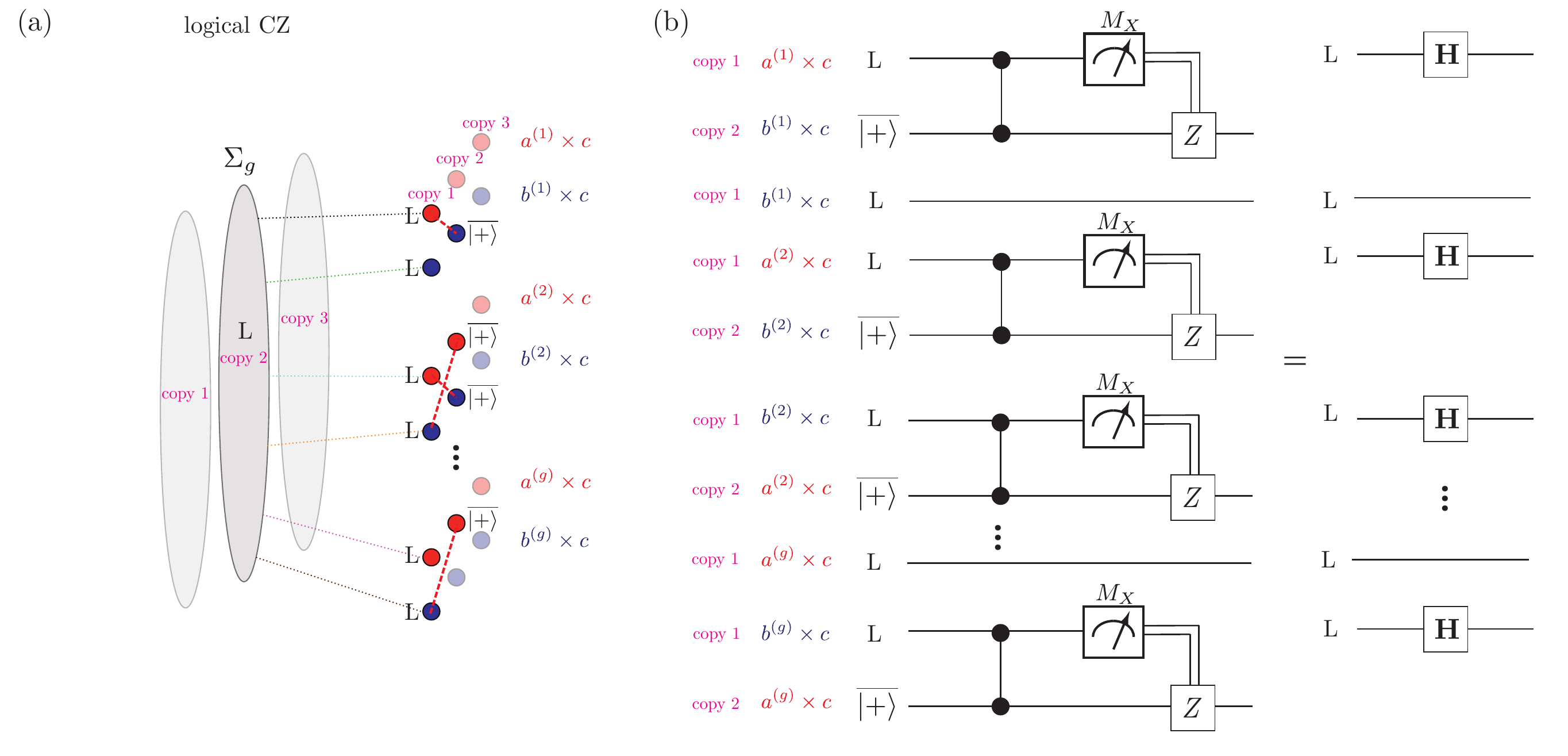}
  \caption{(a) Interaction graph corresponding to logical CZ gates, where the edges represented by red dashed lines represent parallel logical CZ gates.  One can selectively initialize ancillae in the second toric-code copy into $\lo{\ket{+}}$ state.   (b) Circuits to apply parallel logical Hadamard via parallel logical CZ gates, logical-$X$ measurement, and conditional Z correction. }
\label{fig:logical_Hadamard}
\end{figure*}

\begin{figure*}[hbt]
  \includegraphics[width=1.4\columnwidth]{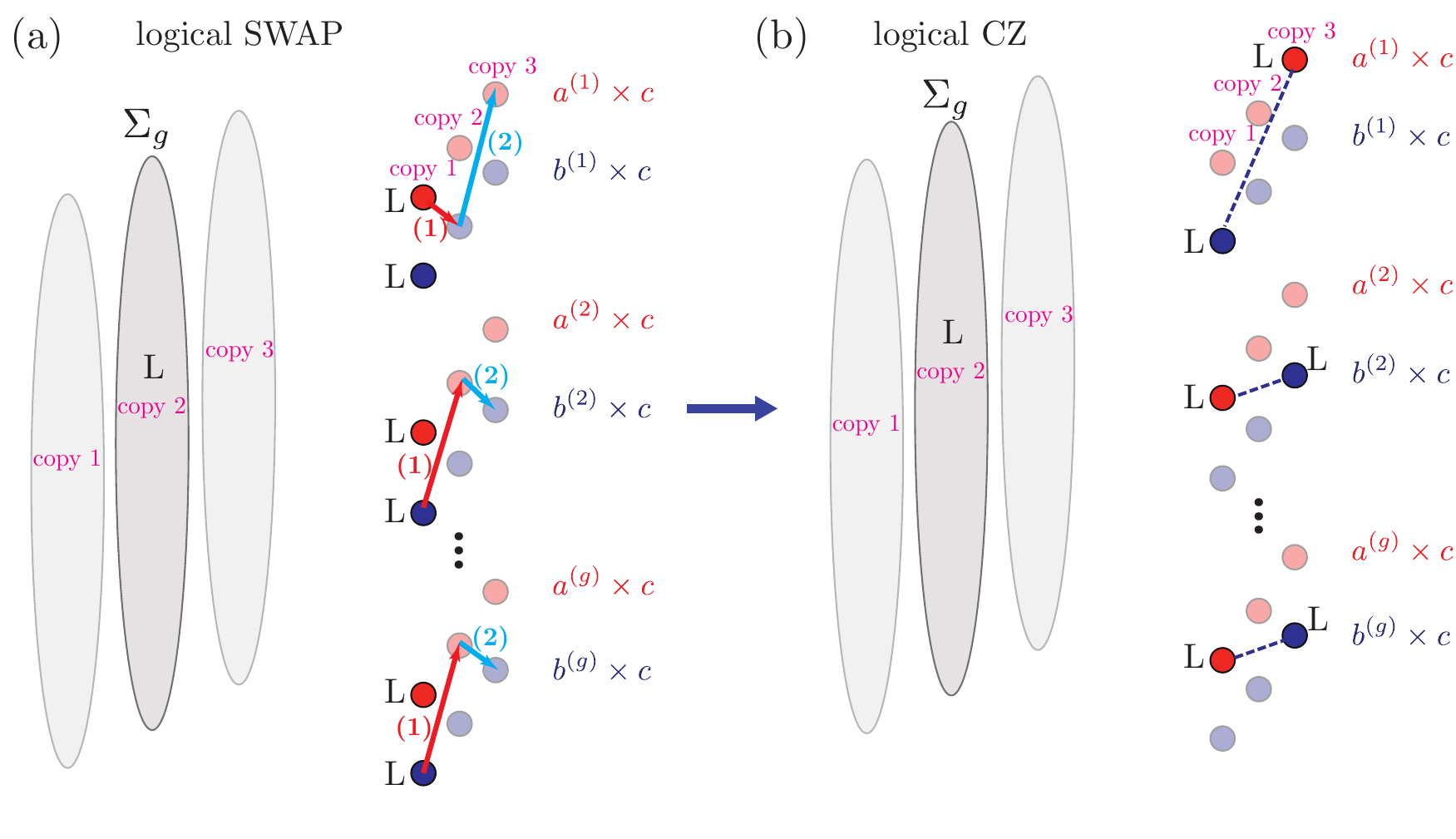}
  \caption{(a) Apply two steps of logical SWAPs to move logical qubits to copy 3 and with the same color. (b) Apply multiple logical CZ gates between logical qubits supported on $a^{(i)}\times c$ (copy 1) and $b^{(i)}\times c$ (copy 3) in parallel. }
\label{fig:parallelized_CZ}
\end{figure*}

We start with implementing parallelizable logical Clifford gates  in the homological qLDPC  code.

The logical CZ gate structure for the quasi-hyperbolic code (as well as the other two codes up to a change of cycle labels)  has been summarized in Eq.~\eqref{eq:CZ_gate_quasi_hyperbolic_n1} and Eq.~\eqref{eq:CZ_gate_quasi_hyperbolic_n2},  and is also encoded in the interaction graph in Fig.~\ref{fig:interaction_hyper-graph}(b).  However, we see that there are two colored edges  with the same type, meaning that a pair of logical CZ gates have to be applied simultanesouly and cannot be separated.    In order to resolve this issue, we set a constant fraction of the logical qubits as ancilla qubits in the logical zero state $\ket{\lo{0}}$ in order to let each logical CZ gate to be individually addressable.  

As shown in Fig.~\ref{fig:application_interaction-graph}(a), we choose a specific encoding scheme by encoding the logical information (labeled as $L$) into the logical qubit labeled as $(\Sigma_g; 2)$ (in the toric-code copy 2) while treating logical qubit labeled as $(\Sigma_g; 2)$ and $(\Sigma_g; 3)$ (in the toric-code copy 2 and 3) as ancilla (labeled as $A$).  In the mean time, we encode logical information into logical qubit labeled as $(a^{(i)}\times c; 1)$ and $(b^{(i)}\times c; 1)$ in the toric-code copy 1, while setting the rest logical qubits in toric-code copy 2 and 3 as ancilla.  Note that all the ancilla logical qubits ($A$) are fixed at the logical-0 state $\ket{\lo{0}}$. 

\begin{figure*}[hbt]
  \includegraphics[width=1.6\columnwidth]{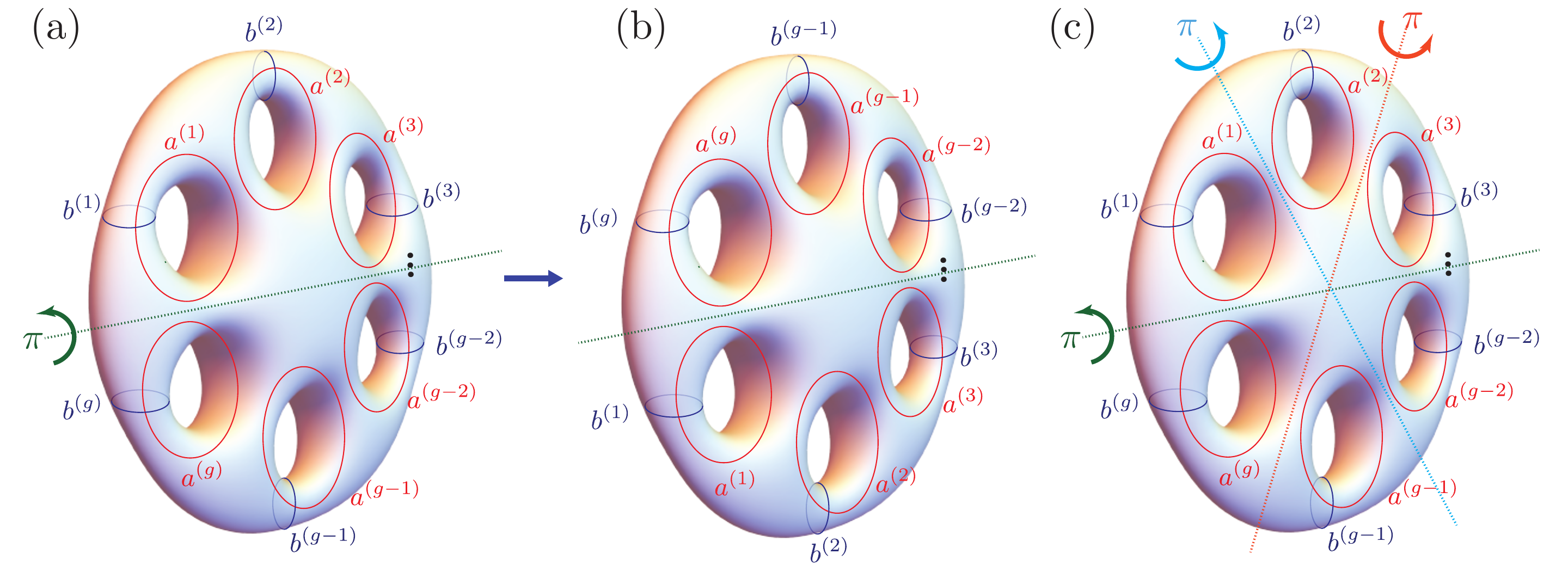}
  \caption{(a) Apply a $\pi$-rotation on each $\Sigma_g$ of the product manifold.  (b) The $\pi$-rotation exchanges homological cycles. (c) One can choose $g$ different axes to implement the $\pi$-rotation. }
\label{fig:pi-rotation}
\end{figure*}

\begin{figure*}[hbt]
  \includegraphics[width=2\columnwidth]{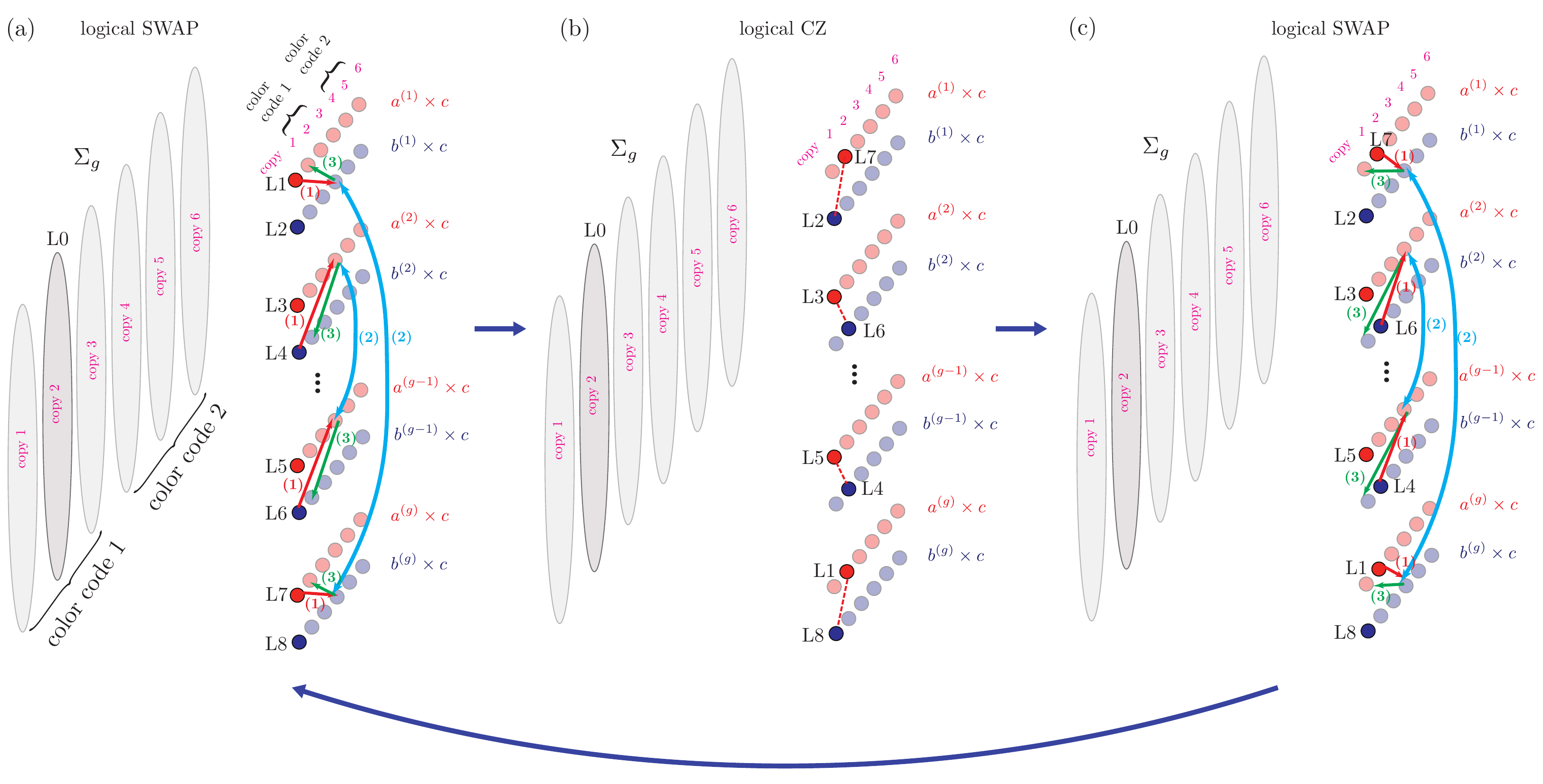}
  \caption{(a) Apply three steps of SWAPs to exchange  logical qubits on $a^{(i)}\times c$ and $a^{(g+1-i)}\times c$ in parallel.  The blue arrows represent the logical SWAP induced by the $\pi$-rotation.   (b) Apply parallel logical qubits between qubits supported on $a^{(i)}\times c$ and $b^{(i)}\times c$.  (c) Apply three steps of SWAPs to shuffle the logical qubits to their original positions.}
\label{fig:SWAP_between_blocks}
\end{figure*}

\begin{figure}[hbt]
  \includegraphics[width=1\columnwidth]{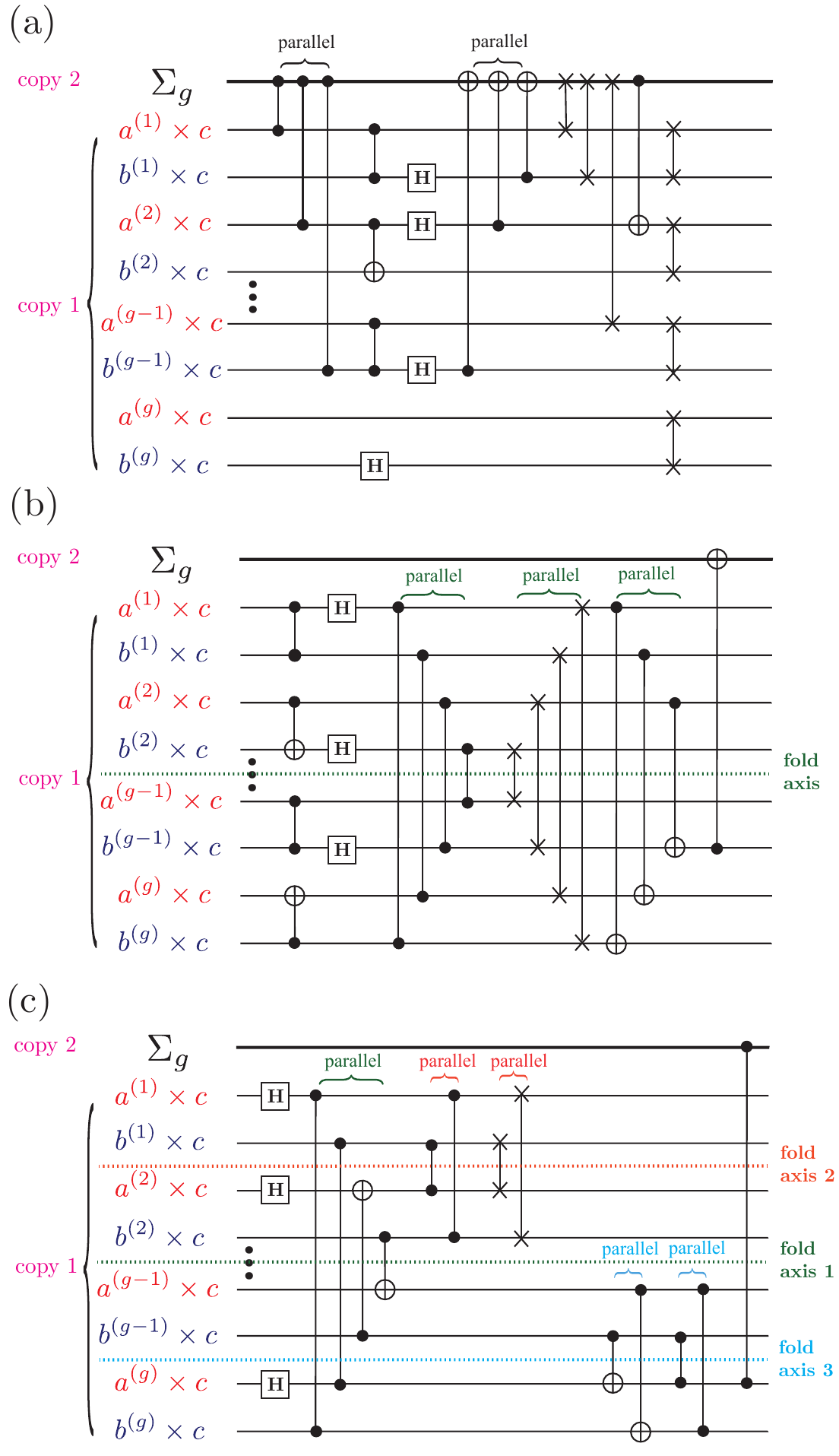}
  \caption{Parallelizability of logical circuits.  (a)  An example logical circuit without the using the isometry of the ancialla copy. Parallel logical CZ, CNOT, and SWAP gates can only be applied between logical qubits with the same $i$ label or have to involve the logical qubits supported on $\Sigma_g$.   (b) The $\pi$-rotation of the ancilla LDPC code copy effectively creates a fold axis in the logical circuit, allowing parallel gates applied between symmetric pairs.  (c) One can introduce $g$ fold axis which greatly enhances the connectivity and parallelizability of the logical circuit.   }
\label{fig:circuit_parallel}
\end{figure}

\begin{figure*}[hbt]
  \includegraphics[width=1.6\columnwidth]{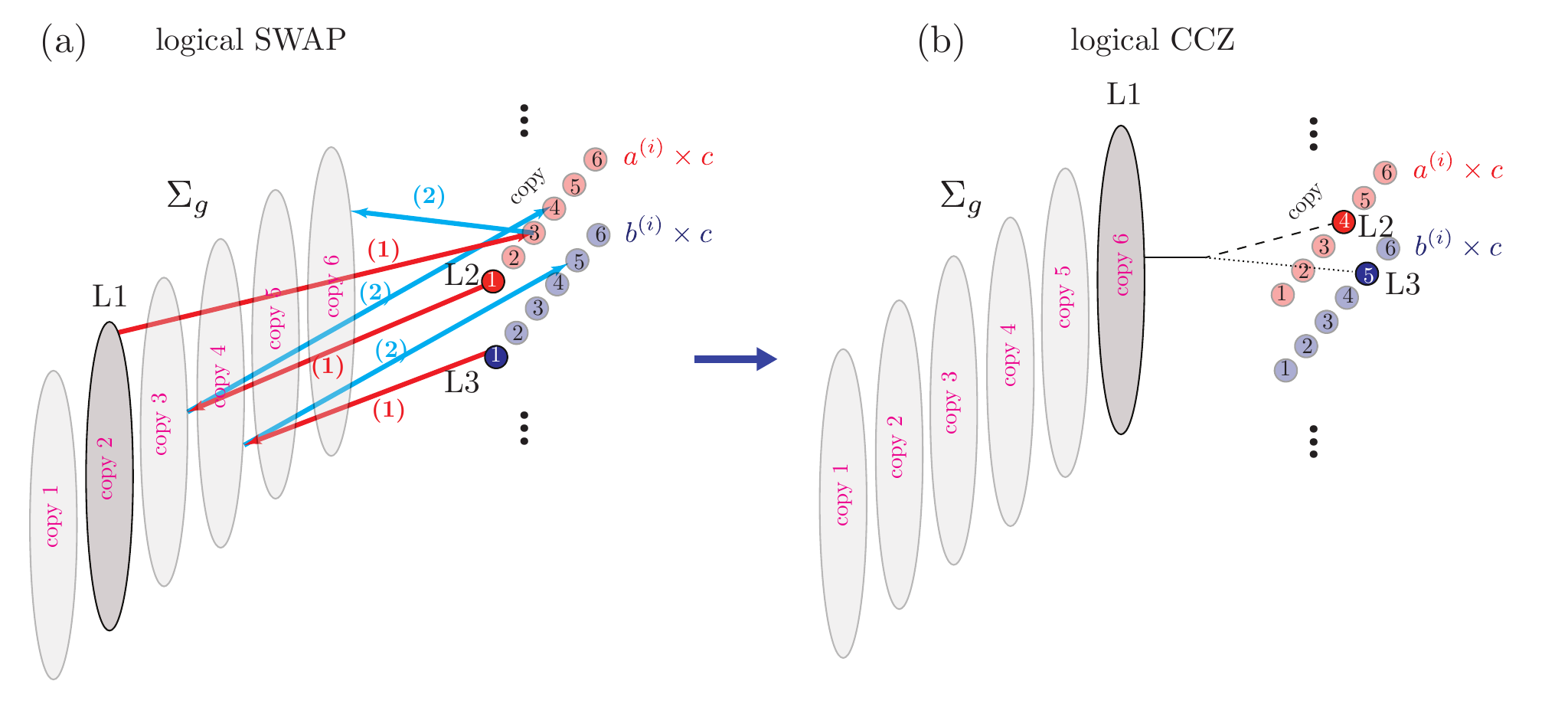}
  \caption{Protocol to implement targeted logical CCZ gate between any selected triplet. (a) The steps of logical SWAP shuffle the selected triplet of logical qubits to the ancilla LDPC color code copy (equivalent to toric-code copy 4,5 and 6).  (b) A logical CCZ gate is applied only to three selected logical qubits in toric-code copy 4, 5 and 6 since all the other logical qubits in the ancilla LDPC code copy are set in the $\lo{\ket{0}}$ state.}
\label{fig:CCZ_implement}
\end{figure*}

With this encoding scheme, we have all the logical CZ gate being individually addressable as illustrated in  Fig.~\ref{fig:interaction_hyper-graph}(b).  Note that each colored edge in this interaction graph has unique type, meaning that the logical CZ gate corresponding to each of them can be independently applied.    We show a logical circuit example in Fig.~\ref{fig:circuit_example}.  We can see that the logical CZ gates in this circuit can be independently applied in parallel. 

Since all the higher symmetry operator discussed in this paper are transversal logical gates, they them alone cannot be universal due to the Eastin-Knill theorem \cite{Eastin:2009cj}. Therefore, we have to consider some other operations such as measurements. In particular, we now consider how to generate logical Hadamard gates with the logical CZ gates implemented by 1-form symmetries.   As shown in Fig.~\ref{fig:logical_Hadamard}(a), we can reinitialize a selected set of ancilla logical qubits in toric-code copy 2 in the logical state $\lo{\ket{+}}$ (+1 eigenstate of the logical-X operator). The reinitialization can be done by measuring logical-X membrane operators of these ancillas via a lattice surgery in the $O(1)$ neighborhood of a logical-X membrane \cite{lavasani2018}.  
After this re-initialization, a collection of colored edges (red-dashed lines) representing simultaneous logical CZ gates via transversal gate $\widetilde{\text{CZ}}^{(1,2)}_{\Sigma_g}$ applied on $\Sigma_g$ and between toric-code copy 1 and 2 are turned on.  As shown in Fig.~\ref{fig:logical_Hadamard}(b), one can apply these logical CZ gates in parallel and again measure the logical-X membranes corresponding of the corresponding logical qubits.  After applying a conditional logical-Z correction via applying logical-Z strings, one effectively applies parallel logical Hadamard gates on all the selected logical qubits.

 The parallelizability of such logical-$X$ measurements and hence the logical Hadamard gates are determined the packing properties of the thickened logical-$X$ membranes in the corresponding 3-manifolds, which will be studies in more details in future works.  For the specific case of the quasi-hyperbolic code, the region of the lattice surgery are supported on a thickened torus $T^2 \times [0,1]$.  Its packing properties are determined by the packing of thickened loops $S^1 \times [0,1]$ in each fibre, i.e., the hyperbolic surface $\Sigma_g$. Since the length of each loop is $O(\log(A))$ ($A$ being the area of the surface), we expect to pack at most $O(A/log A) =O(g/\log g) = O(k/ \log(k))$ loops.  Therefore, one obtains an upper bound of the parallelizability of the logical-$X$ meausrements and Hadamard gates as $O(k/ \log(k))$.  A lower bound will be obtained in future works by detailed study of the hyperbolic geometry.

With the logical Hadamard gate, one can also implement logical CNOT.  Furthermore, three logical CNOTs can compose a logical SWAP gate which can now freely shuffle the logical qubits within and between code blocks.  First of all, one can now apply highly parallelizable logical CZ gates and CNOT gates between logical qubits supported on $a^{(i)}\times c$ and $b^{(i)}\times c$ in different toric-code copies. As shown in Fig.~\ref{fig:parallelized_CZ}(a), one can perform two steps of parallel SWAP, to shuffle several logical qubits in copy 1 into copy 2 and then copy 3.  Then in Fig.~\ref{fig:parallelized_CZ}(b), the interaction edges (blue dashed) corresponding to the collective logical CZ gates between $a^{(i)}\times c$ and $b^{(i)}\times c$  are turned on, and one can perform these CZ gates in parallel via applying a transversal gate $\widetilde{\text{CZ}}^{(1,3)}_{\Sigma_g}$ applied on $\Sigma_g$ and between toric-code copy 1 and 3.   Note that one needs to do two steps of SWAP since the collective logical CZ gates only couple red circles ($a^{(i)}\times c$) and blue circles ($a^{(i)}\times c$) in different toric-code copies. One can then apply another two steps of parallel SWAPs to move the logical qubits in copy 3 back to copy 1, i.e., their original position. Therefore, we have effectively applied parallelized logical CZ gates between  neighboring logical qubits supported on $a^{(i)}\times c$ and $b^{(i)}\times c$ in copy 1.  When conjugated these collective logical CZ with parallel logical Hadamard gate discussed above, parallel logical CNOTs can be applied.  

Nevertheless, in order to apply logical CZ or logical CNOT gates between qubits with different $i$ labels, i.e., between $a^{(i)}\times c$ and $b^{(j)}\times c$, or  $a^{(i)}\times c$ and $a^{(j)}\times c$, or $b^{(i)}\times c$ and $b^{(j)}\times c$ ($i \neq j$),  one needs to SWAP one of them with the global quanum bus supported on $\Sigma_g$ in toric-code copy 2 using the parallelizable logical CZ gates between them as previously illustrated in Fig.~\ref{fig:circuit_example}.  However, since there is only $O(1)$ such global quantum bus, the logical SWAP process can only be done sequentially.  

In order to resolve this issue, one can use an additional LDPC color code block (quasi-hyperbolic color code block in this example) to greatly enhance the parallelizability of logical CZ and CNOT gates.  We note the method we present here is quite generic and work for all the LDPC codes described in this paper.    
In order to apply logical CZ across the code block, one can apply the disentangling unitary $V$ in Eq.~\eqref{eq:disentangler} to disentangle each LDPC color code into three copies of LDPC toric codes:
\be
V^{\otimes 2}[CC(\L^{(1)}_c)\otimes CC(\L^{(2)}_c)\otimes \SS]{V^{\otimes 2}}^{\dag} = \Motimes_{i=1}^6 TC(\L_i),
\ee
which leads to 6 copies of toric codes in total.   Now if we want to apply transversal CZ gates between toric-code copy 1 and 4, we can first apply the inverse of the disentangling unitary $V$ to merge e.g., toric-code copies 1, 4 and 5 into a single 3D color code as 
\be
V^\dag \Motimes_{i=1,4,5} TC(\L_i) V =  CC(\L^{(1')}_c). 
\ee
We can then apply transversal S gate in this new 3D color code copy.    Alternatively, we can also directly apply transversal CZ gates between the disentangled toric-code copies. We can first use the combination of logical CZ and Hadamard to implement logical SWAP to shuffle the logical qubit to toric-code copy 4. We can now disentagle color code $CC(\L^{(1')}_c)$ into toric-code copy 1, 4 and 5, and re-merge copy 1, 2, and 3 into the  first original color code $CC(\L^{(1)}_c)$ and re-merge 4, 5 and 6 into the second original color code $CC(\L^{(2)}_c)$. We can then apply some isometry $\tau$ to the second color code copy $CC(\L^{(2)}_c)$ equivalent to tori-code copy 4, 5 and 6,  for example using the isometry belonging to the Fuschian group introduced in Sec.~\ref{sec:fibre_bundle}.  As a simplest but very useful example, consider the case that $g$ in the quasi-hyperbolic code is even, we can apply an involution, i.e., an order-2 isometry, implementing a $\pi$-rotation (denoted by $\tau_\pi$) of the 3-manifold, which can be understood as a $\pi$-rotation along the axis going between $a^{(1)}$ and $a^{(1+g)}$ on each horizontal hyperbolic surface $\Sigma_g$ perpendicular to the vertical cycle $c$, as illustrated in Fig.~\ref{fig:pi-rotation}(a,b). Note that the existence of such involution requires the fibre surface $\Sigma_g$ being very symmetric. One can choose $\Sigma_g$ to be the canonical hyperbolic 4g-gon with the $i^\text{th}$ and $(i+2)^\text{th}$ edge being identified (see Ref.~\cite{Lavasani2019universal}).  We then have the following map to the homology basis:  
\be
\tau_\pi : a^{(i)} \rightarrow  a^{(g+1-i)},  b^{(i)} \rightarrow  b^{(g+1-i)}, c \rightarrow c.
\ee
This $\pi$-rotation can be done with two steps of long-range SWAPs through ancilla qubits, which needs long-range connection between the qubits paired by the $\pi$-rotation.  Alternatively, since the isometry can be considered as a hyperbolic translation, one can just implement such translation by locally SWAP all the data qubits with by one lattice site (going through the same number of ancilla qubits) per time step. At most $O(\log (n))$ time steps is needed to implement any isometry since it is proportional to the injectivity radius of the hyperbolic surface $\Sigma_g$. After applying this isometry, we can use disentangling and re-merging process like before to apply logical CZ and Hadamrd to implement logical SWAP between toric-code copy 4 and 1 to shuffle the logical qubits back to toric-code copy 1.   Through the above SWAP $\rightarrow$ isometry $\rightarrow$ SWAP process, we have effectively shuffled the logical qubits to their symmetric partners connected by the isometry.  For example, in the case of $\pi$-rotation, for arbitrary set of $i$ and $j$ one can perform the following logical SWAP in toric-code copy 1 in parallel:
\be\label{eq:shuffle_qubits}
(a^{(i)}\times c, 1) \rightarrow (a^{(g+1-i)}\times c, 1);  \ (b^{(j)}\times c, 1) \rightarrow (b^{(g+1-j)}\times c, 1).
\ee
We can then apply logical CZ and CNOT gates between logical qubits on $a^{(g+1-i)}\times c$ and $b^{(g+1-i)}\times c$ in toric-code copy 1  for a selected set of $i$ in parallel.   After that, we can apply the inverse of the logical SWAP in Eq.~\eqref{eq:shuffle_qubits} to shuffle the logical qubits back to their original places,  as illustrated in Fig.~\ref{fig:SWAP_between_blocks}.  This effectively applies logical CZ and CNOT gates between logical qubits 
on $a^{(i)}\times c$ and $b^{(g+1-i)}\times c$ in toric-code copy 1 for a selected set of $i$ in parallel.  One can also understand this as folding the circuit diagram along the axis between $i=g/2$ and $i=g/2+1$ and apply logical gates along those logical qubits which meet each other, as illustrated in Fig.~\ref{fig:circuit_parallel}(b). 

In order to further enhance the connectivity of the logical circuit, one can apply in total $g$ such $\pi$-rotation along the axis going between $a^{(p)}$ and $a^{(p+1)  }$ (mod $g$) with $p=1,2, \cdots, g$  on $\Sigma_g$, as illustrated in Fig.~\ref{fig:pi-rotation}(c). This allows one to apply logical CZ and CNOT gates between logical qubits on $a^{(i+p)}\times c$ and $b^{(g+p+1-i)}\times c$ (mod $g$) in toric-code copy 1 for a selected set of $i$ in parallel.  It also allows one to apply logical SWAP between $a^{(i+p)}\times c$ and $a^{(g+p+1-i)}\times c$ (mod $g$), and  between $b^{(i+p)}\times c$ and $b^{(g+p+1-i)}\times c$ (mod $g$).  
We can also understand this as choosing $g$ possible folding axes  on the circuit diagram between any $i$ and $i+1$ and apply logical gates along those logical qubits aligned with each other,  as illustrated in Fig.~\ref{fig:circuit_parallel}(c).   Through these $g$ isometries, we have greatly enhanced the connectivity and parallelizability of the logical circuit.

\subsection{Generic scheme for universal logical gate set}

Here we consider obtaining a universal logical gate set.   The universal gate set we consider is Toffoli + Hadamard, i.e.,  a Toffoli gate on three arbitrarily chosen logical qubits and Hadamard gates on each logical qubit.  This gate set is also equivalent to CCZ + Hadamard, i.e.,  CCZ on three arbitrarily chosen logical qubits and Hadamard on each logcial qubit, since the combination of a CCZ gate and Hadmard gates can generate a Toffoli gate.

Since we have already demonstrated how to generate parallelizable logical Hadamard gates on each logical qubit as shown in Sec.~\ref{sec:parallel_Clifford} and illustrated in Fig.~\ref{fig:logical_Hadamard}, here we only need to show how to obtain a targeted logical CCZ on three logical qubits.  In fact, we can get targeted logical CCZ gates acting on any desired triplet of logical qubits.   

In order to apply logical CCZ only on three target logical qubits, we should avoid directly applying traversal T which leads to a collective logical CCZ gate acting on logical qubits.   We can use a second ancilla LDPC color code copy (as a specific example the quasi-hyperbolic color code) and setting all the logical qubits in the ancilla copy at state $\ket{\lo{0}}$, as illustrated in Fig.~\ref{fig:CCZ_implement}(a). One can then consider the 2 copies of LDPC color codes as 6 copies of LDPC toric codes. We start with three targeted logical qubits labeled as L1, L2 and L3 as shown in Fig.~\ref{fig:CCZ_implement}(a). The logical qubit L1 is supported on $\Sigma_g$ in toric-code copy 2. The logical qubits L2 and L3 are supported on $a^{(i)}\times c$ and $b^{(i)}\times c$ in toric-code copy 1 respectively. We can use the disentangling and re-merging processes introduced in Sec.~\ref{sec:parallel_Clifford} to apply two steps of parallel logical SWAPs (via the combination of logical CZ and logical Hadmard) to shuffle L1 to toric-code copy 3 and then to copy 6, and shuffle L2 (L3) to copy 3 (copy 4) and then to copy 4 (copy 5), as illustrated in Fig.~\ref{fig:CCZ_implement}(a).  

We can then apply transversal T gates on the second LDPC color code, which yields a logical CCZ between logical qubit L1, L2 and L3, supported on $\Sigma_g$, $a^{(i)}\times c$ and $b^{(i)}\times c$ respectively, as shown by the interaction hypergraph in Fig.~\ref{fig:CCZ_implement}(b).   We hence get the CCZ + Hadamard universal gate set.   As pointed out in Sec.~\ref{sec:parallel_Clifford}, we can arbitrarily SWAP the logical qubits around through the logical qubits supported on $\Sigma_g$.  Therefore, we are able to perform targeted logical CCZ gate on any desired triplet of logical qubits with individual addressability.

\subsection{Enhancing parallelizability of quasi-hyperbolic codes with thickened Dehn twists}\label{sec:Dehn_twists}

\begin{figure*}[hbt]
  \includegraphics[width=2\columnwidth]{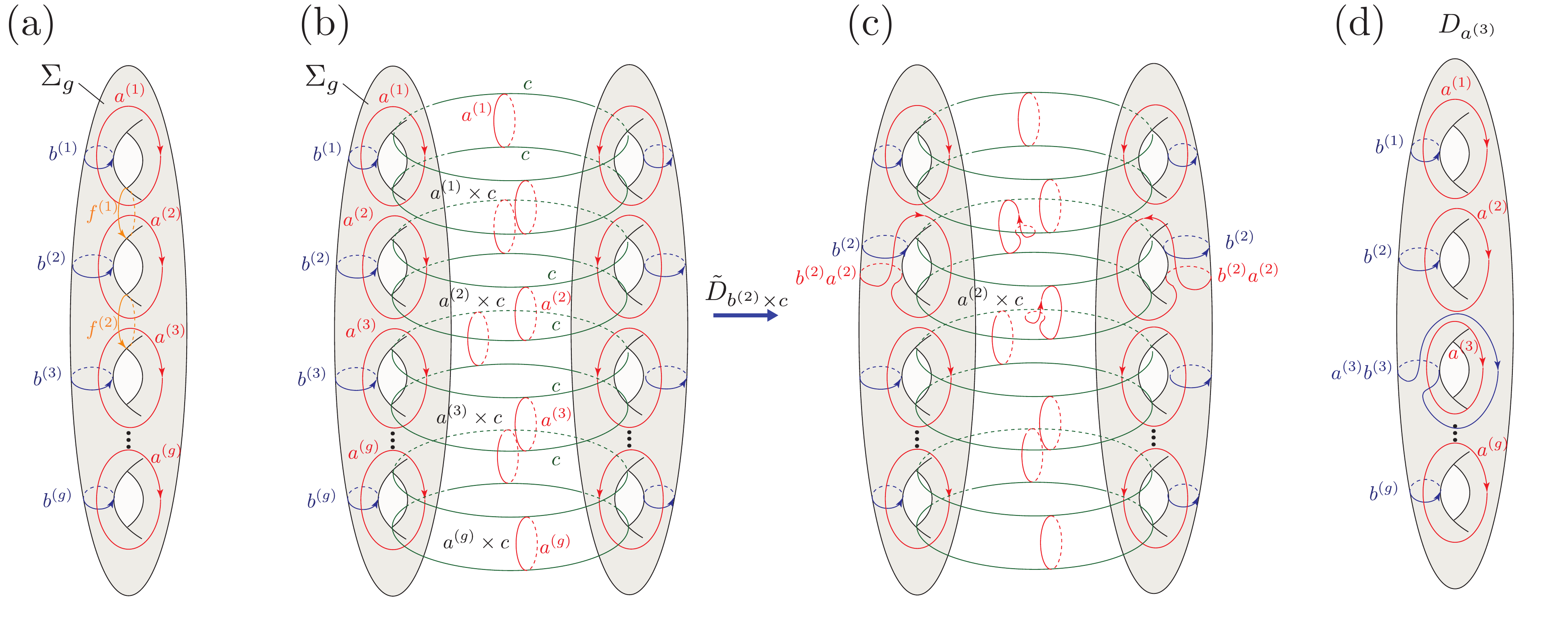}
  \caption{(a) The $3g-1$ Dehn twist generators for the mapping class group on a genus-$g$ surface $\Sigma_g$. (b,c) The action of the thickened Dehn twist along the 2-cycle $b^{(2)}\times c$  on $\Sigma_g \times S^1$. (d) The action of the Dehn twist along $a^{(3)}$  on $\Sigma_g$.}
\label{fig:Dehn_twists}
\end{figure*}

\begin{figure*}[hbt]  \includegraphics[width=1.4\columnwidth]{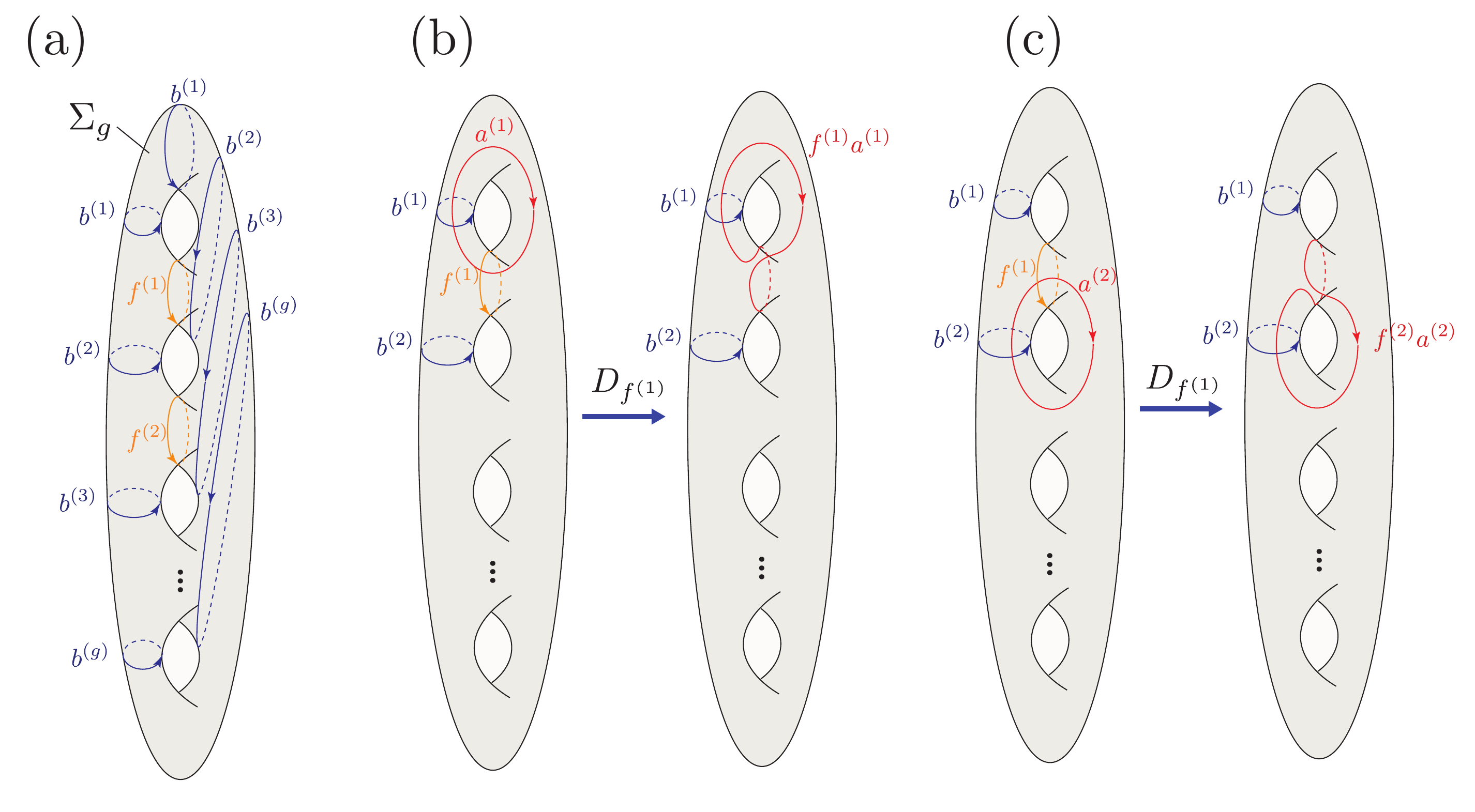}
  \caption{(a) Illustration of the telation between the cycles: $f^{(i)} = (b^{(i)})^{-1}b^{(i+1)}$. (b,c) The action of the Dehn twist along $f^{(1)}$.  }
\label{fig:Dehn_twists_floop}
\end{figure*}

\begin{figure*}[hbt]
\includegraphics[width=1.2\columnwidth]{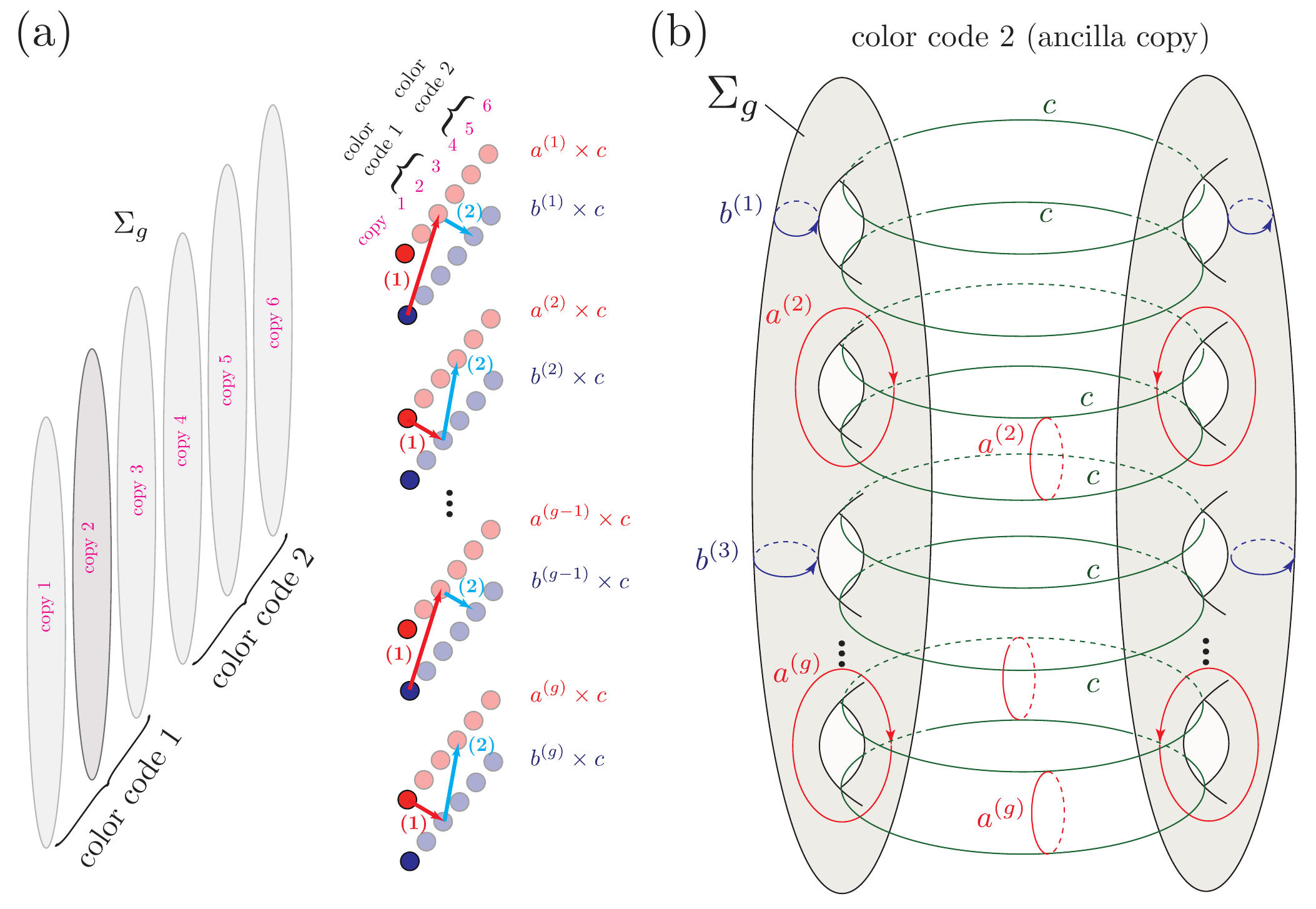}
  \caption{Strategy to perform logical CNOTs between neighboring tori with label $i$ and $i+1$ respectively using thickened Dehn twists.  (a) One first perform two steps of logical SWAPs to shuffle the targeted logical qubits $(a^{(i)}\times c; 1)$ and $(b^{(i+1)}\times c; 1)$ (for even $i$) to the ancilla color code (toric-code copy 4). (b) In the ancilla color code copy, the logical qubits are now encoded sparsely in an alternating pattern, with the corresponding cycles (1-cycles which are extended to 2-cycles by taking a product with the cycle $c$) being highlighted. All the remaining logical qubits are set in the anicilla state $\ket{\lo{0}}$.}
\label{fig:thickened_Dehn_twist_circuit}
\end{figure*}

In Sec.~\ref{sec:parallel_Clifford}, we have shown a method to enhance parallelizability for generic homological LDPC color codes using an ancilla copy and its isometry.    However, this method is slightly more sophisticated due to the use of additional ancilla and SWAPs to implement isometry, and the use of logical-$X$ measurement to implement logical CNOT and SWAP.  

In this section, we provide a specific but simple and  powerful method for the class of quasi-hyperbolic codes, where we use thickened Dehn twists to implement parallel logical CNOT between all logical qubits in the same toric-code copy.  

The thickened Dehn twist $\tilde{D}_{\alpha_1 \times S^1}$ acting along a 2-cycle $\alpha_1 \times S^1$ of the product 3-manifold $\Sigma_g \times S^1$ can be defined as applying Dehn twist $D_{\alpha_1}$ along the cycle $\alpha_1$ on each hyperbolic surface $\Sigma_g$ perpendicular to the base cycle $S^1$. We note that Dehn twists can be implemented by $O(d)$ rounds of local re-triangulation using Thurston's earthquake \cite{Koenig:2010do} or via a constant-depth circuit consisting of local re-triangulation and long-range SWAPs \cite{Lavasani2019universal}.  Both schemes are fault-tolerant.  The thickened the Dehn twists can be directly implemented by straightforwardly  extending either of these schemes to each layer (fibre) of the product manifold and are hence also fault-tolerant.  

Now we investigate what logical gates these thickened Dehn twists correspond to.  In Fig.~\ref{fig:Dehn_twists}(a), we show the standard $3g-1$ Dehn twists which form a generating set of the  mapping class group on genus-$g$ surface   MCG($\Sigma_g$) \cite{FarbMargalit}. These Dehn twists are acted three types of loops: the first two types are $a^{(i)}$ (red) and $b^{(i)}$ (blue) ($i=1,2,\cdots, g$) introduced before, while the third type are $g-1$ loops connecting neighboring genera, denoted by $f^{(i)}$ (yellow) ($i=1,2,\cdots, g-1$).   Note that when consider the mapping classes, all these loops are elements of the 1st-homotopy group instead of 1st homology group (Abelianization of the 1st-homotopy group), and hence have directions which are indicated by arrows. 

We first consider the thickened Dehn twist $\tilde{D}_{b^{(i)}\times c}$ applied along the 2-cycle $\tilde{D}_{b^{(i)}\times c}$, as illustrated in Fig.~\ref{fig:Dehn_twists}(b,c). The Dehn twist $D_{b^{(i)}}$  along the 1-cycle $b^{(i)}$ induces the following map:
\be
D_{b^{(i)}}: a^{(i)} \rightarrow  b^{(i)} a^{(i)},
\ee
and acts trivially on any other loop along the standard 1st-homology basis of $\Sigma_g$ introduced in Sec.~\ref{sec:quasi-hyperbolic}, i.e., $\{a^{(i)},b^{(i)} |  i=1,2, \dots, g\}.$  As illustrated in Fig.~\ref{fig:Dehn_twists}(c), $b^{(i)} a^{(i)}$ is a twisted loop that goes through both cycle $b^{(i)}$ and $a^{(i)}$.   Now the thickened Dehn twist $\tilde{D}_{b^{(i)}\times c}$ induces the following map:
\begin{align}
\non \tilde{D}_{b^{(i)}\times c}:& a^{(i)} \times c \rightarrow  b^{(i)} a^{(i)} \times c = (b^{(i)} \times c)(a^{(i)} \times c), \\ 
& a^{(i)} \rightarrow  b^{(i)} a^{(i)}  
\end{align}
and acts trivially on any other loops and membranes along the 1st- and 2nd-homology basis of $\Sigma_g \times S^1$ introduced in Sec.~\ref{sec:quasi-hyperbolic}, i.e.,  $B_1=\{a^{(i)},b^{(i)},c | \  i=1,2, \dots, g\}$ and  $B_2=\{b^{(i)}\times c,a^{(i)} \times c, \Sigma_g | \  i=1,2, \dots, g\}$.  As illustrated in Fig.~\ref{fig:Dehn_twists}(c), the thickened Dehn twist does a twist on all the loop $a^{i}$ into $b^{(i)}a^{(i)}$ in the 3-manifold. 
This thickened Dehn twist induces the following map on the logical operators in toric-code copy $1$ (for simplicity with set the logical qubits in the other 2 copies at $\ket{\lo{0}}$):
\begin{align}\label{eq:Dehn_twsit_b}
\non \tilde{D}_{b^{(i)}\times c}:&  \lo{X}_{a^{(i)} \times c}^{(1)} \rightarrow  \lo{X}_{a^{(i)} \times c}^{(1)} \lo{X}_{b^{(i)} \times c}^{(1)} \\
&\lo{Z}_{a^{(i)} }^{(1)} \rightarrow  \lo{Z}_{a^{(i)} }^{(1)} \lo{Z}_{b^{(i)} }^{(1)},
\end{align}
where we have used the superscript to label the toric-code copy. 
Note that the Poincaré dual cycles of $a^{(i)}\times c$ and $b^{(i)}\times c$ are $b^{(i)}$ and $a^{(i)}$ respectively, meaning that  $\lo{X}_{a^{(i)} \times c}^{(1)}$ ($\lo{X}_{b^{(i)} \times c}^{(1)}$) and $\lo{Z}_{b^{(i)} }^{(1)}$ ($\lo{Z}_{a^{(i)} }^{(1)}$) are logical operators of the same logical qubit. In order to clarify the logical gate map, we also introduce the logical operator notation with an argument using the logical qubit label which are related to the notations with the support of the operator as: $\lo{X}(a^{(i)} \times c)\equiv \lo{X}_{a^{(i)} \times c} $, $\lo{Z}(a^{(i)} \times c)\equiv \lo{Z}_{b^{(i)}} $, $\lo{X}(b^{(i)} \times c)\equiv \lo{X}_{b^{(i)} \times c} $, $\lo{Z}(b^{(i)} \times c)\equiv \lo{Z}_{a^{(i)}} $. Note that here we use the logical X-membrane (2-cycle) label to label the qubits.  We can then re-wrtie Eq.\eqref{eq:Dehn_twsit_b} with the alternative notations as:
\begin{align}\label{eq:Dehn_twsit_b_rewrite}
\non \tilde{D}_{b^{(i)}\times c}:&  \lo{X}^{(1)}({a^{(i)} \times c}) \rightarrow \lo{X}^{(1)}({a^{(i)} \times c}) \lo{X}^{(1)}({b^{(i)} \times c}) \\
&\lo{Z}^{(1)}(b^{(i)}\times c) \rightarrow  \lo{Z}^{(1)}(b^{(i)}\times c) \lo{Z}^{(1)}(a^{(i)}\times c).
\end{align}
Therefore, the map in Eq.~\eqref{eq:Dehn_twsit_b_rewrite} corresponds to the following logical CNOT gate:
\be\label{eq:CNOT_Db}
\tilde{D}_{b^{(i)}\times c} = \lo{\text{CNOT}}( (a^{(i)}\times c; 1),( b^{(i)}\times c; 1)).
\ee

Similarly, the Dehn twist ${D}_{a^{(i)}}$  along the 1-cycle $a^{(i)}$ induces the following non-trivial map:
\be
D_{a^{(i)}}: b^{(i)} \rightarrow  a^{(i)} b^{(i)},
\ee
and acts trivially on the loop along other homology basis, as illustrated in Fig.~\ref{fig:Dehn_twists}(d).  The thickened Dehn twist $\tilde{D}_{a^{(i)}\times c}$ hence induces the following non-trivial map on the membranes and loops:
\begin{align}
\non \tilde{D}_{a^{(i)}\times c}:& b^{(i)} \times c \rightarrow  a^{(i)} b^{(i)} \times c = (a^{(i)} \times c)(b^{(i)} \times c), \\ 
& b^{(i)} \rightarrow  a^{(i)} b^{(i)},  
\end{align}
and the following non-trivial map on the logical operators:
\begin{align}\label{eq:Dehn_twsit_a}
\non \tilde{D}_{a^{(i)}\times c}:&  \lo{X}_{b^{(i)} \times c}^{(1)} \rightarrow  \lo{X}_{b^{(i)} \times c}^{(1)} \lo{X}_{a^{(i)} \times c}^{(1)} \\
&\lo{Z}_{b^{(i)} }^{(1)} \rightarrow  \lo{Z}_{b^{(i)} }^{(1)} \lo{Z}_{a^{(i)} }^{(1)}.
\end{align}
We can again rewrite the above expression using the logical qubit label as:
\begin{align}\label{eq:Dehn_twsit_a_rewrite}
\non \tilde{D}_{a^{(i)}\times c}:&  \lo{X}^{(1)}({b^{(i)} \times c}) \rightarrow \lo{X}^{(1)}({b^{(i)} \times c}) \lo{X}^{(1)}({a^{(i)} \times c}) \\
&\lo{Z}^{(1)}(a^{(i)}\times c) \rightarrow  \lo{Z}^{(1)}(a^{(i)}\times c) \lo{Z}^{(1)}(b^{(i)}\times c).
\end{align}
Therefore we know the thickened Dehn twist corresponds to the following logical CNOT:
\be
\tilde{D}_{a^{(i)}\times c} = \lo{\text{CNOT}}( (b^{(i)}\times c; 1),( a^{(i)}\times c; 1)).
\ee

Now we investigate what logical gates the thickened Dehn twists $\tilde{D}_{f^{(i)}\times c}$ correspond to.   We first see that all the $g-1$ loops $f^{(i)}$ can be expressed by loops $b^{(j)}$ as
\be
f^{(i)} = (b^{(i)})^{-1}b^{(i+1)}, \qquad i=1,2,\cdots, g-1,
\ee
since loop $f^{(i)}$ is essentially the subtraction of loop $b^{(i)}$ from loop $b^{(i+1)}$ as illustrated in Fig.~\ref{fig:Dehn_twists_floop}(a).   Now we consider the Dehn twist $D_{f^{(i)}}$ along the loop $f^{(i)}$, which implements the following non-trivial map on the neighboring loop $a^{(i)}$ and $a^{(i+1)}$:
\begin{align}
\non D_{f^{(i)}}:& a^{(i)} \rightarrow  f^{(i)} a^{(i)}, \\
            & a^{(i+1)} \rightarrow  f^{(i)} a^{(i+1)},
\end{align}
as illustrated in Fig.~\ref{fig:Dehn_twists_floop}(b) and (c)  respectively.
It twists the loop $a^{(i)}$ to go through the $f^{(i)}$ cycle and hence becomes $f^{(i)} a^{(i)}$,  as illustrated in Fig.~\ref{fig:Dehn_twists_floop}(b).  The thickened Dehn twists hence implements the following map on membrane $a^{i} \times c$ and loop $a^{i}$:
\begin{align}
\non \tilde{D}_{f^{(i)}\times c}:& a^{(i)} \times c \rightarrow  f^{(i)} a^{(i)} \times c = (b^{(i)})^{-1}b^{(i+1)}a^{(i)} \times c, \\ 
\non & a^{(i)} \rightarrow  f^{(i)} a^{(i)}=(b^{(i)})^{-1}b^{(i+1)}a^{(i)}, \\
\non & a^{(i+1)} \times c \rightarrow  f^{(i)} a^{(i+1)} \times c  \\
\non &\qquad \qquad = (b^{(i)})^{-1}b^{(i+1)}a^{(i+1)} \times c, \\ 
 & a^{(i+1)} \rightarrow  f^{(i)} a^{(i+1)}=(b^{(i)})^{-1}b^{(i+1)}a^{(i+1)},
\end{align}
which leads to the following maps on the corresponding logical operators:
\begin{align}
\non \tilde{D}_{f^{(i)}\times c}:&  \lo{X}_{a^{(i)} \times c}^{(1)} \rightarrow  \lo{X}_{b^{(i)} \times c}^{(1)}  \lo{X}_{b^{(i+1)} \times c}^{(1)} \lo{X}_{a^{(i)} \times c}^{(1)}, \\
\non  &\lo{Z}_{a^{(i)} }^{(1)} \rightarrow  \lo{Z}_{b^{(i)} }^{(1)} \lo{Z}_{b^{(i+1)} }^{(1)} \lo{Z}_{a^{(i)} }^{(1)}, \\
\non  &  \lo{X}_{a^{(i+1)} \times c}^{(1)} \rightarrow  \lo{X}_{b^{(i)} \times c}^{(1)}  \lo{X}_{b^{(i+1)} \times c}^{(1)} \lo{X}_{a^{(i+1)} \times c}^{(1)}, \\
  &\lo{Z}_{a^{(i+1)} }^{(1)} \rightarrow  \lo{Z}_{b^{(i)} }^{(1)} \lo{Z}_{b^{(i+1)} }^{(1)} \lo{Z}_{a^{(i+1)} }^{(1)}.
\end{align}
From the above map, we can infer that the thickened Dehn twist $\tilde{D}_{f^{(i)}\times c}$ corresponds to the following logical gate:
\begin{align}
\non \tilde{D}_{f^{(i)}\times c} =& \lo{\text{CNOT}}( (a^{(i)}\times c; 1),( b^{(i)}\times c; 1)) \\
\non  & \cdot \lo{\text{CNOT}}( (a^{(i+1)}\times c; 1),( b^{(i+1)}\times c; 1)) \\
\non &\cdot \lo{\text{CNOT}}( (a^{(i)}\times c; 1),( b^{(i+1)}\times c; 1)) \\
\non &\cdot \lo{\text{CNOT}}( (a^{(i+1)}\times c; 1),( b^{(i)}\times c; 1)) 
\end{align}
By using Eq.~\eqref{eq:CNOT_Db} and the fact that $\tilde{D}_{b^{(i)}\times c}^2=1$ (when acting on the $\ZZ_2$-homology),  we reach the following sequence of thickened Dehn twists to implement a logical CNOT:
\begin{align}\label{eq:CNOT_between_genus}
\non &\tilde{D}_{b^{(i+1)}\times c }  \tilde{D}_{b^{(i)}\times c} \tilde{D}_{f^{(i)}\times c} \\
\non =& \lo{\text{CNOT}}( (a^{(i)}\times c; 1),( b^{(i+1)}\times c; 1)) \\
& \cdot \lo{\text{CNOT}}( (a^{(i+1)}\times c; 1),( b^{(i)}\times c; 1))
\end{align}
The Dehn twists in the above sequence are applied from right to left, consistent with the usual operators in quantum mechanics. 
Equation \eqref{eq:CNOT_between_genus} represents a pair of logical CNOT between neighboring tori corresponding to label $i$ and $i+1$ in the product 3-manifold $\Sigma_g \times S^1$, which were hard to generate using logical CZ gates and Hadamard without extra isometry of an ancilla LDPC color code copy discussed in Sec.~\ref{sec:parallel_Clifford}.   

In order to obtain a single independent logical CNOT between the neighboring tori, we can choose a sparse encoding, i.e., set logical qubits labeled by $(a^{(i+1)}\times c; 1)$ and $( b^{(i)}\times c; 1)$ at state $\ket{\lo{0}}$ for even $i$ as ancilla logical qubits. In other words, we only store logical information in lobical qubits labeled by  $(a^{(i)}\times c; 1)$ and $(b^{(i+1)}\times c; 1) $ for even $i$, as illustrated in Fig.~\ref{fig:thickened_Dehn_twist_circuit}(b).  This gives rise to 
\be
\tilde{D}_{b^{(i+1)}\times c }  \tilde{D}_{b^{(i)}\times c} \tilde{D}_{f^{(i)}\times c} = \lo{\text{CNOT}}( (a^{(i)}\times c; 1),( b^{(i+1)}\times c; 1))
\ee
for even $i$.

Now we can use the following strategy to apply parallel logical CNOTs:  we can apply logical CNOTs between the two types of tori with the same $i$ label, i.e.,  $\lo{\text{CNOT}}( (a^{(i)}\times c; 1),( b^{(i)}\times c; 1))$ and $\lo{\text{CNOT}}( (b^{(i)}\times c; 1),( a^{(i)}\times c; 1))$ via $\tilde{D}_{b^{(i)}\times c } $ and $\tilde{D}_{a^{(i)}\times c }$  respectively.  Now if we want to apply parallel logical CNOTs between neighboring tori with label $i$ and $i+1$ respectively, we can split them into two groups and perform them in two separate rounds.   In order to perform parallel $\lo{\text{CNOT}}( (a^{(i)}\times c; 1),( b^{(i+1)}\times c; 1))$ (for even $i$), we can SWAP all the targeted logical qubits $(a^{(i)}\times c; 1)$ and $(b^{(i+1)}\times c; 1)$ (for even $i$)  into an ancilla LDPC code copy at the same homological cycles with all the other logical qubits in the ancilla code copy set at the $\ket{\lo{0}}$ state and then perform the sequence of Dehn twists $\tilde{D}_{b^{(i+1)}\times c }  \tilde{D}_{b^{(i)}\times c} \tilde{D}_{f^{(i)}\times c}$ on the ancilla LDPC code. as illustrated in as illustrated in Fig.~\ref{fig:thickened_Dehn_twist_circuit}(a,b).  Afterwards, we SWAP these logical qubits back to the original code copy.  Next, in order to perform parallel $\lo{\text{CNOT}}( (a^{(i+1)}\times c; 1),( b^{(i)}\times c; 1))$ (for even $i$), we can SWAP all the targeted logical qubits $(a^{(i+1)}\times c; 1)$ and $(b^{(i)}\times c; 1)$  into the ancilla LDPC code copy and then perform the sequence of Dehn twists $\tilde{D}_{b^{(i+1)}\times c }  \tilde{D}_{b^{(i)}\times c} \tilde{D}_{f^{(i)}\times c}$ on the ancilla LDPC code.  The protocol is complete after performing SWAP of logical qubits back to the original code copy.

We note that the thickened Dehn twists can be performed in parallel, and with all the Dehn twists discussed above we hence can obtain parallel logical CNOTs and SWAPs between all the neighboring logical qubits in the logical quantum circuit.    We further note that one can also perform thickened Dehn twists along very non-local membrane beyond those related to the (3$g$-1) MCG generator, and in this case logical CNOTs and SWAPs can be performed on logical qubits separated far apart in the logical circuit diagram.

\subsection{Parallelizability and connectivity of the interaction hypergraph}\label{sec:connectivity}

We have seen in the previous sections that for the quasi-hyperbolic code, we are only able to apply the logical CCZ gate sequentially even with the use of extra ancilla LDPC code copy.  Similarly, although we have shown parallel logical Clifford gates, their implementation is sophisticated in the case of the quasi-hyperbolic code.   We point out that this is due to the restrictive structure of the interaction hypergraph and graph of the quasi-hyperbolic code.

As shown in Fig.~\ref{fig:interaction_hyper-graph}(a), the interaction graph of the quasi-hyperbolic code has a star-like structure such that all the hyperedges in the interaction graph  connects to the three logical qubits supported on the $\Sigma_g$ surface.   This means even with the ancilla LDPC code copy, all the logical CCZ gates have to invovle a logical qubit supported on $\Sigma_g$, which significantly limits the parallelizability of the non-Clifford gates.   Similarly, the interaction graph of the logical CZ gate in Fig.~\ref{fig:interaction_hyper-graph}(b) and Fig.~\ref{fig:application_interaction-graph}(b) is also star-like, since except those collective logical CZ gate, all the individually addressable logical CZ gates are connected to the logical qubit on $\Sigma_g$.

We can summarize that there are two generic issues in the interaction hypergraph and interaction graph in this case.  First, there are some vertices with very high degree, and in fact with degree $O(k)$, where $k$ is the total number of logical qubits.  This means that one can only individually address $O(1)$ groups of logical CCZ gates.  Second, there are $O(k)$ low-degree vertices in the interaction graph only connected to $O(1)$ high-degree vertices, this significantly limits the connectivity and parallelizability of the logical circuit.   For example, these $O(k)$ low-degree vertices can only interact with each other through the mediation of the $O(1)$ high-degree vertices.   

Nevertheless, these issues are due to the very special product construction of the quasi-hyperbolic code, which leads to a very special triple intersection structure.  For codes defined on generic 3-manifolds or hyperbolic 3-manifolds, such as the generic mapping-torus construction of 3-manifolds discussed in Sec.~\ref{sec:generic_3-manifold}
, we expect each 2D submanifold (surface) behaves similar to each other, meaning that they have triple intersection with similar number of other submanifolds. This suggest that each vertex in the interaction hypergraph and graph will have similar degree.  Therefore, the special star-like structure in the interaction hypergraph and graph of quasi-hyperbolic code will not show up, since no such special vertex should exist.   Moreover,  we expect there is a large family of generic 3-manifold such that there corresponding interaction hypergraph and graph is sparse, i.e., each vertex only has a constant degree (in dependent of the size of the hypergraph/graph), such as the example shown in Fig.~\ref{fig:base_hypergraph}(b). The sparcity of the interaction hypergraph ensures that constant number of logical CZ gate can be individually addressed. Moreover, with such sparse interaction hypergraph and additional ancilla LDPC code copy, one can perform parallel and addressable logical CCZ gates. We further conjecture that these manifolds can also have Betti number propotional to the volume (the corresponding codes have constant degree) and 1-systole and 2-systole growing with the volume (the corresponding codes have growing distance).

\section{Discussion and outlook}\label{sec:outlook}
In this work, we have shown deep connection between seemingly distant fields.  The essence of individually addressable and parllelizable logical gates in homological qLDPC is deeply rooted in the higher-form symmetries of  the underlying TQFT.   Similarly, the triple intersection structure exists in qLDPC codes, TQFT, as well as the topology and geometry of manifolds. It also determines the complexity of the injected hypergrpah magic state.   Therefore, the deeper understanding of one subject may greatly improve the others.

{
Although we focus on the construction of LDPC color codes supported on 3-manifolds which have transversal $T$ and $S$ gates as the logical gates, one might worry that the color code stabilizers have relatively large weight and the decoding for 3D color code also may be more sophisticated than the 3D toric codes \cite{chamberland2020triangular}.  An alternative approach will be still using 3 copies of homological LDPC codes supported on the same 3-manifolds and the corresponding constant-depth circuit corresponding to cup products.   In this way, the homological LDPC code will still have relatively low stabilizer weight corrresponding to the triangulations of the 3-manifold and one can use the usual 3D toric-code decoder or the adapted version for hyperbolic spaces.  
}

As for other future directions, although we have found homological qLDPC codes on 3-manifold with constant rate, we need to further bound its distance.  It is known that even to obtain the  lower bound on the systole of 2D hyperbolic geometry is generically hard outside the realm of arithmetic manifold.  For hyperbolic codes with regular tiling, numerical studies give concrete estimate how the systole/distance scales with the system size \cite{Breuckmann:2017hy}. Therefore, it is possible that systematic numerical study of the systole of such family of 3-manifolds will shed light on the understanding of its scaling  properties and verify our conjecture.   

For the families of almost-constant rate qLDPC codes we found in this paper, such as the quasi-hyperbolic code, the $O(\log(n))$ overhead may not significantly reduce its value in practice.  Although we have focused on directly implementing non-Clifford gates in this paper, these types of codes may be especially useful for efficient magic state injection.   For example,  if one is able to find an almost-constant rate qLDPC codes with sparse interaction hypergraph which are hence able to perform addressable logical CCZ gates in parallel, then it is possible to produce $O(k \sim n / \log n)$ magic states in constant time, while in this case the $O(\log n)$ factor reduction has negligible effect comparing to the constant-rate case even though the distance is only $d=O(\log(n))$.  The low-distance feature makes these codes suitable for fast and efficient injection of intermediate-fidelity magic state to greatly improve the magic-state distillation protocol.          

{ Finally, we point out that although this paper focuses on homological codes defined on manifolds, the cup-product operation and triple intersection structure,  as well as the higher-form symmetries can also be extended to expander-based quantum codes, as shown in the follow-up papers \cite{zhu2025topological, Hsin2024:classifying}.  Therefore, the methods developed in this paper are expected to have a broad application to fault-tolerant quantum computing with qLDPC codes.   
}

\vspace{0.2in}

\noindent{\it Acknowledgements} --- We thank Sergey Bravyi, Jay Gambetta, Tomas Jochym-O'Connor and Ted Yoder for valuable discussion on this work. G.Z. thanks Po-Shen Hsin and Ryohei Kobayashi for the discussion and collaboration on the operator-valued cochain formalism and cup products. G.Z. appreciate Maissam Barkeshli and Ali Lavasani for former collaboration on 2D hyperbolic codes. A.C. and G.Z. are supported by the U.S. Department of Energy, Office of Science, National Quantum Information Science Research Centers, Co-design Center for Quantum Advantage (C2QA) under contract number DE-SC0012704.

\begin{appendix}
\section{Back-engineering of 3-manifolds from the interaction hypergraph: enhancing parallelizability of logical gates}\label{sec:back-engineering}

Instead of searching for 3-manifolds with desired triple intersection structure, we can also try to back-engineer 3-manifolds from a given  triple intersection structure or equivalently the corresponding base interaction hypergraph (or synonymously the intersection hypergraph). In particular, it is desirable to obtain a sparse interaction hypergraph and graph as pointed out in Sec.~\ref{sec:connectivity}. There exists such a general recipe by Sullivan \cite{Sullivan}. 

Let $V = \ZZ^m$ be a lattice of rank $m$ for some positive integer $m$ and let
\begin{align}
    \mu : \wedge^3V \to \ZZ
\end{align}
be a skew-symmetric three form on $V$. In Ref.~\cite{Sullivan}, Sullivan proves that for any $(V, \mu)$ there exists a three manifold $\M^3$ such that $H_2(\M^3, \ZZ)=V$ and the algebraic triple intersection form on  $H_2(\M^3, \ZZ)=V$;
\begin{align}
    | \cdot \cap \cdot \cap \cdot |_\ZZ : \wedge^3 H_2(\M^3, \ZZ) \to \ZZ
\end{align}
coincides with $\mu$. Note that the  discussion in this section is more general than the earlier sections where we only consider the $\ZZ_2$ algebaric triple intersection. Instead, we also consider the general algebraic triple intersection with any integer value in $\ZZ$. Here in analogy with Eq.~\eqref{eq: algebraicintersection} the algebraic triple intersection evaluated on three classes $\alpha^{(i)}_2 \in H_2(\M^3, \ZZ)$ for $1\leq i \leq 3$ is defined as 
\begin{align}
    | \alpha^{(1)}_2 \cap \alpha^{(2)}_2 \cap \alpha^{(3)}_2 |_\ZZ : = \int_{M^3}  \alpha_{(1)}^1 \cup \alpha_{(2)}^1 \cup \alpha_{(3)}^1 \in \ZZ, 
\end{align}
where $\alpha_{(i)}^1$ is the Poincaré dual in $H^1(\M^3, \ZZ)$. 
The construction of the three manifold $\M^3$ goes as follows: Let $\mathsf{H}_m$ be the handlebody of genus $m$, \,i.e. a solid surface of genus $m$. Explicitly, one takes $m$ copies of the solid cylinder, or a handle, which is $D^2 \times [0, 1]$ and glues the two ends of the handles to the boundary of the solid ball which is $D^3$. See Fig.~\ref{fig: handlebody}. 

\begin{figure}[hbt]
  \includegraphics[width=1\columnwidth]{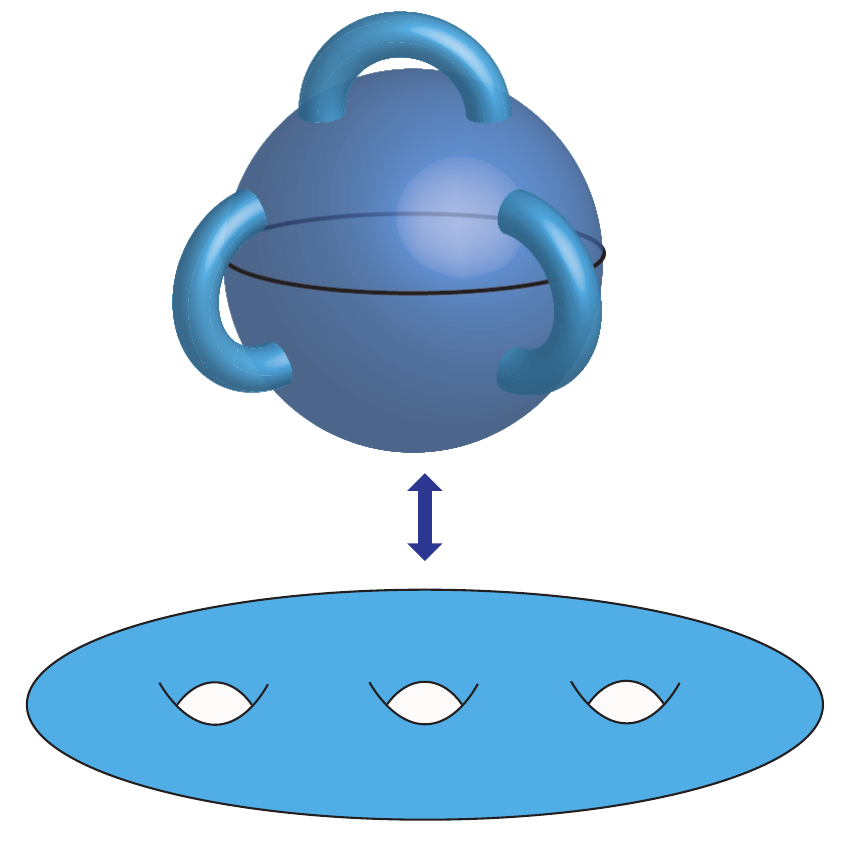}
  \caption{Two equivalent representations of the genus three handlebody.  }
\label{fig: handlebody}
\end{figure}

The boundary of the handlebody $\mathsf{H}_m$ is a surface of genus $m$,  
\begin{align}
    \partial \mathsf{H}_m = \Sigma_m
\end{align}
Given a mapping class $f$ in $\Gamma(\Sigma_m)$, two copies of $\mathsf{H}_m$ can glued along their boundaries using the mapping class
\begin{align}
    f: \partial \mathsf{H}_m \to \partial \mathsf{H}_m
\end{align}
The result is a closed, connected three manifold which will be denoted as $\mathsf{H}(f)$.

Examples: The genus zero handlebody is the three dimensional disk with the boundary  
\begin{align}
    \partial \mathsf{H}_0 = S^2 
\end{align}
Gluing two copies of $\mathsf{H}_0$ along their boundaries using the identity diffeomorphism yeilds the three sphere 
\begin{align}
    \mathsf{H}(Id) = S^3
\end{align}
The genus one handle body is 
\begin{align}
    \mathsf{H}_1 = D^2 \times S^1
\end{align}
Gluing two copies of $\mathsf{H}_1$ along their boundaries using the identity diffeomorphism identifies the boundaries of the two $D^2$ together yeilding a $S^2$ and the two circles are identified to give a circle, the resulting three manifold is  
\begin{align}
    \mathsf{H}(Id) = S^2 \times S^1 
\end{align}
Let $f$ be the mapping class of the two torus represented by the matrix 
\begin{align}
    f= \begin{pmatrix} 0 & -1 \\ 1 & 0 \end{pmatrix} 
\end{align}
Since $f$ maps the standard $\alpha$ and $\beta$ cycles of the two torus to $\beta$ and $\alpha$ respectively, gluing two copies of $\mathsf{H}_1$ along their boundaries using $f$ identifies the boundary of the $D^2$ of the first $\mathsf{H}_1$ with the $S^1$ of the second $\mathsf{H}_1$ and vice versa, this yields the three sphere, 
\begin{align}
    \mathsf{H}(f) = S^3.
\end{align}

The inclusion of the boundary $\partial \mathsf{H}_m=\Sigma_m$ in $\mathsf{H}_m$ induces the following map on the homology, 
\begin{align}
    H_1(\partial \mathsf{H}_m, \ZZ) \to H_1(\mathsf{H}_m, \ZZ)
\end{align}
Let $K(m)$ be the kernel of this inclusion map and let $G(m)$ be the subgroup of the mapping class group $\Gamma(\Sigma_m)$ such that $G(m)$ acts trivially on $K(m)$. The kernel $K(m)$ is generated by those homology classes in $\partial \mathsf{H}_m$ which bound a disk in the three manifold $\mathsf{H}_m$.  In figure \ref{fig: kernel} 
the cycles that bound a disk and lie in $K(m)$ are shown in red. The rank of the kernel $K(m)$ is $m$. If $f$ belongs to $G(m)$ then
\begin{align}
    H_2(\mathsf{H}(f), \ZZ) = K(m) =\ZZ^m
\end{align}
\begin{figure}[t]  \includegraphics[width=1\columnwidth]{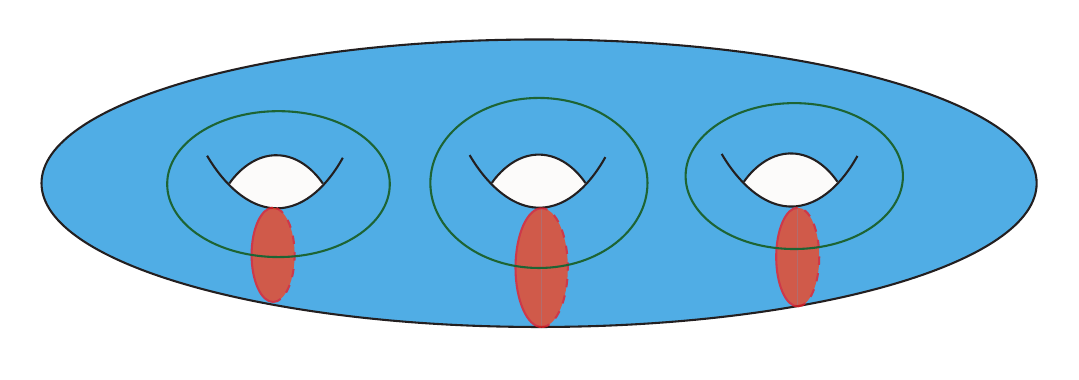}
  \caption{The red loops on the boundary surface bound a disk in $\mathsf{H}_3$ and hence lie in the kernel $K(3)$.  }
\label{fig: kernel}
\end{figure}
We now identify the intersection form on $\mathsf{H}(f)$. Let $m=3$, one calculates that any $f$ that belongs to $G(m)$ induces an action on $H_1(\Sigma_m, \ZZ)$ given by the matrix
\begin{align}
\label{sigma}
 \sigma:=   \begin{pmatrix} 1 &0&0&0&1&1  \\ 0 &1&0&1&0&1\\ 0 &0&1&1&1&0\\* *&*&*&1&0&0\\**&*&*&0&1&0\\* *&*&*&0&0&1\end{pmatrix}
\end{align}
Assuming that $f$ belongs to $G(3)$ we get that $H_2(\mathsf{H}(f), \ZZ)=\ZZ^3$, and in particular,
\begin{align}
    \mathsf{H}(f) = T^3=\RR^3/\ZZ^3 = S^1 \times S^1 \times S^1
\end{align}
Moreover, the intersection number of the three classes in $H_2(T^3, \ZZ)$ is precisely one. In figure \ref{fig: torusheegaarddecomposition} a genus three surface is shown in $T^3$ such that the complement of this surface in $T^3$ is a disjoint union of two $\mathsf{H}_3$.
\begin{figure}[hbt]
 \includegraphics[width=0.7\columnwidth]{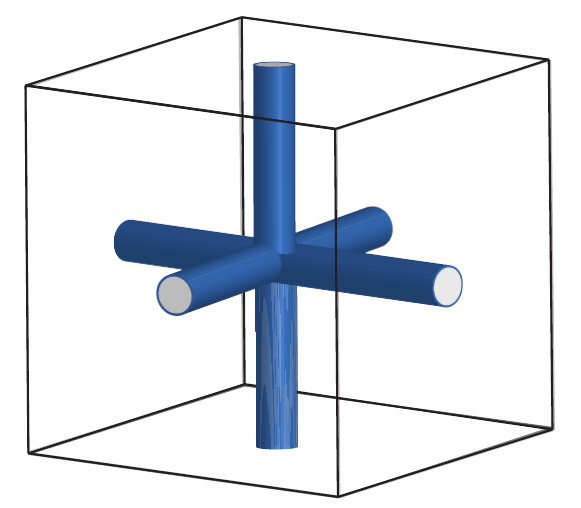}
  \caption{$T^3$ represented as a cube with opposite sides identified. The surface in blue is contained in $T^3$ and has genus three.  }
\label{fig: torusheegaarddecomposition}
\end{figure}

Consider $\mathsf{H}_m$ as $m$ handles attached to the three sphere, with the handles $h_i$, $h_j$, and $h_k$, for $0<i<j<k<m+1$, attached to the northern hemisphere and the rest of the $m-3$ handles attached to the southern hemisphere. Let $\sigma_{i, j, k}$ be the mapping class
\begin{align}
    \sigma_{i, j, k} : \partial \mathsf{H}_m \to \partial \mathsf{H}_m
\end{align}
such that it acts as $\sigma$ from \eqref{sigma} on its restriction to the northern hemisphere and extends by identity to the southern hemisphere. It is clear that $\sigma_{i, j, k}$ belongs to the sub group $G(m)$ hence $H_2(\mathsf{H}(\sigma_{i, j, k}), \ZZ) =\ZZ^m$. Moreover, $\M^3(\sigma_{i,j,k})$ is a connected sum of $T^3$ with $m-3$ copies of $S^2 \times S^1$. In particular, there is only one tuple of three distinct surfaces in $\M^3(\sigma_{i,j,k})$, namely the three two cycles of $T^3$, such that they have non-trivial intersection and the number of their intersections is exactly one.

Let
\begin{align}
    \tau = \Pi_{i<j<k} (\sigma_{i,j,k})^{a_{i,j,k}}.
\end{align}
The mapping class $\tau$ also belongs to the subgroup $G(m)$ as it is a product of elements from the subgroup $G(m)$. Hence, $H_2(\mathsf{H}(\tau), \ZZ)=\ZZ^m$, and there is a natural labelling on the cycles in  $H_2(\mathsf{H}(\tau), \ZZ)=\ZZ^m$ coming from the labelling on the handles of $\mathsf{H}_m$. Let $\Gamma_i$ for $1\leq i \leq m$ denote the cycles in $H_2(\mathsf{H}(\tau), \ZZ)$. In Ref.~\cite{Sullivan} it is proved that
\begin{align}
    |\Gamma_i \cap \Gamma_j \cap \Gamma_k |_{\ZZ}:= \int_{\mathsf{H}(\tau)} \Gamma^i\cup \Gamma^j \cup \Gamma^k = a_{i,j,k}
\end{align}
where $\Gamma^i$ is the Poincaré dual in $H^1(\mathsf{H}(\tau), \ZZ)$ of the cycle $\Gamma_i$. 
   If we let $\overline{\Gamma}_i, \overline{\Gamma}_j$ and $\overline{\Gamma}_k$ be the image of the cycles in $H_2(\mathsf{H}(\tau); \ZZ_2)$ via universal coeffecient theorem, then we get that 
   \begin{align}
       |\overline{\Gamma}_i \cap \overline{\Gamma}_j \cap \overline{\Gamma}_k| = a_{i, j, k}\,\,mod\,2,
   \end{align}
where the left hand side is as defined in Eq.~\eqref{eq: algebraicintersection}.  
We can encode the triple intersection numbers of cycles in $H_2(\mathsf{H}(\tau), \ZZ_2)$ into the base interaction hypergraph introduced in Sec.~\ref{sec:interaction_hypergraph}. 

\begin{figure}[t]
 \includegraphics[width=1\columnwidth]{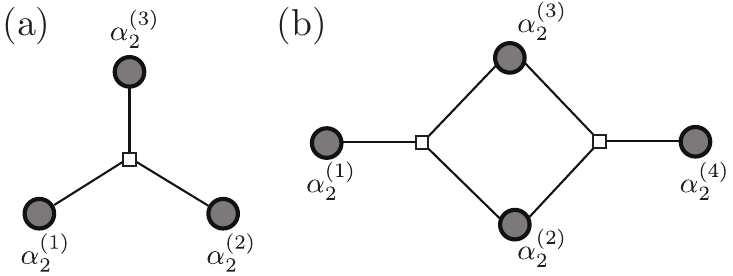}
  \caption{ (a) shows $T^3$, the grey vertices are the three 2-cycles and the tri-junction corresponds to the intersection point of the three cycles. (b) shows a manifold with four two cycles represented by the grey vertices. There are two triplets of two cycles which have algebraic intersection number one, these two triplets can be read off from the tri-junction.  }
\label{fig: intersectiongraph1}
\end{figure}
Figure \ref{fig: intersectiongraph1} (a) corresponds to the three torus, it is composed of a single prong corresponding to the single intersection of the three cycles. In Fig.~\ref{fig: intersectiongraph1} (b) we have four two cycles labelled $\alpha_2^{(i)}$ for $1\leq i \leq 4$ with 
\begin{align}
    |\alpha^{(1)}_2 \cap \alpha^{(2)}_2 \cap \alpha^{(3)}_2| =  |\alpha^{(2)}_2 \cap \alpha^{(3)}_2 \cap \alpha^{(4)}_2| = 1
\end{align}
and all other triplets of 2 cycles have intersection number zero. The manifold $\mathsf{H}(\tau)$ which corresponds to Fig.~\ref{fig: intersectiongraph1} is constructured by gluing two copies to $\mathsf{H}_4$, handlebody of genus four, by using the mapping class 
\begin{align}
    \tau= \sigma_{1, 2, 3}\cdot \sigma_{2, 3, 4}. 
\end{align}
If we impose the condition that for each two cycle $\alpha_2^{(i)}$, there exist exactly two triplets such that 
\begin{align}
    |\alpha^{(i)}_2 \cap \alpha^{(j)}_2 \cap \alpha^{(k)}_2| =  |\alpha^{(i)}_2 \cap \alpha^{(l)}_2 \cap \alpha^{(m)}_2| = 1
\end{align}
then the corresponding hypergraph is a sparse tree. 

\begin{figure}[hbt] \includegraphics[width=1\columnwidth]{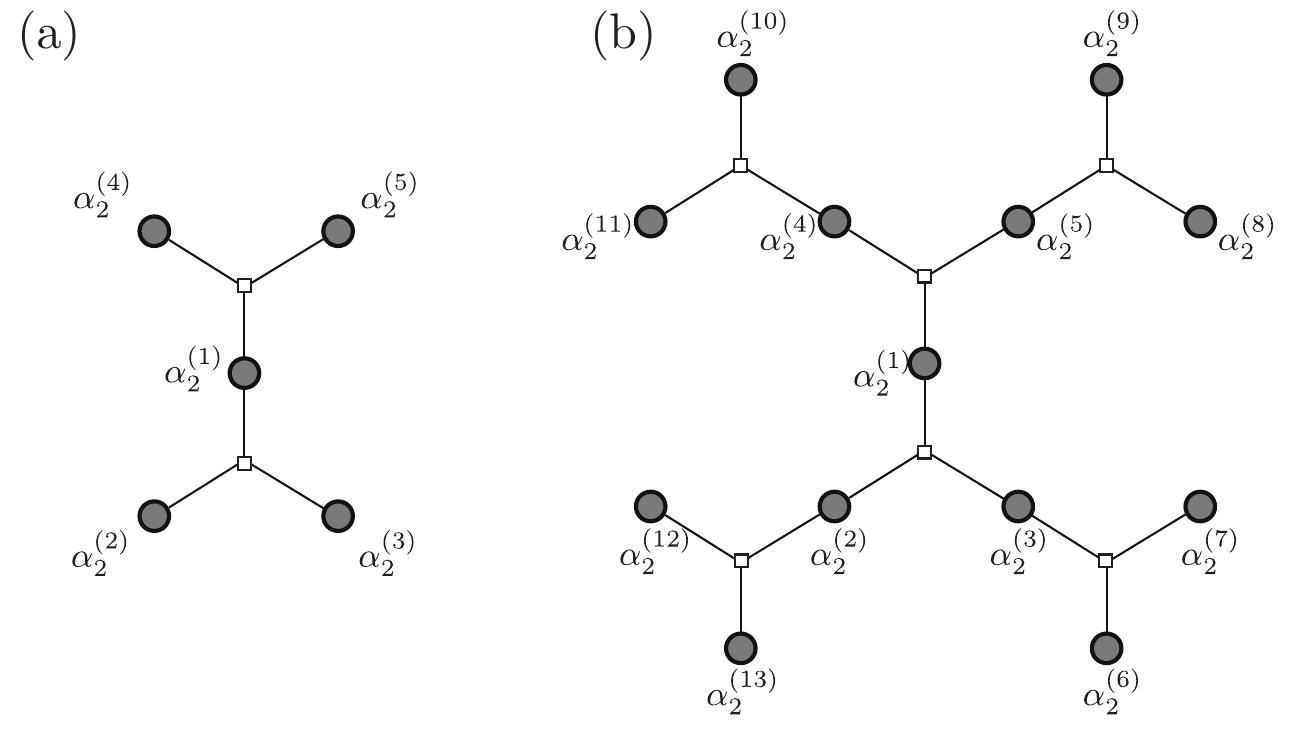}
  \caption{Sparse trees corresponding to triple intersection structure of three manifolds $\mathsf{H}(\tau)$. In (a) the manifold is obtained from gluing two handlebodies of genus five, and in (b) the manifold is obtained from gluing two handlebodies of genus thirteen. The corresponding mapping classes used to glue can found in Eq.~\eqref{eq: mappingclass1} and Eq.~\eqref{eq: mappingclass2}.  }
\label{fig: intersectiongraph2}
\end{figure}

Figure \ref{fig: intersectiongraph2} (a) shows a tree composed of two prongs, it consists of five two cycles $\alpha_2^{(i)}$ for $1 \leq i \leq 5$ with intersection numbers
\begin{align}
    |\alpha^{(1)}_2 \cap \alpha^{(2)}_2 \cap \alpha^{(3)}_2| =  |\alpha^{(1)}_2 \cap \alpha^{(4)}_2 \cap \alpha^{(5)}_2| = 1
\end{align}
The corresponding $\mathsf{H}(\tau)$ is constructed by gluing two copies of $\mathsf{H}_5$, the handle body of genus five, by using the mapping class 
\begin{align}
\label{eq: mappingclass1}
    \tau= \sigma_{1, 2, 3}\cdot \sigma_{1, 4, 5}. 
\end{align}
In Fig.~\ref{fig: intersectiongraph2} (b) a sparse tree corresponding to intersection structure of thirteen two cycles $\alpha_2^i$ for $ 1 \leq i \leq 13$ is shown. It is constructed by gluing two copies of $\mathsf{H}_{13}$, the handlebody of genus thirteen, by using the mapping class 
\begin{align}
\label{eq: mappingclass2}
    \tau= \sigma_{1, 2, 3}\cdot \sigma_{1, 4, 5}\cdot \sigma_{4, 11, 12}\cdot \sigma_{5, 9, 8}\cdot \sigma_{2, 12, 13} \cdot \sigma_{3, 6, 7}.
\end{align}

Nevertheless, while such a recipe gives the right topology which determines the desired logical gate structure, we need to further optimize the geometry such that the Betti number scales linearly or almost linearly with the volume and hence the encoding rate approaches constant, and at the mean time keep the systole  (distance) growing with the volume (number of qubits).  This remains a challenging task and we will further explore it in future works. 

\section{Proof of Theorem \ref{theorem:3-manifold}}\label{sec:proof}
Here, we provide the proof of Theorem \ref{theorem:3-manifold} in Sec.~\ref{sec:fibre_bundle}.
\begin{proof}
The volume of the twisted manifold $\M^3_g$ scales as
\begin{align}\label{eq:volume_vs_g}
\non V\equiv & vol(\M^3_g) \equiv vol(\Sigma_g \times [0,1])/\big((x,0) \sim (\tau x, 1)\big)   \\
=& vol(\Sigma_g \times [0,1]) = area(\Sigma_g) =4\pi (g-1) = O(g),
\end{align}
where the first equality in the second line uses the fact that the twist $\tau$ does not change the volume at all since it only changes the identification/gluing along the cut.  The second equality in the second line is due to the fact that the base circle $S^1$ has unit length, while  the third equality is due to the Gauss-Bonnet theorem Eq.~\eqref{eq:Gauss2}.

The homology $H_1(\M_g^3; \ZZ_2)$ can be computed from the action of $\tau$ on $H_1(\Sigma_g; \ZZ_2)$ which we denote by $\tau_*$. Let 
\begin{align}
    I(\tau_*):= \{v \in H_1(\Sigma_g; \ZZ_2) \, |\, \tau_*(v)=v\}
\end{align}
be the subspace space of homology classes which are invariant under the action of $\tau_*$. Notice that $I(\tau_*)$ corresponds to the eigen-subspace of eigenvalue one.  
As will be discussed later in details, Equation \eqref{invariant} in section \ref{sec:generic_3-manifold}  shows  that
\begin{align}
\label{eq:homology_twisted_manifold}
 H_1(\M_g^3; \ZZ_2) = \ZZ_2^{Rank\big(I(\tau_*)\big) +1 }
\end{align}
for any mapping class $\tau_*$. Since $\tau$ is a finite order element it has its associated finite covering of degree equal to the order $|\tau|$, 
\begin{align}
\label{covering}
 \pi:   \Sigma_g  \to  \Sigma_g/\langle\tau\rangle =: {_gS} 
\end{align}
which is unramified. The theory of covering spaces gives the following short exact sequence of groups; 
\begin{align}
   1\to \pi_1(\Sigma_g) \to \pi_1(\Sigma_g/\langle\tau\rangle ) \xrightarrow{q} \langle\tau\rangle \to 1
\end{align}
If $\alpha\subset \Sigma_g/\langle\tau\rangle$ is a simple closed geodesic such that $q(\alpha)$ is the identity element of the group $\langle \tau \rangle $, then the preimage $\pi^*(\alpha)$ is in $\pi_1(\Sigma_g)$, in particular it is a disjoint union of closed geodesics in $\Sigma_g$ with the number of components equal to the order $|\tau|$ and each component has length equal to the length of $\alpha$. If $q(\alpha)$ is not conjugate to the identity, then one still has that $\pi^*(\alpha)$ is a disjoint union of closed geodesics but the number of components depends on the image $q(\alpha)$ in $\tau$. The number of components is always less than or equal to the order $|\tau|$ and it always divides
the order $|\tau|$. See Ref.~\cite{Sarnak} section 3.10 and proposition 3.11 for details. 

Let $\alpha \subset \Sigma/\langle\tau\rangle $ be a simple closed loop which represents some non-trivial homology class in $H_1(\Sigma_g/\langle\tau\rangle, \ZZ_2)$. The lift $\pi^*(\alpha) \subset \Sigma_g$ is a disjoint union of $N(\alpha)$ loops where $N(\alpha)$ divides the order $|\tau|$;
\begin{align}
    \pi^*(\alpha) =\{\tilde{\alpha}_1, \dots , \tilde{\alpha}_{N(\alpha)} \}  
\end{align}
and there exists an element $\hat{\tau}_*$ in $\langle \tau_* \rangle$ where for $i$ in $\{1, \dots, N(\alpha)\}$ we have 
\begin{align}
    \hat{\tau}_* (\tilde{\alpha}_i) = \tilde{\alpha}_{i+1}\,\,\,\text{and}\,\,\, \hat{\tau}_*(\tilde{\alpha}_{N(\alpha)} ) =\tilde{\alpha}_1
\end{align}
In particular, the linear combination 
\begin{align}
    \tilde{\alpha}_1 + \dots + \tilde{\alpha}_{N(\alpha)}=:v
\end{align}
will be invariant under $\tau_*$. On the other hand, if $v$ in $H_1(\Sigma_g; \ZZ/2)$ is invariant under the action of $\tau_*$ then its push forward $\pi_*(v)$ is a non-trivial homology class in $H_1(\Sigma_g/\langle\tau\rangle; \ZZ_2)$. 
In figure \eqref{fig:covering} we show two loops $\alpha$ and $\beta$ on $\Sigma_g/\langle\tau\rangle$ in colors red and green respectively and their lifts to $\Sigma_g$ assuming that both lifts have five components. We also show a loop $\gamma$ on $\Sigma_g/\langle\tau\rangle$ in blue, and assume its lift has three components. 
In fact, we have a non-canonical isomorphism 
\begin{align}
\label{isomorphism}
I(\tau_*) \cong H_1(\Sigma_g/\langle\tau\rangle; \ZZ_2)
\end{align}
which implies 
\begin{align}
\label{eq:rank}
    Rank\big(I(\tau_*)\big) = 2 genus(\Sigma_g/\langle\tau\rangle )
\end{align}
An easy application of Riemann-Hurwitz formula gives 
\begin{align}
\label{eq: genus}
    \text{genus}(\Sigma_g/\langle\tau\rangle)=\frac{genus(\Sigma_g)-1}{|\tau|} -1 
\end{align}
Recall that $g:=\text{genus}(\Sigma_g)$ and $|\tau|= O(\log^{\frac{1}{2}}(g))$, equations \eqref{eq: genus}, \eqref{eq:rank}, and \eqref{eq:homology_twisted_manifold} now imply 
\begin{align}
\non b_1(\M^3_g) \equiv&    Rank\big(H_1(\M^3_g; \ZZ_2)\big) = O\bigg(\frac{2(g-1)}{\log^{\frac{1}{2}}(g)}-1\bigg) \\
=&   O(g/\log^{\frac{1}{2}}(g)) = O(V/\log^{\frac{1}{2}}(V)),
\end{align}
where the last equality is obtained via Eq.~\ref{eq:volume_vs_g}.

\begin{figure}[hbt]
  \includegraphics[width=1\columnwidth]{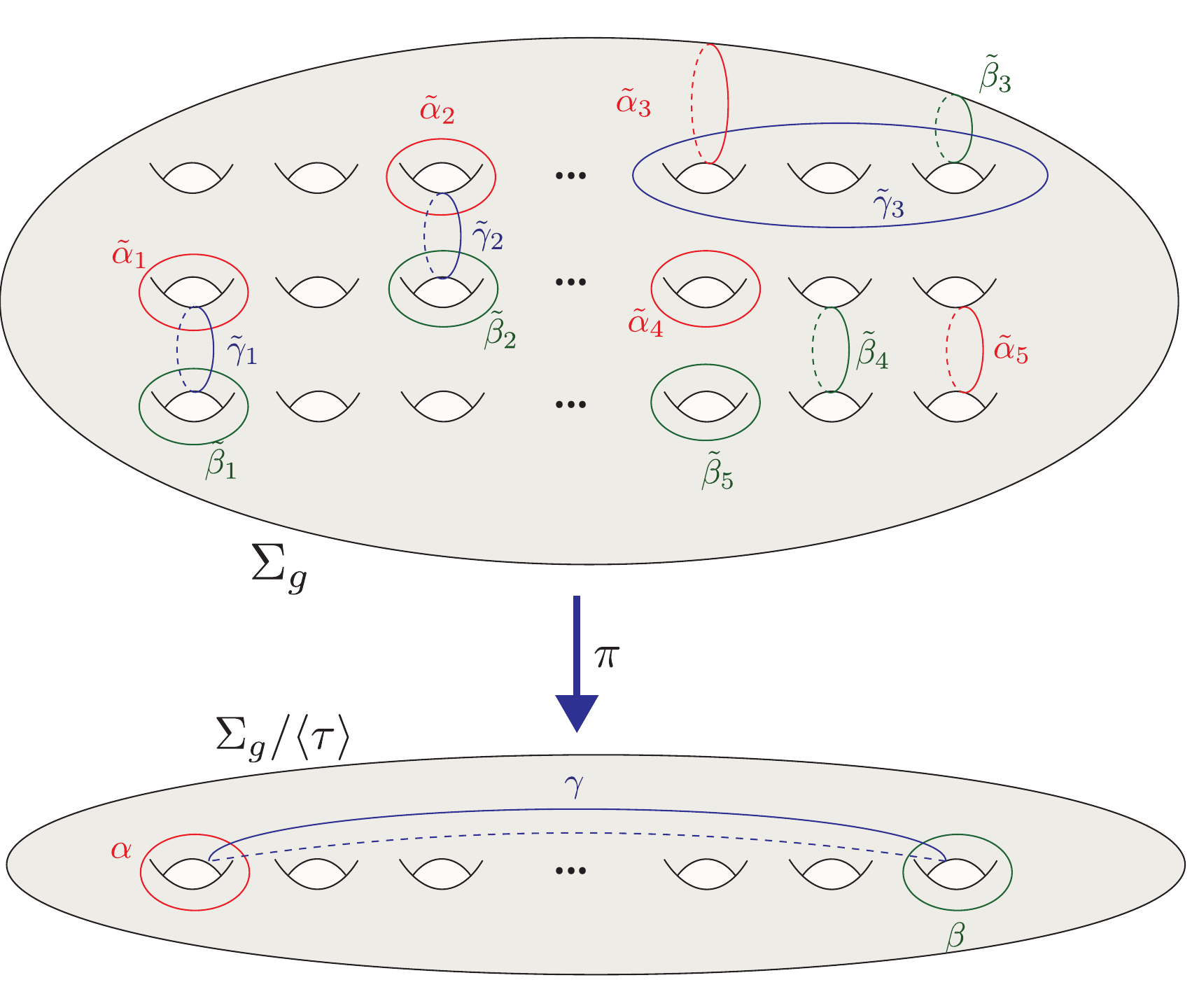}
  \caption{The covering from equation \eqref{covering} showing lift of three loops colored red, blue, and green. }
  \label{fig:covering}
\end{figure}

One can prove the lower bound of the 1-systole $sys_1(\M^3; \ZZ_2)$$=$$\Omega(\log^{\frac{1}{2}}(V))$ by contradiction.   We assume there exists a non-contractible loop $\gamma$ in $\M^3_g$ such that $length(\gamma) < O(\log^{\frac{1}{2}}(V))=O(\log^{\frac{1}{2}}(g))$. Consider lifting $\gamma$ in $\M^3_g$ to a curve $\tilde{\gamma}$ in the covering space $\Sigma_g \times \mathbb{R}$.   The lifted curve $\tilde{\gamma}$ starts at point $(p, t)$ and ends at $(\tau^m p, t+m)$ for some $m \in \ZZ$. Since the bundle projection $\pi: \M^3_g \rightarrow [0,1]/\langle 0=1 \rangle$ is length non-increasing, we have $|m| \le  length(\tilde{\gamma}) = length(\gamma) < O(\log^{\frac{1}{2}}(g)) $.   Since we have $order(\tau)= \Omega(\log^{\frac{1}{2}}(g))$ from Eq.~\eqref{eq:isometry_order_bound}, we obtain $|m| < order(\tau)$, which suggests $p \neq \tau^m p$. Note that under the covering map $\sigma:\Sigma_g \rightarrow \Sigma_g / \langle \tau \rangle =: {_gS}$,  both $p$ and $\tau^m p$ in $\Sigma_g$ are projected to the same point in $_gS$, therefore $p$ and $\tau^m p$ must differ by a non-trivial covering translation of the cover $\sigma$.   However, any non-trivial covering translation must move each point of the total space $\Sigma_g$ by no less than twice the injectivity radius of the base $_gS$ since the geodisc balls of radius smaller than the injectivity radius $R(_gS)$ must be disjoint.  According to Eq.~\eqref{eq:injectivity_radius_gS}, we have $R(_gS) \ge O(\log^\frac{1}{2} (g_k))$. Since the projection $\Sigma_g \times \mathbb{R} \rightarrow \Sigma_g$ is length nonincreasing,  we have
\be
length(\gamma) = length(\tilde{\gamma}) \ge d(p, \tau^m p) \ge O(\log^\frac{1}{2} (g_k)).  
\ee
This is contradictory to our initial assumption that $length(\gamma) < O(\log^{\frac{1}{2}}(g))$
and we conclude the proof of the lower bound of the 1-systole.

In the following, we prove the lower bound of the 2-systole.
Let $z$ represent a non-trivial 2-cycle in $H_2(\M_g; \mathbb{Z}_2)$. By the co-area formula we have

$$Area(z) \ge \int_0^{1/2} |z \cap \Sigma_g \times \{t\}| dt$$

So for some $t_0 \in [0, 1/2]$, we have $|z \cap \Sigma_g \times \{t_0\}| \le 2Area(z)$. 
If $z \cap \Sigma_g \times \{t_0\}$ is non-trivial in $H_1(\Sigma_g; \mathbb{Z}_2)$ then from $sys_1(\Sigma_g) \gtrsim log(g)$ we get that
$Area(z) \gtrsim log(g)$. Otherwise, $z \cap \Sigma_g \times \{t_0\}$ bounds a 2-cycle $F \subset \Sigma_g \times \{t_0\}$. Using Buser's inequality with $\lambda_1(\Sigma_g) \gtrsim 1$, we can find a $F$ with $Area(F) \lesssim Area(z)$. But then the cycle $z + F - F$ can be lifted to $\Sigma_g \times \mathbb{R}$ from which we can see that $Area(z)+2Area(F) \ge Area(\Sigma_g) \gtrsim g$. From the estimate on $Area(F)$ we get that $Area(z) \gtrsim g$. Thus, we see that $sys_2(\M_g) \gtrsim \log(g) =\Omega(\log V)$.
\end{proof}

\end{appendix}

\bibliography{mybib_merge.bib, Ben.bib, andrew.bib, TI, bibliography}

\end{document}